\documentclass[fleqn,10pt]{wlscirep}
\usepackage{cite}

\usepackage{nameref}

\usepackage{microtype}

\usepackage{amsmath,amsthm}
\usepackage{amsfonts,amssymb}
\usepackage{algorithm}
\usepackage{cancel}
\usepackage{colortbl}
\usepackage{booktabs}
\usepackage[noend]{algpseudocode}
\usepackage{rotating}
\usepackage{makecell}

\def\ci{\perp\!\!\!\perp}

\newtheorem{mydef}{Definition}
\newtheorem{mythm}{Theorem}

\newtheorem{mylemma}{Lemma}
\theoremstyle{definition}

\DeclareMathOperator*{\argmin}{argmin} 

\title{Detecting causal associations in large nonlinear time series datasets
}

\author[1,2*]{Jakob Runge}
\author[2,3,4]{Peer Nowack}
\author[5]{Marlene Kretschmer}
\author[4]{Seth Flaxman}
\author[6,7]{Dino Sejdinovic}
\affil[1]{German Aerospace Center, Institute of Data Science, 07745 Jena, Germany}
\affil[2]{Grantham Institute, Imperial College, London SW7 2AZ, United Kingdom}
\affil[3]{Department of Physics, Blackett Laboratory, Imperial College, London SW7 2AZ, United Kingdom}
\affil[4]{Data Science Institute, Imperial College, London SW7 2AZ, United Kingdom}
\affil[5]{Potsdam Institute for Climate Impact Research, 14473 Potsdam, Germany}
\affil[6]{The Alan Turing Institute for Data Science, London NW1 3DB, United Kingdom}
\affil[7]{Department of Statistics, University of Oxford, Oxford OX1 3LB, United Kingdom}

\affil[*]{jakob.runge@dlr.de}

\begin{abstract}
Identifying causal relationships from observational time series data is a key problem in disciplines such as climate science or neuroscience, where experiments are often not possible. Data-driven causal inference is challenging since datasets are often high-dimensional and nonlinear with limited sample sizes. Here we introduce a novel method that flexibly combines linear or nonlinear conditional independence tests with a causal discovery algorithm that allows to reconstruct causal networks from large-scale time series datasets. We validate the method on a well-established climatic teleconnection connecting the tropical Pacific with extra-tropical temperatures and using large-scale synthetic datasets mimicking the typical properties of real data. The experiments demonstrate that our method outperforms alternative techniques in detection power from small to large-scale datasets and opens up entirely new possibilities to discover causal networks from time series across a range of research fields.
\end{abstract}

\begin{document}

\flushbottom
\maketitle
\thispagestyle{empty}

\section*{Introduction}  
\begin{figure*}[t]  
\centering
\includegraphics[width=.7\linewidth]{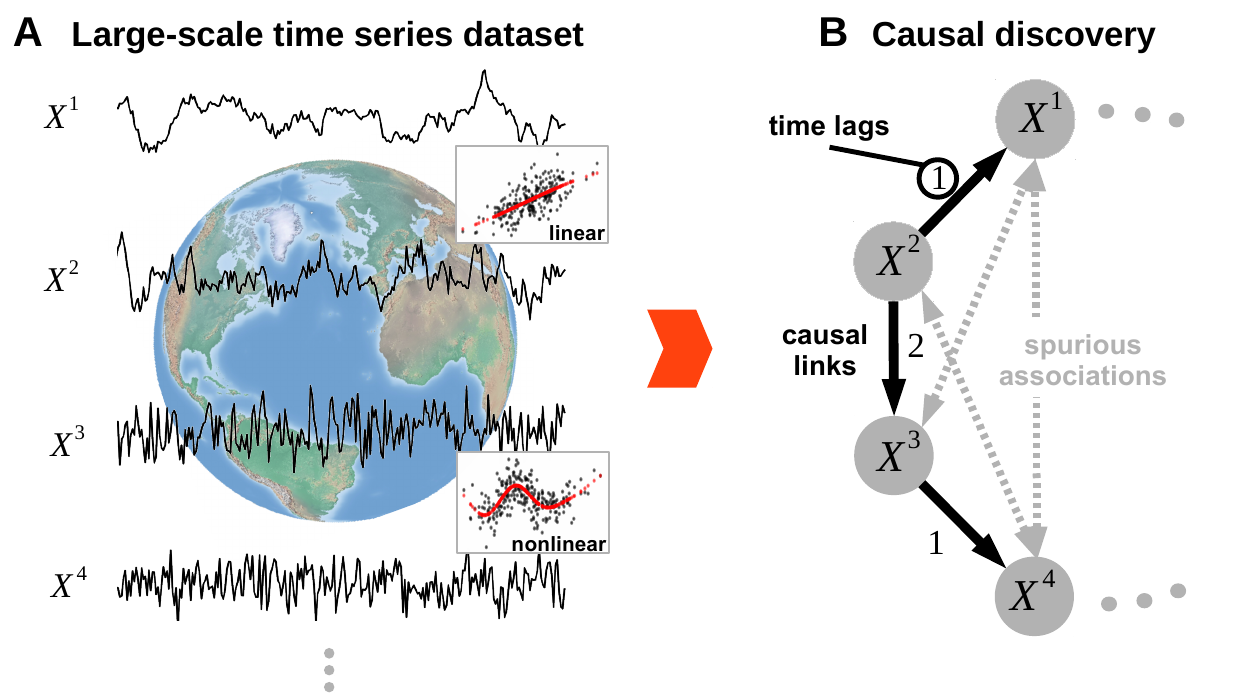}%
\caption{
Causal discovery problem.
Consider a large-scale time series dataset (panel \textbf{A}) from a complex system such as the Earth system of which we try to reconstruct the underlying \emph{causal} dependencies (panel \textbf{B}), accounting for linear and nonlinear dependencies and including their time lags (link labels). Pairwise correlations yield spurious associations due to common drivers (e.g., $X^1 \leftarrow X^2 \rightarrow X^3$) or transitive indirect paths (e.g., $X^2 \rightarrow X^3 \rightarrow X^4$). Causal discovery aims to unveil such spurious associations leading to reconstructed causal networks that are, therefore, much sparser than correlation networks.
}
\label{fig:problem_setting}
\end{figure*}

How do major climate modes such as the El Ni\~no Southern Oscillation (ENSO) or the North Atlantic Oscillation couple to remote regions via global teleconnections? Through which pathways do different brain regions interact? 
Identifying causal association networks of multiple variables is a key challenge in the analysis of such complex dynamical systems, especially since here interventional real experiments, the gold standard of scientific discovery, are often expensive, unethical, or practically impossible. In climate research, model simulations can help to discover causal mechanisms, but confidence in the results is limited due to the realism of the modeled physical processes, which is often not given \cite{Stocker2013}.
Therefore, there is an urgent need to reconstruct causal association networks from observational time series which has become more attractive since recent decades have seen an explosion in the availability of computational resources, cheap data storage, and automated observational data capture of many forms (satellite data, station-based observations, field site measurements \cite{Baldocchi2008}), as well as climate model output \cite{Stocker2013}. 

In a typical scenario in climate science a researcher has an hypothesis on the causal influence between two climatological processes given observed time series data. For example, she may be interested in the influence of ENSO on temperatures over North America. Suppose the time series show a clear correlation, suggesting a relationship between the two processes. In order to exclude other possible hypotheses that may explain such a correlation, she will then include other relevant variables. In the highly interconnected climate system there are many possible drivers she could test, quickly leading to high-dimensional causal discovery problems.

The goal in time series causal discovery from complex dynamical systems is to reliably reconstruct causal links including their time lags since, e.g., climatic teleconnections typically take days to months. The challenge lies in typically high-dimensional and strongly interdependent datasets comprising dozens to hundreds of variables where correlations in some cases arise due to direct causal effects, but also due to a plethora of other reasons, including autocorrelation within each time series, indirect links, and common drivers (Fig.~\ref{fig:problem_setting}). Ideally, a causal discovery method detects as many true causal relationships as possible (high detection power) and controls the number of false positives (incorrect link detections).

A major current approach in Earth data analysis\cite{Mosedale2006,Attanasio2012,Papagiannopoulou2017,McGraw2018}, but also in neuroscience\cite{Bullmore2009,Seth2015}, is to test time-lagged causal associations using autoregressive models in the framework of Granger causality \cite{Granger1969,Barnett2015}. If implemented using standard regression techniques, the high-dimensionality of typical datasets leads to very low detection power (``curse of dimensionality'') since sample sizes are often only on the order of few hundreds (e.g., for a monthly time resolution with 30 years of satellite data). This shortcoming leads to a dilemma that has limited applications of Granger causality mostly to bivariate analyses that cannot, however, account for indirect links and common drivers.

There are methods that can cope with high-dimensionality such as regularized regression techniques \cite{Tibshirani1996,Zou2006}, but mainly in the context of prediction and not causal discovery where assessing the significance of causal links is more important. An exception is Lasso regression\cite{Tibshirani1996} which also allows to discover \emph{active variables}.
Another approach with some recent applications in climate research\cite{Ebert-Uphoff2012b,Runge2014a,Kretschmer2016} are algorithms aimed specifically at causal discovery \cite{Spirtes1991,Pearl2000,Spirtes2000}, which remove redundant or irrelevant variables utilizing iterative independence and conditional independence testing.
However, both regularized regression and recent implementations of causal discovery algorithms do not deal well with the strong interdependencies due to the spatio-temporal nature of the variables as we show here. In particular, controlling false positives at a desired level is difficult for such methods \cite{Bach2008,Lockhart2014,Taylor2015}, and becomes even more challenging for nonlinear estimators. 
In summary, these problems lead to brittle causal network reconstructions and a more reliable methodology is required.

We present a causal discovery method suitable for moderately large time series datasets on the order of tens to hundreds of variables featuring linear as well as nonlinear, time-delayed dependencies given sample sizes of a few hundreds or more. Through analytical results and extensive numerical experiments we demonstrate that the proposed method has advantages over the current state-of-the-art approaches in dealing with high-dimensional interdependent time series datasets yielding reliable false positive control and higher detection power without significantly increasing computational demand. Our approach enables causal analyses among much more variables opening up new possibilities to more credibly reconstruct causal networks from time series in Earth system science, neuroscience, and many other fields.

\section*{Causal discovery}  
\begin{figure*}[t]  
\centering
\includegraphics[width=.7\linewidth]{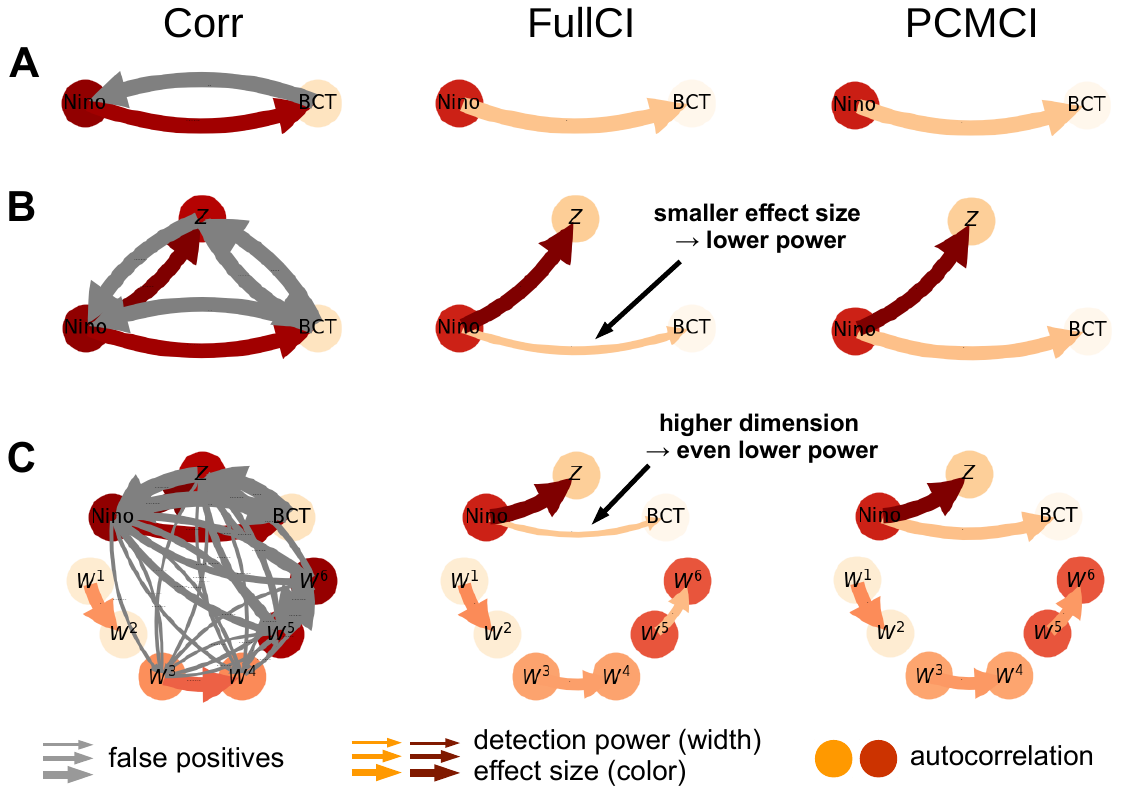}%
\caption{
Motivational climate example. Correlation, FullCI partial correlation, and PCMCI partial correlation between the monthly climate index Nino (3.4 region) \cite{Rayner2003} and land air temperature over British Columbia \cite{Jones2012} (\textbf{A}) for 1979--2017 ($T=468$ months), as well as artificial variables ($Z$ and $W^i$ in \textbf{B},\textbf{C}).
Node colors depict autocorrelation strength, edge colors the partial correlation effect size, and edge widths the detection rate estimated from $500$ realizations of the artificial variables $Z$ and $W^i$ at a significance level of 5\%. The maximum lag is $\tau_{\max}=6$. 
Correlation does not allow for a causal interpretation leading to spurious correlations BCT$\to$Nino (\textbf{A}). FullCI identifies the correct direction, but looses power due to smaller effect size and higher dimensionality if more variables are added (\textbf{B},\textbf{C}). PCMCI avoids conditioning on irrelevant variables leading to larger effect size, lower dimensionality, and, hence, more detection power. See Supplementary Fig.~\ref{fig:example_map} for more information.
%
}
\label{fig:method_motivation}
\end{figure*}
\subsection*{Motivating example from climate science} 
In the following, we illustrate the causal discovery problem on a well-known tropical-extratropical long-range teleconnection. We work out two main factors which lead the common autoregressive modeling approach to have low detection power: reduced effect size due to conditioning on irrelevant variables and high-dimensionality. 

Given a finite time series sample, every causal discovery method has to balance the trade-off between too many false positives (incorrect link detections) and too few true positives (correct link detections). A causality method ideally controls false positives at a pre-defined significance level (typically 5\%) and maximizes detection power.
The power of a method to detect a causal link depends on the available sample size, the significance level, the dimensionality of the problem (e.g., the number of coefficients in an autoregressive model), and effect size, which is the magnitude of the effect \emph{as measured by the test measure} (e.g., the partial correlation coefficient). Since sample size and significance level are usually fixed in the present context, a method's power can only be improved by reducing the dimensionality or increasing the effect size (or both).

Consider a typical causal discovery scenario in climate research (Fig.~\ref{fig:method_motivation}). We wish to test whether the observational data supports the hypothesis that tropical Pacific surface temperatures, as represented by the monthly Nino 3.4 index (see map and region in Supplementary Fig.~\ref{fig:example_map}, further referred to as Nino)  \cite{Rayner2003,Nowack2017}, causally affect extratropical land air temperatures over British Columbia (BCT) 1979--2017 ($T=468$ months). We chose this example since it is well-established and physically understood that atmospheric wave trains induced by increased sea-surface temperatures over the tropical Pacific can affect North American temperatures, but not the other way around\cite{Ropelewski1986,Shabbar1996,McGraw2018}. Thus, the ground truth here is ${\rm Nino}\to{\rm BCT}$ on the (intra-)seasonal time scale allowing us to validate causality methods.

We start with a time-lagged correlation analysis and find that both variables are correlated in both directions, that is, for both positive and negative lags (Fig.~\ref{fig:method_motivation}A, see Supplementary Fig.~\ref{fig:example_map} for lag functions), suggesting also an influence from BCT on Nino. The correlation ${\rm Nino}\to{\rm BCT}$ has an effect size of $\approx 0.3$ ($p<10^{-4}$) at a lag of two months. In the networks in Fig.~\ref{fig:method_motivation} the link colors denote effect sizes (grey links are spurious) and the node colors denote the autocorrelation strength.

Clearly, lagged correlation cannot be used to infer causal directionality, and not even the correct time lag of a coupling \cite{Runge2014a}. Hence, we now move to causal methods.
To test ${\rm Nino}\to{\rm BCT}$, the most straightforward approach then is to fit a linear autoregressive model of BCT on past lags of itself as well as Nino and test whether and which past coefficients of Nino are significantly different from zero. This is equivalent to a lag-specific version of Granger causality, but one can phrase this problem also more generally as testing for conditional independence between ${\rm Nino}_{t-\tau}$ and ${\rm BCT}_t$ conditional on (or controlling for) the common past $\mathbf{X}^-_t=({\rm Nino}_{t-1},\,{\rm BCT}_{t-1},\,\ldots)$, denoted ${\rm Nino}_{t-\tau} \ci {\rm BCT}_{t} ~|~ \mathbf{X}^-_{t}\setminus \{{\rm Nino}_{t-\tau}\}$, up to a maximum time lag $\tau_{\max}$. We call this approach full conditional independence testing (FullCI) and illustrate it in a linear partial correlation implementation for this example, that is, we test $\rho({\rm Nino}_{t-\tau}, {\rm BCT}_t|\mathbf{X}^-_t)\neq 0$, which is the effect size for FullCI, for different lags $\tau$.

Using a maximum time lag $\tau_{\max}=6$ months, we find a significant FullCI partial correlation for ${\rm Nino}\to{\rm BCT}$ at lag $2$ of $0.1$ ($p=0.037$) (Fig.~\ref{fig:method_motivation}A) and no significant association in the other direction. That is, the effect size of FullCI is strongly reduced compared to the correlation when taking into account the past. However, as mentioned before, such a bivariate analysis can usually not be interpreted causally, because other processes might explain the relationship. To further test our hypothesis, we then include another variable $Z$ that may explain the association between Nino and BCT (Fig.~\ref{fig:method_motivation}B). Here we generate $Z$ artificially for illustration purposes and define $Z_t=2 \cdot {\rm Nino}_{t-1}+\eta^Z_t$ for independent standard normal noise $\eta$. Thus, Nino drives $Z$ with lag $1$, but $Z$ has no causal effect on BCT, which we assume \emph{a priori} unknown. Here we simulated different realizations of $Z$ to measure detection power and false positive rates. Now the correlation would be even more misguiding a causal interpretation since we observe spurious links between all variables (Fig.~\ref{fig:method_motivation}B). The FullCI partial correlation, now with $\mathbf{X}^-_t$ including the past of all three processes, is significant only for ${\rm Nino}_{t-2}\to {\rm BCT}_t$, with an effect size of $0.09$. At a 5\% significant level this link is only detected in 53\% of the realizations (true positive rate).

What happened here? As mentioned above, detection power depends on dimensionality and effect size. Additionally conditioning on the variable $Z$ (including its past) slightly increases dimensionality of the conditional independence test, but this only partly explains the low detection power. For example, if $Z$ is constructed in such a way that it is independent of Nino, the partial correlation is $0.1$ again and power 85\%. The more important factor is that, since ${\rm Nino}_{t-1}\to Z_t$, $Z$ contains information about Nino. Because $Z$ is part of the conditioning set $\mathbf{X}^-_t$, it now `explains away' some part of the partial correlation $\rho({\rm Nino}_{t-2}, {\rm BCT}_t|\mathbf{X}^-_t)$, thereby leading to an effect size that is just $0.01$ smaller which already strongly reduces power. 

Suppose we got one of the realizations of $Z$ for which the link ${\rm Nino}_{t-2}\to {\rm BCT}_t$ is still significant. To further illustrate this phenomenon and the effect of including more variables on detection power, we now include six more variables $W^i$, that are, unknown to us, all independent of Nino, BCT, and $Z$, but coupled between each other in the following way (Fig.~\ref{fig:method_motivation}C): $W^{i}_t=a^i W^i_{t-1}+c W^{i-1}_{t-2} + \eta^i_t$ for $i=2,4,6$ and $W^{i}_t=a^i W^{i}_{t-1}+ \eta^{i}_t$ for $i=1,3,5$, all with the same coupling coefficient $c=0.15$ and $a^{1,2}=0.1$, $a^{3,4}=0.5$, and $a^{5,6}=0.9$. Now the FullCI effect size for ${\rm Nino}_{t-2}\to {\rm BCT}_t$ is still $0.09$, but the detection power is even lower than before and decreases to only 40\% due to the higher dimensionality. Thus, the true causal link ${\rm Nino}_{t-2}\to {\rm BCT}_t$ is likely to be overlooked. 

Effect size is also affected by autocorrelation effects of the included variables: The variable pairs $(W^i,W^{i-1})$ differ in their autocorrelation (as visualized by their node color in Fig.~\ref{fig:method_motivation}C) and even though the coupling coefficient $c$ is the same for each pair, their partial correlations are $0.15,\,0.13,\,0.11$ (from lower to higher autocorrelation). Similar to the above case, conditioning on other lags explains away information leading to a smaller effect size and lower power for the variable pairs $(W^i, W^{i-1})$ the higher their autocorrelation is. Conversely, we here observe more spurious correlations for higher autocorrelations (Fig.~\ref{fig:method_motivation}C left panel).

This example, thus, illustrates the dilemma we started with: To strengthen the credibility of causal interpretations we need to include more variables that might explain a spurious relationship, but these lead to lower power to detect true causal links due to higher dimensionality and possibly lower effect size.  

\subsection*{PCMCI approach} 
\begin{figure*}[t]  
\centering
\includegraphics[width=1.\linewidth]{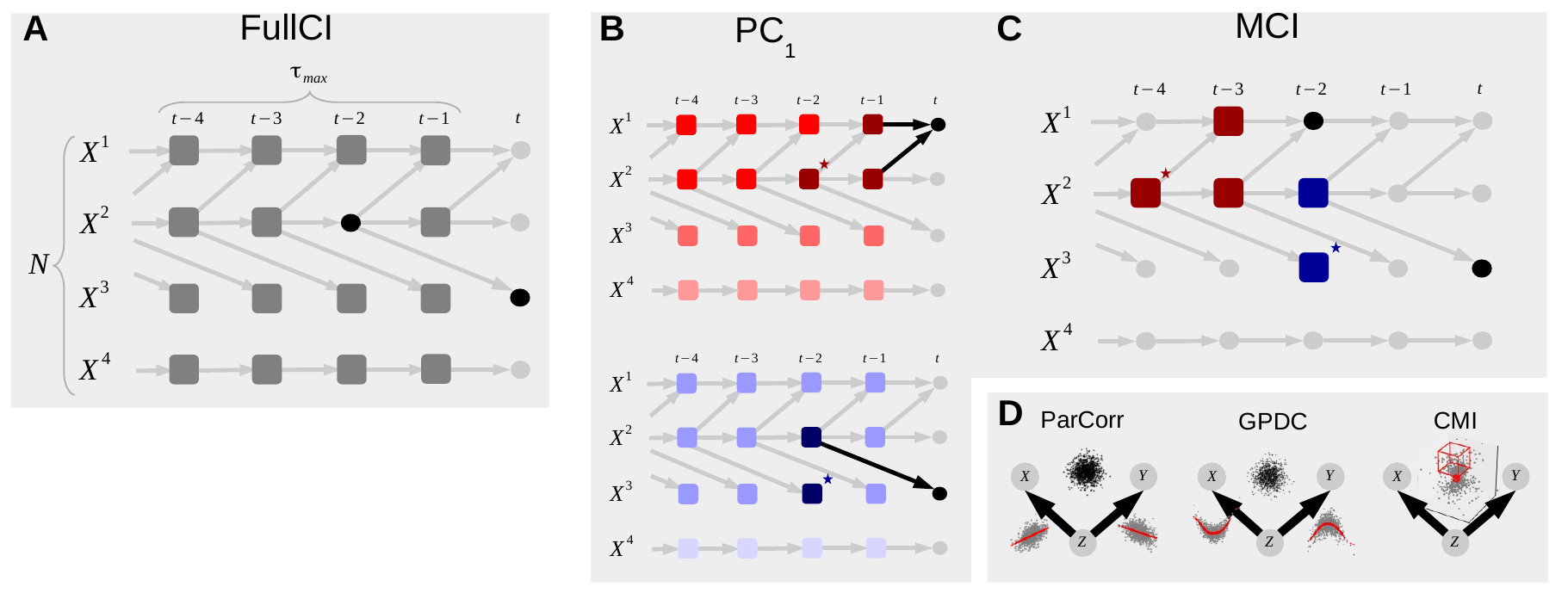}%
\caption{
Proposed causal discovery method.
\textbf{(A)} \textit{Time series graph} \cite{Eichler2011,Runge2012prl} representing the time-lagged causal dependency structure underlying the data. FullCI tests the presence of a causal link by $X^i_{t-\tau} \cancel{\ci} X^j_t ~|~ \mathbf{X}^-_t \setminus \{X^i_{t-\tau}\}$ where $\ci$ denotes (conditional) independence and $\mathbf{X}^-_t \setminus \{X^i_{t-\tau}\}$ the past of all $N$ variables up to a maximum time lag $\tau_{\max}$ excluding $X^i_{t-\tau}$ (grey boxes).  
\textbf{(B)} Illustration of PC$_1$ condition selection algorithm for the variables $X^1$ (top) and $X^3$ (bottom): The algorithm starts by initializing the preliminary parents $\widehat{\mathcal{P}}(X^j_t) = \mathbf{X}^-_t$. In the first iteration, variables without even an unconditional association (e.g., uncorrelated) are removed from $\widehat{\mathcal{P}}(X^j_t)$ (lightest shade of red and blue, respectively). In the second iteration, variables that become independent conditional on the driver in $\widehat{\mathcal{P}}(X^j_t)$ with largest association in the previous step are removed. In the third iteration variables are removed that are independent conditionally on the \emph{two} strongest drivers and so on until there are no more conditions to test in $\widehat{\mathcal{P}}(X^j_t)$. In this way PC$_1$ adaptively converges to typically only few relevant conditions (dark red/blue) that include the causal parents $\mathcal{P}$ with high probability and potentially some false positives (marked with an asterisk).
\textbf{(C)} These low-dimensional conditions are then used in the MCI conditional independence test: For testing $X^1_{t-2} \to X^3_t$ the conditions $\widehat{\mathcal{P}}(X^3_t)$ (blue boxes) are sufficient to establish conditional independence, while the additional conditions on the parents $\widehat{\mathcal{P}}(X^1_{t-2})$ (red boxes) account for autocorrelation and make MCI an estimator of causal strength. 
(\textbf{D})~Both the PC$_1$ and the MCI step can be flexibly combined with linear (ParCorr) or nonlinear (GPDC and CMI) independence tests (see Supplementary Sect.~\ref{sec:ci_tests} and Tab.~\ref{tab:CI_overview}). ParCorr assumes linear additive noise models and GPDC only additivity. The grey scatter plots illustrate regressions of $X,Y$ on $Z$ and the black scatter plots the residuals. The red cubes in CMI illustrate the data-adaptive model-free $k$-nearest neighbor test\cite{Runge2018a} which does not require additivity.
}
\label{fig:mci_schematic}
\end{figure*}
The previous example has shown the need for an automated procedure that better identifies the typically few relevant variables to condition on. We now introduce such a causal discovery method that helps to overcome the above dilemma and more reliably estimates causal networks from time series data. 
While the networks depicted in Fig.~\ref{fig:problem_setting}B, and Fig.~\ref{fig:method_motivation} are easier to visualize, they do not fully represent the spatio-temporal dependency structure underlying complex dynamical systems which is more comprehensively grasped in a \textit{time series graph} \cite{Eichler2011,Runge2012prl} (Fig.~\ref{fig:mci_schematic}). The nodes in a time series graph represent the variables at different lag-times and a causal link $X^i_{t-\tau} \to X^j_t$ exists if $X^i_{t-\tau}$ is \emph{not} conditionally independent of $X^j_t$ given the past of all variables, formally defined by $X^i_{t-\tau} ~\cancel{\ci}~ X^j_t | \mathbf{X}^-_t\setminus \{X^i_{t-\tau}\}$ with $\cancel{\ci}$ denoting the absence of a (conditional) independence, the vertical bar $|$ meaning ``conditional on'', and $\mathbf{X}^-_t \setminus \{X^i_{t-\tau}\}$ denoting the past of all $N$ variables up to a maximum time lag $\tau_{\max}$ excluding $X^i_{t-\tau}$ (grey boxes in Fig.~\ref{fig:mci_schematic}A). FullCI directly tests the link-defining conditional independence.
Recall that in Fig.~\ref{fig:mci_schematic}A the high dimensionality of including $N \tau_{\max}-1$ conditions on the one hand, and the reduced effect size due to conditioning on $X^1_{t-1}$ and $X^2_{t-1}$ (similar to $Z$ in Fig.~\ref{fig:method_motivation}), on the other, leads to a potentially drastically reduced detection power of FullCI.

Causal discovery theory\cite{Pearl2000,Spirtes2000} tells us that the \emph{parents} $\mathcal{P}$ of a variable $X^j_t$ (in Fig.~\ref{fig:mci_schematic}B represented as nodes with black arrows) are a sufficient conditioning set that allows to establish conditional independence (\emph{Causal Markov property}\cite{Spirtes2000}). Thus, in contrast to conditioning on the whole past of all processes as in FullCI, conditioning only on the set of parents of a variabel $X^j_t$ suffices to identify spurious links. Markov discovery algorithms\cite{Spirtes2000,Aliferis2010} such as the PC algorithm (named after its inventors) \cite{Spirtes1991} allow to detect these parents and can be flexibly implemented with different kinds of conditional independence tests which can handle nonlinear dependencies and variables that are discrete or continuous, and univariate or multivariate.
However, as shown in our numerical experiments, the PC algorithm cannot be directly used for the time series case, in particular since autocorrelation can lead to high false positive rates (Fig.~\ref{fig:highdim_parcorr}).

Our proposed approach is also based on the conditional independence framework, but adapts it to the highly interdependent time series case. The method, which we name PCMCI, consists of two stages: (1)~PC$_1$ condition selection (Fig.~\ref{fig:mci_schematic}B, Algorithm~\ref{algo:pcs1}) to identify relevant conditions $\widehat{\mathcal{P}}(X^j_t)$ for all included time series variables $j\in \{1,\ldots,N\}$ and (2)~the \emph{momentary conditional independence} (MCI) test (Fig.~\ref{fig:mci_schematic}C, Algorithm~\ref{algo:mit}) to test whether $X^i_{t-\tau} \to X^j_{t}$ with 
\begin{align} \label{eq:mit_test}
\text{MCI:}~~~~X^i_{t-\tau} ~&\ci~ X^j_{t} ~|~ \widehat{\mathcal{P}}(X^j_t)\setminus \{X^i_{t-\tau}\},\,\widehat{\mathcal{P}}(X^i_{t-\tau})\,.
\end{align}
Thus, MCI conditions on both the parents of $X^j_{t}$ and $X^i_{t-\tau}$. These two stages serve the following purposes: PC$_1$ is a Markov set discovery algorithm based on the PC algorithm that removes irrelevant conditions for each of the $N$ variables by iterative independence testing (illustrated by shades of red and blue in Fig.~\ref{fig:mci_schematic}B). A liberal significance level $\alpha$ in the tests lets PC$_1$ adaptively converge to typically only few relevant conditions (dark red/blue) that include the causal parents $\mathcal{P}$ with high probability, but will also include some false positives (marked with an asterisk). 
The MCI test (Fig.~\ref{fig:mci_schematic}C) then addresses false positive control for the highly-interdependent time series case: As an example, for testing $X^1_{t-2} \to X^3_t$, the conditions $\widehat{\mathcal{P}}(X^3_t)$ (blue boxes in Fig.~\ref{fig:mci_schematic}B) are sufficient to establish conditional independence (Markov property), that is, to identify indirect and common cause links. On the other hand, the additional condition on the parents $\widehat{\mathcal{P}}(X^1_{t-2})$ (red boxes) accounts for autocorrelation leading to correctly controlled false positive rates at the expected level as further discussed below. The main free parameter of PCMCI is the significance level $\alpha$ in PC$_1$ which can be chosen based on model-selection criteria such as the Akaike Information Criterion (AIC) or cross-validation. Further technical details can be found in Supplementary Sect.~\ref{sec:causal_discovery_SI}.


\subsection*{Linear and nonlinear implementations} 
Both the PC$_1$ and the MCI step can be flexibly combined with any kind of conditional independence test. Here we present results for linear partial correlation (ParCorr) and nonlinear (GPDC and CMI) independence tests (Fig.~\ref{fig:mci_schematic}D). GPDC is based on Gaussian process regression\cite{Rasmussen2006} and a \emph{distance correlation}\cite{Szekely2007} test on the residuals which is suitable for a large class of nonlinear dependencies with additive noise. CMI is a fully non-parametric test based on a $k$-nearest neighbor estimator of conditional mutual information that accommodates almost any type of dependency\cite{Runge2018a}. The drawback of greater generality for GPDC or CMI, however, is lower power for linear relationships in the presence of small sample sizes. These conditional independence tests are further discussed in Supplementary Sect.~\ref{sec:ci_tests} and Tab.~\ref{tab:CI_overview}.

\section*{Results} 
\begin{figure*}[ht]  
\centering
\includegraphics[width=.5\linewidth]{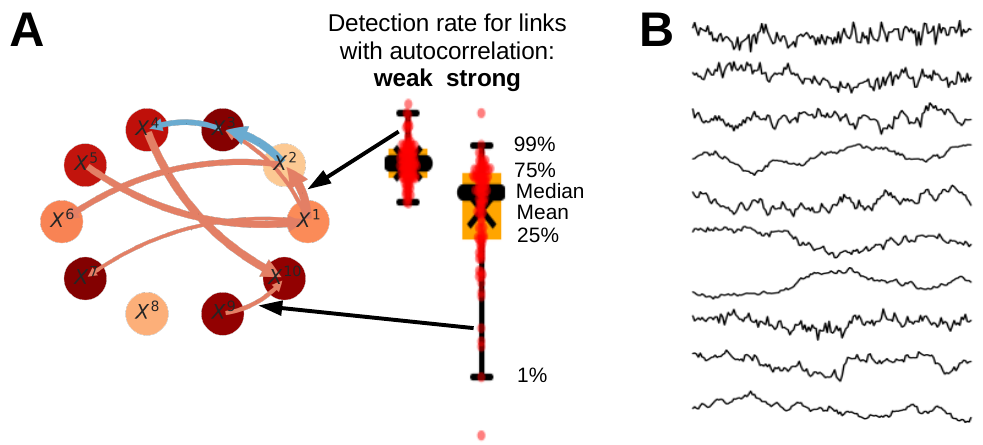}
\includegraphics[width=1.\linewidth]{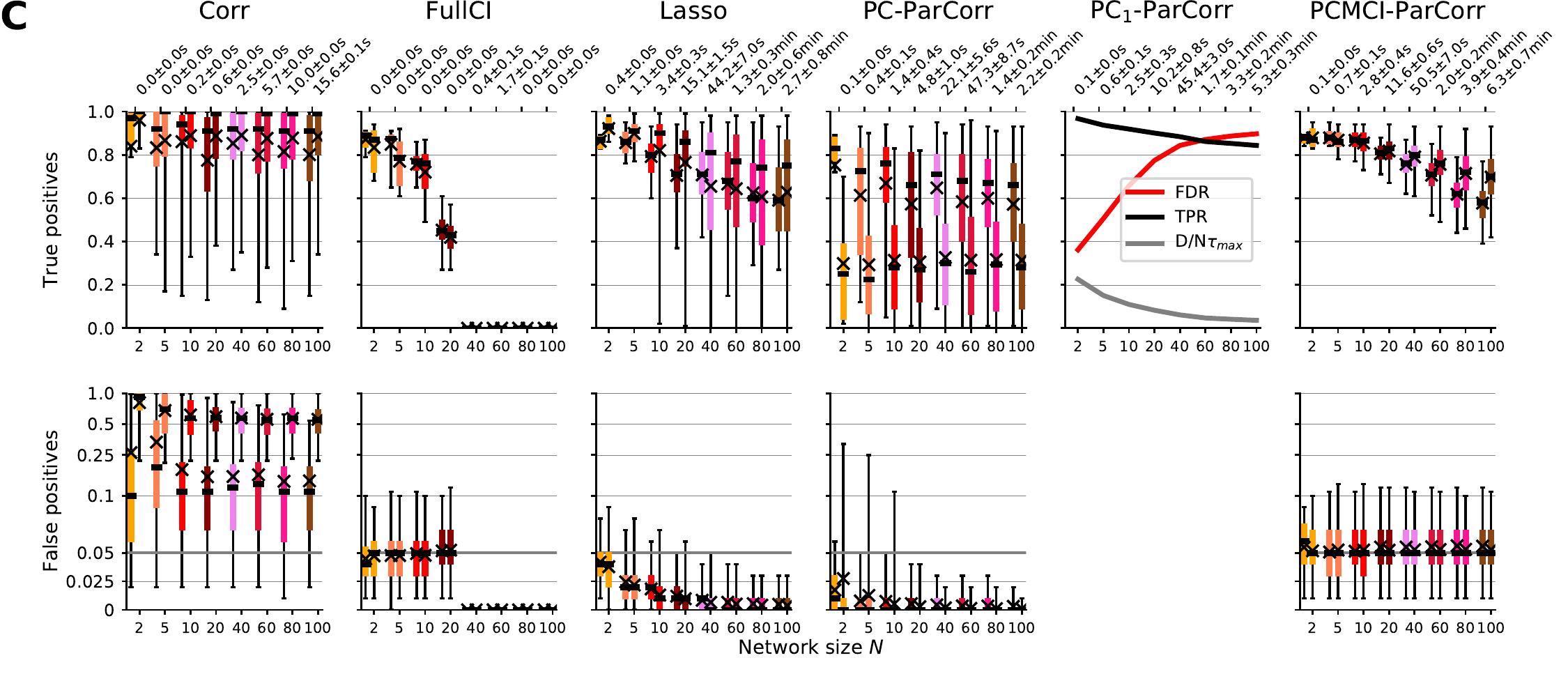}
\caption{Numerical experiments for models with linear dependencies.
(\textbf{A}) The full model setup is described in Sect.~\ref{sec:algo_model_description} and Tab.~\ref{tab:experiments}. In total 20 coupling topologies for each network size $N$ were randomly created, where all cross-link coefficients are fixed while the variables have different autocorrelations. Shown is an example network for $N=10$ with node colors denoting autocorrelation strength and the arrow colors the (positive or negative) coefficient strength. The arrow width here illustrates the detection rate of a particular method. 
As illustrated here, the boxplots in the figures below show the distribution of detection rates across individual links with the left (right) boxplot depicting links between weakly (strongly) autocorrelated variable pairs defined by the average autocorrelation of both variables being smaller or larger than $0.7$.
(\textbf{B}) Example time series realization of a model depicting partially highly autocorrelated variables. Each method's performance is assessed on 100 such realizations for each random network model. 
(\textbf{C}) Performance of different methods for models with linear relationships with time series length $T=150$, Tab.~\ref{tab:methods} provides implementation details.
The bottom row shows boxplot pairs (for weakly and strongly autocorrelated variables) of the distribution of false positives and the upper row the distributions of true positives for different network sizes $N$ along the $x$-axis in each plot. Average runtime and its standard deviation are given on top. The second last column gives the dimensionality (grey line, as a fraction of $N \tau_{\max}$), true positive rate (TPR, black line) and false discovery rate (FDR, red line) of the PC$_1$ condition-selection step.
}
\label{fig:highdim_parcorr}
\end{figure*}
\subsection*{Theoretical properties of PCMCI}
We briefly discuss several advantageous properties of PCMCI that are explained in more detail in Supplementary Sect.~\ref{sec:mit}: Consistency, generally larger effect size than FullCI, and interpretability as causal strength.
Consistency implies that PCMCI provably estimates the true graph in the limit of infinite sample size under standard assumptions of causal discovery. In Supplementary Sect.~\ref{sec:mit} we also elaborate on why MCI controls false positives correctly even for highly autocorrelated variables. Theoretical power levels for finite samples would require strong assumptions\cite{Robins2003,Kalisch2008} or are mostly impossible, especially for nonlinear associations. 
Due to the condition-selection step, MCI typically has a much lower conditioning dimensionality than FullCI. Further, avoiding conditioning on irrelevant variables also can be shown to yield a larger effect size than FullCI. Both of these factors lead to typically much higher detection power than FullCI as we show in our numerical experiments. 
Finally, MCI can be related to a hypothetical experimental setting in line with causal effect theory \cite{Pearl2000} and estimates a well-interpretable notion of \emph{causal strength}: MCI quantifies the causal effect of a hypothetical perturbation in $X^i_{t-\tau}$ on $X^j_t$ \cite{Pompe2011,Runge2015}. Thus, the value of the MCI statistic (e.g., partial correlation or CMI) allows to rank causal links in large-scale studies in a meaningful way as demonstrated in our numerical experiments.

\subsection*{PCMCI on real climate example}
Returning to the climate example (right panels in Fig.~\ref{fig:method_motivation}), PCMCI efficiently estimates the true causal relationships with high power in all three cases.
The condition-selection algorithm PC$_1$ identifies only the relevant conditions and finds, in particular, that $Z$ is not a parent of BCT. The MCI conditional independence test for the link ${\rm Nino}_{t-2}\to {\rm BCT}_t$ then has the same partial correlation effect size $\approx 0.10$ ($p=0.036$ in case A) in all three cases (Fig.~\ref{fig:method_motivation}A--C). The detection power is $>80$\% even for the high-dimensional case in Fig.~\ref{fig:method_motivation}C. Furthermore, PCMCI correctly estimates the causal effect strength $\approx 0.14$ among the pairs $(W^i,W^{i-1})$ resulting in similar detection power irrespective of varying autocorrelations in different $W^i$ time series.



\subsection*{Model setup for high-dimensional synthetic data experiments}
Following our illustrative climate example, we evaluate and compare the performance of our approach together with other common causal methods more systematically in numerical experiments. To validate and compare causal discovery methods, we ideally would have large-scale real datasets with known underlying ground truth of causal dependencies (e.g., derived from experiments). Since such datasets typically do not exist on a large scale in climate research (as well as many other fields), we validate the method with synthetic data that mimics the properties of real data, but where the true underlying relationships are known.

Here we model four of the major challenges of time series from complex systems such as the Earth: High-dimensionality, time lagged causal dependencies, autocorrelation, and potentially strong nonlinearity. 
Fig.~\ref{fig:highdim_parcorr}A gives an example model for $N=10$ variables and Fig.~\ref{fig:highdim_parcorr}B shows a time series realization illustrating some strongly autocorrelated variables. We create a number of models with different random network topologies of $N=2,\ldots,100$ time series variables with each network having $L=N$ linear or nonlinear causal dependencies (except for the bivariate case $N=2$ with $L=1$). From each of these models, we generate $100$ time series datasets (each of length $T$) to assess true and false positive rates of individual causal links in a model with the different causal methods. As illustrated in Fig.~\ref{fig:highdim_parcorr}A the boxplots in the following figures show the distribution of these individual link false and true positive rates across the large variety of random networks, for each network size $N$ differentiated between weakly and strongly autocorrelated pairs of variables in the left and right boxplot, respectively (defined by the average autocorrelation of both variables being smaller or larger than $0.7$).
The full model setup is detailed in Supplementary Sect.~\ref{sec:algo_model_description}, Tab.~\ref{tab:experiments} lists the experimental setups, and Tab.~\ref{tab:methods} gives details on the compared methods. 

\subsection*{Experiments with linear relationships}
In Fig.~\ref{fig:highdim_parcorr}C we first investigate the performance of linear causal discovery methods on numerical experiments with linear causal links. The setup has a sample length of $T=150$ observations and all cross-links have the same coupling coefficient and, hence, the same causal effect strength. Next to correlation (Corr) and FullCI (here implemented with an efficient vector-autoregressive model estimator), we compare PCMCI with the original PC algorithm as a standalone method and Lasso regression as the most widely used representative of regularized high-dimensional regression techniques that can be used for causal variable selection. Table~\ref{tab:methods} gives an overview over the compared methods, implementation details for alternative methods are given in Supplementary Sect.~\ref{sec:alternative_methods}. The maximum time lag is $\tau_{\max}=5$ for all methods.

Correlation is obviously inadequate for causal discovery with very high false positive rates (first column in Fig.~\ref{fig:highdim_parcorr}C). But even detection rates for true links vary widely with some links with under 20\% true positives, despite the equal coefficients for all causal links. This counterintuitive result is further investigated in Fig.~\ref{fig:algo_results_powerscaling}.
In contrast, all causal methods control false positives well around or below the chosen 5\% significance level with Lasso and the PC algorithm overcontrolling at lower than expected rates. An exception here are some highly autocorrelated links which are not correctly controlled with the PC algorithm (whiskers extending to 25\% false positives in Fig.~\ref{fig:highdim_parcorr}C) since it does not appropriately deal with the time series case. 

While FullCI has a detection power of around 80\% for $N=5$, this rate drops to 40\% for $N=20$ and FullCI cannot be applied anymore for larger $N$ when the dimensionality is larger than the sample size ($N \tau_{\max}>T=150$). But also for $N=5$ some links between strongly autocorrelated variables have a detection rate of just 60\%. 
Lasso has higher detection power than FullCI on average and the PC algorithm interestingly displays not much difference in detection power between $N=5$ and $N=100$, but the rates are lower than for Lasso on average, and higher autocorrelation also here has a detrimental effect. 

PCMCI robustly shows high detection power even for network sizes with dimensions exceeding the sample size and displays almost the same power for links with the same causal effect, regardless of whether autocorrelations are weak or strong, up to $N=20$. The second last column in  Fig.~\ref{fig:highdim_parcorr}C depicts the detection rates for the condition-selection step PC$_1$: Until $N=100$ still more than 80\% of the true parents are detected (black line). As mentioned, PC$_1$ is tuned to high power and for $N=100$ more than 80\% of the selected conditions are false positives (red line), but still the number of conditions is only a small fraction (grey line) of the conditions used for FullCI.

Runtime depends on implementation details, but all methods are in the same order of magnitude except for FullCI. Theoretically, in the worst case PCMCI scales polynomially with $N$ and $\tau_{\max}$. Our numerical experiments show that for smaller networks, PCMCI is faster than Lasso. Most of the time of PCMCI is spent on the condition-selection step, mainly because of the hyperparameter optimization of $\alpha$ via AIC. Fixing $\alpha$ is much faster and still gives good results (Figs.~\ref{fig:algo_par_corr_allvsmit_SI},\ref{fig:algo_par_corr_allvsmit_SI_300}), but may not always control false positives. 
The runtime of the standalone PC algorithm strongly depends on the number of conditioning sets tested. In theory \emph{all} combinations of conditioning sets are tested which makes PC extremely slow and leads to a highly varying runtime, but here we limited the number of combinations (see Supplementary Sect.~\ref{sec:alternative_methods}).

Summarizing, our key result here is that PCMCI has high power even for network sizes with dimensions, given by $N \tau_{\max}$, exceeding the sample size. Average power levels (marked by `x' in Fig.~\ref{fig:highdim_parcorr}C) are higher than FullCI and PC for all considered network sizes. PCMCI has similar \emph{average} power levels compared to Lasso, but an important difference is the \emph{worst-case} performance: Even for small networks ($N=10$), a significant part of the links is constantly overlooked with Lasso, while for PCMCI 99\% of the links have a detection power greater than 70\%.

\begin{figure*}[t]  
\centering
\includegraphics[width=.54\linewidth]{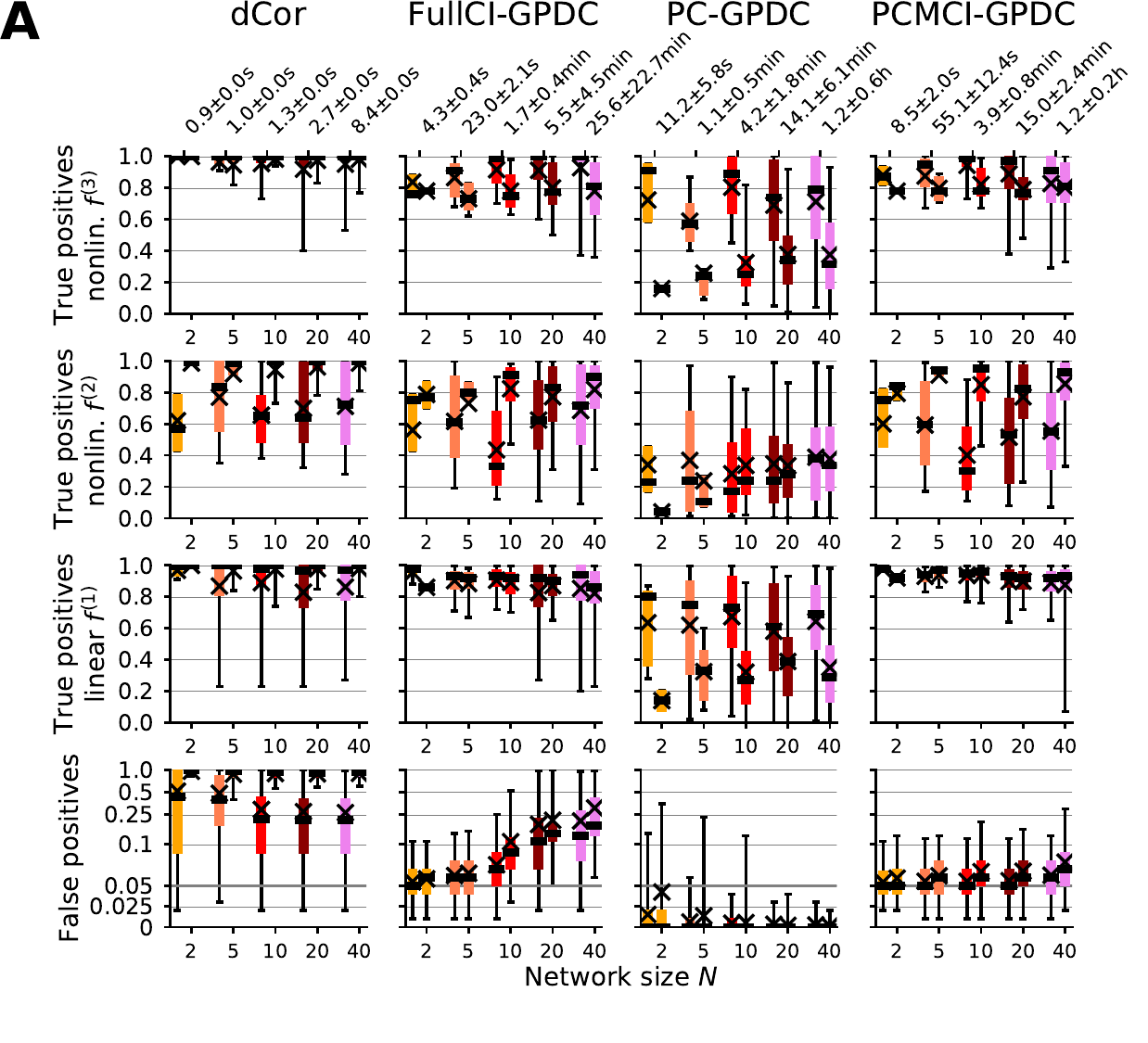}
\includegraphics[width=.44\linewidth]{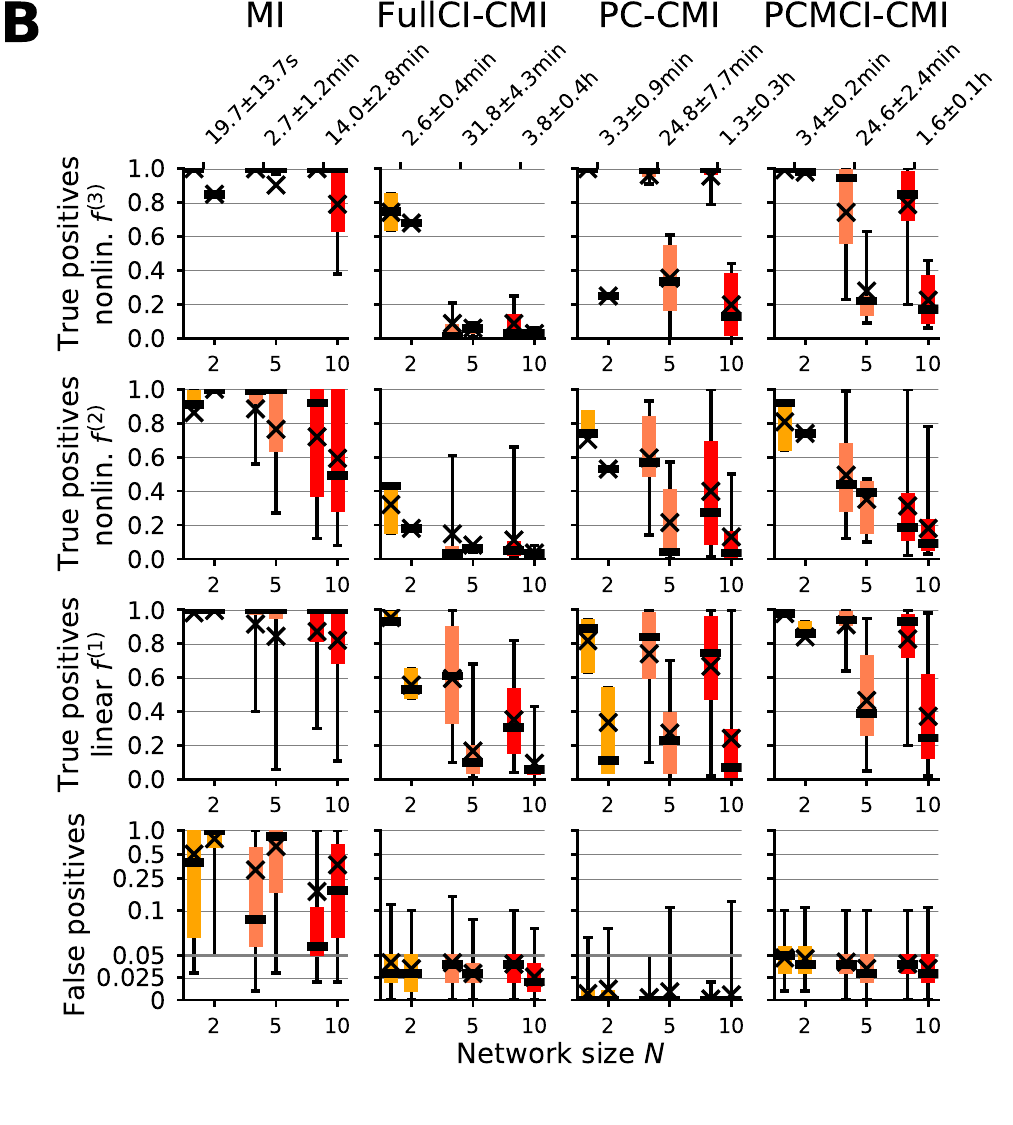}%
\caption{Numerical experiments for nonlinear models. We differentiate between linear and two types of nonlinear links (upper three rows). See Tab.~\ref{tab:experiments} for model setups and Supplementary Sect.~\ref{sec:ci_tests} and Tab.~\ref{tab:CI_overview} for a description of the nonlinear conditional independence tests.
(\textbf{A}) Results for GPDC implementation with $T=250$ where DCor denotes distance correlation.
(\textbf{B}) Results for CMI implementation with $T=500$ where MI denotes mutual information.
In Supplementary Figs.~\ref{fig:algo_results_samplesize_SI_gpdc},\ref{fig:algo_results_samplesize_SI_cmi} we investigate further sample sizes.
}
\label{fig:highdim_nonlin}
\end{figure*}

\subsection*{Experiments with nonlinear relationships}
Figure~\ref{fig:highdim_nonlin} displays results for nonlinear models where we differentiate between linear and two types of nonlinear links (upper three rows). In essence, here we find that PCMCI's ability to avoid high-dimensionality is even more crucial not only for detection power, but also to control false positives correctly.  

In Fig.~\ref{fig:highdim_nonlin}A, FullCI, PC, and PCMCI are all implemented with the GPDC conditional independence test and dCor denotes the distance correlation as the nonlinear analog to correlation (see Tab.~\ref{tab:CI_overview}). Distance correlation alone detects nonlinear links, but does not account for indirect or common driver effects leading to high false positives, especially for strong autocorrelation. FullCI here works well only up to $N=5$, but cannot control false positives anymore for $N\geq 10$ since the GPDC test does not work well in high dimensions. PC overcontrols false positives again (except for strong autocorrelation) and has the lowest power levels among all methods. PCMCI has the highest power levels which only slightly decrease for larger networks. Here we find that for nonlinear links weakly and strongly autocorrelated links also have different power levels, unlike for linear links (further discussed in Supplementary Sect.~\ref{sec:mit}). False positives are mostly controlled correctly, but there is a slight inflation of false positives for larger networks, again because even with condition-selection the dimensionality increases for larger networks and GPDC does not work well in high-dimensions.
For GPDC, runtime for PC and PCMCI is larger than for FullCI.

Figure~\ref{fig:highdim_nonlin}B depicts results for the fully non-parametric implementation with CMI. Then FullCI has the slowest runtime and almost no power, especially for nonlinear links, while PCMCI controls false positives correctly and has a higher power than PC except for some types of nonlinear links.

\begin{figure}[t]
\centering
\includegraphics[width=.7\linewidth]{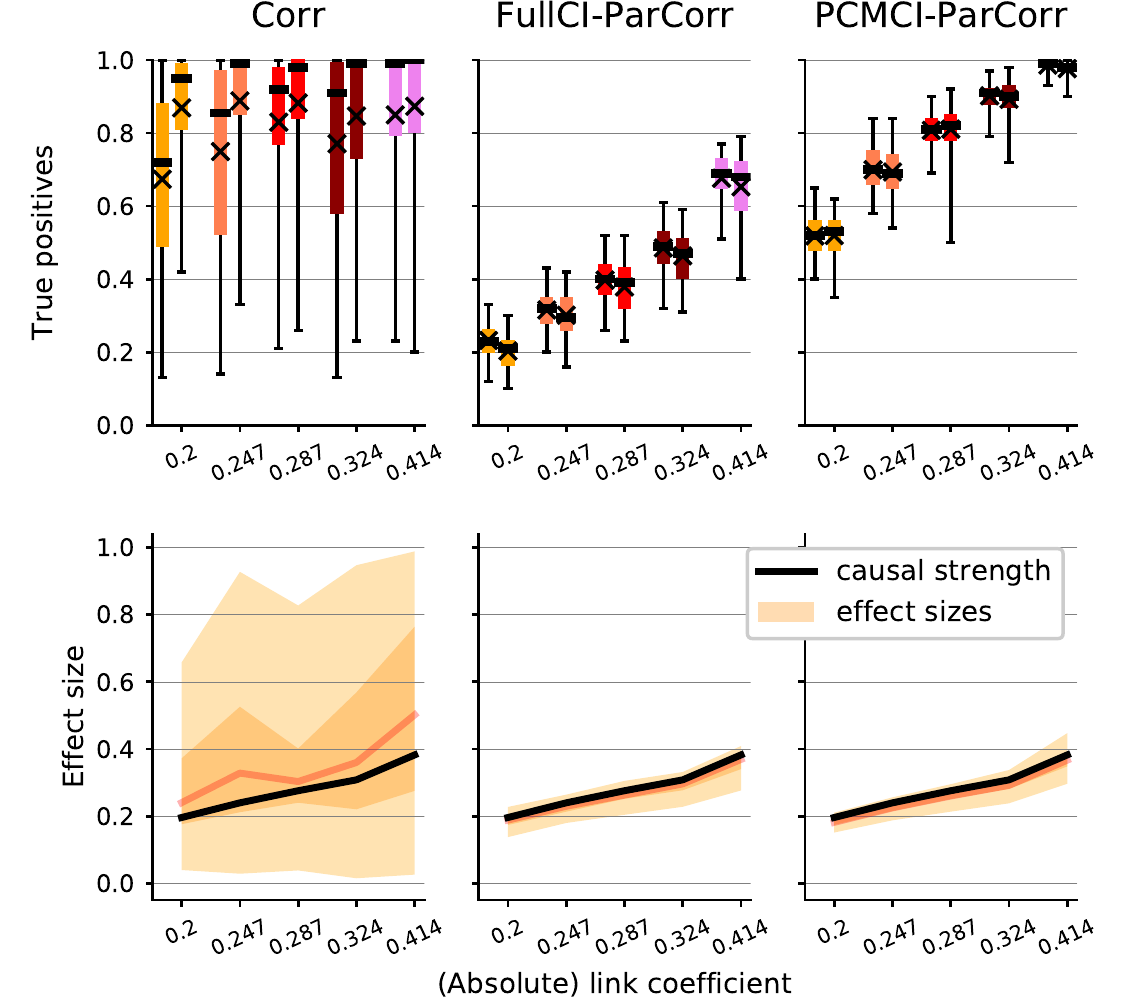}
\caption{Numerical experiments for causal strength.
Relation between detection power (top row) and effect size (bottom row) as given by correlation (Corr), FullCI partial correlation, and MCI partial correlation for different link coefficient strengths $c$ along the $x$-axis in each plot. In the bottom row the orange shades give the 1\%, 25\%, 75\%, and 99\% quantiles as well as the median of the respective (partial) correlations of all links (mean over 100 realizations for each link), and the black line denotes the (standardized) causal effect strength $|c|/\sqrt{1+c^2}$ which is the same for all links in a model.
}
\label{fig:algo_results_powerscaling}
\end{figure}

\subsection*{Causal strength experiments}
Finally, in Fig.~\ref{fig:algo_results_powerscaling} we more systematically investigate the close relationship between detection power and causal effect size illustrated in the climate example. In particular, we show that PCMCI has higher effect size than FullCI and often even higher effect size than correlation and, thus, more power.

Different from the model setup before, we now fix a network size of $N=20$ time series variables and vary the link coefficients $c$ ($x$-axis in Fig.~\ref{fig:algo_results_powerscaling}). The (standardized) causal effect is $|c|/\sqrt{1+c^2}$ (black line).
The bottom row of Fig.~\ref{fig:algo_results_powerscaling} depicts the distribution of effect sizes, that is, correlation, FullCI partial correlation, and MCI partial correlation across the different links for various random network topologies. Correlation values for links with the same causal effect span the whole range from zero to high correlation values indicating that correlation is rather unrelated to causal effect strength. Some correlation values are much smaller and even tend to zero which provides evidence for the observation that the detection power of correlation (or the other unconditional measures dCor and MI) can, counter-intuitively, even be lower than that of FullCI or PCMCI. The distribution of FullCI values is much narrower, but tends to be smaller than causal strength. In Supplementary Sect.~\ref{sec:mit} we provide proofs that MCI is larger than FullCI and that MCI is an estimator of causal strength as confirmed by our numerical experiments. In sum, larger effect size and lower dimensionality lead to more detection power of PCMCI compared to FullCI. 

\subsection*{Further experiments}
In the Supplementary Material we investigate some further methodological variants and show that our results are robust also for larger sample sizes (Supplementary Figs.~\ref{fig:algo_par_corr_allvsmit_SI},\ref{fig:algo_par_corr_allvsmit_SI_300},\ref{fig:algo_results_samplesize_SI},\ref{fig:algo_results_samplesize_SI_gpdc},\ref{fig:algo_results_samplesize_SI_cmi}), higher network coupling densities (Figs.~\ref{fig:algo_results_degree_SI},\ref{fig:algo_results_degree_300_SI}), and observational noise (Fig.~\ref{fig:algo_results_noise_SI}).
All methods display a similar sensitivity to observational noise with levels up to 25\% of the dynamical noise standard deviation having only minor effects. For levels of the same order as the dynamical noise we observe a stronger degradation with also the false positives not being correctly controlled anymore since common drivers are essentially not well detected anymore. See ref.~\cite{Runge2018b} for a discussion on observational noise.


\section*{Discussion and conclusion}  

Causal discovery on large-scale time series datasets is plagued by a dilemma: Including more variables makes an analysis more  credible regarding a causal interpretation, but if the added variables are irrelevant, that is, not explanatory for the observed relationships, they not only increase dimensionality but may also lead to smaller effect sizes. Both of these factors lead to lower power and increase the risk that important true causal links are overlooked. Furthermore, some nonlinear tests do not even control false positives anymore in high dimensions.

Our method circumvents this problem by a condition-selection step to remove irrelevant variables and a conditional independence test designed for highly interdependent time series. The former improves power levels for large-scale causal discovery analyses, while the latter also yields more power than classical techniques in analyses involving only few variables implying an improved `causal signal-to-noise ratio'. At the same time the MCI test demonstrates correctly controlled false positive rates even for highly autocorrelated time series data.
Furthermore, MCI can be interpreted as a measure of causal strength, allowing to rank causal links in exploratory studies on large datasets with many time series in a meaningful way. Such rankings can help to identify the strongest inferred causal links, which may be of main interest in some domain contexts.


PCMCI allows accommodating a large variety of conditional independence tests adapted to different types of data (see Sect.~\ref{sec:ci_tests}), for example, discrete or continuous time series. Networks can also be reconstructed with multivariate variables as nodes in the graph. This flexibility can help to represent causal associations on different aggregation levels and opens up a way to study causal networks on multiple interdependent layers\cite{Boccaletti2014}. 

Our method focuses on time-lagged dependencies, where there is no ambiguity in terms of cause-effect directionality. Recently, a growing body of literature addresses the inference of causality without relying on time-lags \cite{Spirtes2016,Peters2018} which could help to determine causal directionality for contemporaneous links.

For a causal interpretation, our approach rests on the standard assumptions \cite{Spirtes2000} of \emph{Causal Sufficiency}, implying that all common drivers are observed, and the \emph{Causal Markov Condition}, stating that once we know the direct causes of a variable, all other variables in the past become irrelevant for prediction, among other, more technical, assumptions. See ref.~\cite{Runge2018b} for an overview over causal discovery in time series.
Causal Sufficiency implies that the term `causal', as we use it here, has to be understood relative to the set of included variables. Non-included variables can still be the cause of a link in a non-experimental analysis. But the causal links inferred from the available observational data can then yield new hypotheses to be rejected or confirmed by further data analyses involving more variables (as illustrated in our climate example) or model simulations.
On the other hand, the finding of non-causality, that is, the absence of a causal link, relies on weaker assumptions\cite{Runge2018b}. Given that the observed data faithfully represents the underlying process and that potential nonlinearities are powerfully enough captured by the dependence measure, the absence of evidence for a statistical relationship makes it unlikely that a linking physical mechanism in fact exists. Such findings of non-causality are, therefore, more robust.


Growing data availability promises an unprecedented opportunity for novel insights through causal discovery across all disciplines of science. But current causality methods have been mostly limited to bivariate analyses that make a causal interpretation less credible. Our novel method enables multivariate causal analyses opening up new possibilities to more credibly reconstruct causal networks. 

\section*{Materials and Methods}
The main text describes the novel method introduced, some further material can be found in the Supplement. Software to reproduce the examples and numerical experiment results is available online under 

\texttt{https://github.com/jakobrunge/tigramite} 

including a comprehensive documentation of the method.




\section*{Acknowledgments}
We thank G. Balasis, D. Coumou, J. Donges, F. Fr\"ohlich, J. Haigh, J. Heitzig, J. Kurths, M. Mengel, M. Reichstein, C.-F. Schleussner, E. van Sebille, K. Zhang, and J. Zscheischler for helpful discussions and comments. The authors thank C. Linstead for help with high-performance computing. 

\section*{Funding}
J.R. received funding from a postdoctoral award by the James S. McDonnell Foundation. The work was supported by the European Regional Development Fund (ERDF), the German Federal Ministry of Education and Research and the Land Brandenburg for supporting this project by providing resources on the high-performance computer system at the Potsdam Institute for Climate Impact Research.
We declare no conflicts of interest.

\section*{Author contributions}
J.R. designed the method, analyzed the data, and prepared the manuscript. D.S. contributed to mathematical formulation, M.K. contributed to the climate analysis. All authors discussed the results and contributed to editing the manuscript.











\section*{Supplementary Materials}
Further material that can be found in the Supplement.


























\clearpage

\onecolumn
\renewcommand\theequation{S\arabic{equation}}
\setcounter{equation}{0}
\renewcommand\thefigure{S\arabic{figure}}    
\setcounter{figure}{0}   
\renewcommand\thesection{S\arabic{section}}    
\setcounter{section}{0} 
\renewcommand\thealgorithm{S\arabic{algorithm}}   
\setcounter{table}{0} 
\renewcommand\thetable{S\arabic{table}}    
\setcounter{algorithm}{0} 

{\Huge \textbf{Supplementary Material: Detecting causal associations in large nonlinear time series datasets}}
\\
{Jakob Runge, Peer Nowack, Marlene Kretschmer, Seth Flaxman, and Dino Sejdinovic}\\
\flushbottom

\tableofcontents
\thispagestyle{empty}

\clearpage
\section{Causal discovery} 
\label{sec:causal_discovery_SI}

In this section we describe the proposed causal discovery method PCMCI as well as alternative techniques in more detail. 

\subsection{Causal discovery method (PCMCI)} \label{sec:pcmci}
In our framework, the dependency structure of a set of time series variables is represented in a graphical model \cite{lauritzen1996graphical}. While the process graph depicted in Fig.~\ref{fig:problem_setting}B is easier to visualize, it does not fully represent the spatio-temporal dependency structure underlying complex dynamical systems which is more comprehensively grasped in a time series graph \cite{Eichler2011} as shown in Fig.~\ref{fig:mci_schematic}. If, for example, graphical models are estimated without taking into account lagged variables, associations can easily be confounded by the influence of common drivers at past times.


\begin{mydef}[Definition of time series graph] \label{eq:def_graph}
Let $\mathbf{X}$ be a multivariate discrete-time stochastic process and $\mathcal{G}=(V\times \mathbb{Z},\,E)$ the associated time series graph. The set of nodes in that graph consists of the set of components $V$ at each time $t\in \mathbb{Z}$. The edges $E$ of the graph are defined as follows: 
Variables $X^i_{t-\tau}$ and $X^j_t$ are connected by a lag-specific \emph{directed link} ``$X^i_{t-\tau} \to X^j_t$'' in $\mathcal{G}$ pointing forward in time if and only if $\tau>0$ and 
\begin{align}  
  X^i_{t-\tau}~~ \cancel{\ci}~~ X^j_t~~ |~~ \mathbf{X}_t^-\setminus \{X_{t-\tau}\},
\end{align}
where `$\cancel{\ci}$' denotes the absence of a conditional independence and $\mathbf{X}^-_{t}\setminus \{ X^i_{t-\tau}\}=(\mathbf{X}_{t-1},\,\mathbf{X}_{t-2},\,\ldots)\setminus \{ X^i_{t-\tau}\}$ the past of the multivariate process excluding $X^i_{t-\tau}$.
\end{mydef}
\noindent
That is, the graph is actually infinite, but in practice estimated up to some maximum time lag $\tau_{\max}$.
Throughout this work we assume stationarity, the links are repeated for every $t'<t$ if a link exists at time $t$. 
The parents of a node $X^j_t$ are defined as
\begin{align}
\mathcal{P}(X^j_t) &= \{X^k_{t-\tau}:~ X^k\in \mathbf{X},~\tau>0,~X^k_{t-\tau}\to X^j_t\}\,, 
\end{align}
In summary, the parents $\mathcal{P}(X^j_t)$ for all variables $X^j \in \mathbf{X}$ define the graph $\mathcal{G}$. Contemporaneous dependencies for $\tau=0$ can be defined in different ways \cite{Runge2015}. Here they are left undirected, but other techniques \cite{Peters2018,Spirtes2016} could be applied to determine causal directionality for contemporaneous links. 

Our causal discovery technique to estimate the time series graph is based on a two-step procedure:
\begin{enumerate}
\item Condition-selection: Obtain an estimate $\widehat{\mathcal{P}}(X^j_t)$ of (a superset of) the parents $\mathcal{P}(X^j_t)$ for all variables $X^j_t \in \mathbf{X}_t=(X^1_t,\,X^2_t,\ldots,X^N_t)$ with Algorithm~\ref{algo:pcs1}.
\item Use these parents as conditions in the MCI causal discovery Algorithm~\ref{algo:mit}, which tests \emph{all} variable pairs $(X^i_{t-\tau},X^j_{t})$ with $i,j\in \{1,\ldots,N\}$ and time-delay $\tau \in\{1,\ldots,\tau_{\max}\}$ and establishes a link, that is, $X^i_{t-\tau} ~\to~ X^j_{t} \in \mathcal{G}$, if and only if
\begin{align} \label{eq:mit_test_SI}
\text{MCI:}~~X^i_{t-\tau} ~&\cancel{\ci}~ X^j_{t} ~|~ \widehat{\mathcal{P}}(X^j_t)\setminus \{X^i_{t-\tau}\},\,\widehat{\mathcal{P}}_{p_X}(X^i_{t-\tau})\,,
\end{align}
where $\widehat{\mathcal{P}}_{p_X}(X^i_{t-\tau})\subseteq \widehat{\mathcal{P}}(X^i_{t-\tau})$ denotes the $p_X$ strongest parents according to the sorting in Algorithm~\ref{algo:pcs1}. This parameter is just an optional choice. One can also restrict the maximum number of parents used for $\widehat{\mathcal{P}}(X^j_t)$, but here we impose no restrictions. For $\tau=0$ one can also consider undirected contemporaneous links \cite{Runge2015}.
\end{enumerate}
Both the conditional independence tests in the condition-selection step and the MCI test can be implemented with different test statistics as detailed in Sect.~\ref{sec:ci_tests} and Tab.~\ref{tab:CI_overview}.

Algorithm~\ref{algo:pcs1} in the first step is a variant of the skeleton-discovery part of the PC algorithm \cite{Spirtes1991} in its more robust modification called PC-stable \cite{Colombo2014} and adapted to time series. The algorithm then is as follows: For every variable $X^j_t\in \mathbf{X}_t$ the algorithm starts by initializing the preliminary parents $\widehat{\mathcal{P}}(X^j_t) = (\mathbf{X}_{t-1},\,\mathbf{X}_{t-2},\,\ldots,\,\mathbf{X}_{t-\tau_{\max}})$. Then we start with $p=0$ and iteratively remove a variable $X^i_{t-\tau}$ from $\widehat{\mathcal{P}}(X^j_t)$ if the null hypothesis
\begin{align} \label{eq:pc_test}
\text{PC:}~~~X^i_{t-\tau} ~\ci~ X^j_{t} ~|~\mathcal{S} ~~~\text{for any $\mathcal{S}$ with $|\mathcal{S}|=p$}
\end{align}
is accepted at a significance threshold $\alpha$, where $\mathcal{S}$ are different combinations of subsets of $\widehat{\mathcal{P}}(X^j_t)\setminus \{X^i_{t-\tau}\}$ with cardinality $p$ up to a maximum number of combinations $q_{\max}$. In the first step ($p=0$), $\mathcal{S}$ is empty and we, thus, test unconditional dependencies. In each next step, the cardinality is increased $p\to p+1$ and Eq.~\ref{eq:pc_test} is tested again for different combinations of $\mathcal{S}$. The algorithm converges for a link $X^i_{t-\tau} \to X^j_{t}$ once $\mathcal{S}=\widehat{\mathcal{P}}(X^j_t)\setminus \{X^i_{t-\tau}\}$ and the null hypothesis $X^i_{t-\tau} ~\ci~ X^j_{t} ~|~\widehat{\mathcal{P}}(X^j_t)\setminus \{X^i_{t-\tau}\}$ is rejected (if the null hypothesis is accepted, the link is removed). Our fast variant PC$_1$ is obtained by restricting the maximum number of combinations $q_{\max}$ per iteration to $q_{\max}=1$. Further, we sort $\widehat{\mathcal{P}}(X^j_t)$ after every iteration according to the test statistic value (ParCorr, GPDC or CMI) and pick $\mathcal{S}$ in lexicographic order, that is, for $q_{\max}=1$, only the conditions with highest association.  Other causal variable-selection algorithms employ similar heuristics \cite{Aliferis2010,Aliferis2010b,Sun2015}. 
The MCI step is inspired by the information-theoretic measure \textit{momentary information transfer} introduced in \cite{Pompe2011,Runge2012b}.

\begin{table}[b]
\centering
\caption{Parameters of PCMCI method.}
\begin{tabular}{c|c|c}
Parameter &  Description & Recommended values \\
\midrule
 $\alpha$  &  Significance threshold in pre-selection Algorithm~\ref{algo:pcs1}   & \makecell{ParCorr: by AIC\\ GPDC, CMI: $0.2$} \\
 $p_{X}$  &  Max. number of parents of $X^i_{t-\tau}$ in Algorithm~\ref{algo:mit}    &  \makecell{unrestricted (but can be as low as 1, see discussion)} \\
$\tau_{\max}$  &  Max. time delay in Algorithms~\ref{algo:pcs1},\ref{algo:mit}  & last lag with significant unconditional association \\
\bottomrule
\end{tabular}
\label{tab:parameters}
\end{table}

As listed in Tab.~\ref{tab:parameters}, the free parameters of this method (in addition to free parameters of the conditional independence test statistic) are the maximum time delay $\tau_{\max}$, the significance threshold $\alpha$, and the maximum number $p_X$ of conditions of the driver variable in Algorithm~\ref{algo:mit}. We abbreviate different parameter choices by PC$_1^{\alpha}$+MCI$_{p_X}$, if not clear from the context. In Fig.~\ref{fig:algo_results_pcthres_SI} we show the overall performance of PCMCI for different choices of $\alpha$.
In the present implementation we do not take into account the reduced degrees of freedom in the MCI test due to the condition-selection step. We have not found any sign of inflated false positives in our numerical experiments. However, in order to ensure a more conservative false positive control, one could split the dataset and conduct the condition-selection on a separate part of the dataset than that used for the MCI tests. This would, however, clearly lead to a decrease in power and carry the problem of choosing the split in time series.

\paragraph{Choice of $\tau_{\max}$} The maximum time delay depends on the application and should be chosen according to the maximum physical time lag expected in the complex system. In practice we recommend a rather large choice that includes peaks in the lagged cross-correlation function (or a more general measure corresponding to the chosen independence test), because a too large choice of $\tau_{\max}$ merely leads to longer runtimes of PCMCI, but not to an increased estimation dimension as for FullCI.


\paragraph{Choice of $\alpha$}
$\alpha$ should not be seen as a significance test level in PC$_1$ since the iterative hypothesis tests do not allow for a precise assessment of uncertainties. $\alpha$ here rather takes the role of a regularization parameter as in model-selection techniques. The conditioning sets $\widehat{\mathcal{P}}$ estimated with PC$_1$ should include the true parents and at the same time be small in cardinality to reduce the estimation dimension of the MCI test and improve its power. But the first demand is typically more important.  In Fig.~\ref{fig:algo_results_pcthres_SI} we investigate the performance of PCMCI implemented with ParCorr, GPDC, and CMI for different $\alpha$. 
Too small values of $\alpha$ result in many true links not being included in the condition set for the MCI tests and, hence, increase false positives. Too high levels of $\alpha$ lead to high dimensionality of the condition set which reduces detection power and increases the runtime.
Note that for a threshold $\alpha=1$ in Algorithm~\ref{algo:pcs1}, no parents are removed and all $N \tau_{\max}$ variables would be selected as conditions. Then the MCI test becomes a FullCI test. 
As in any variable-selection method \cite{Aliferis2010,Aliferis2010b}, $\alpha$ can be optimized using cross-validation or based on scores such as BIC or AIC. For all ParCorr experiments (except for the ones labeled with PC$_1^{\alpha}$+MCI$_{p_X}$), we optimized $\alpha$ with AIC as a selection criterion. More precisely, for each $X^j_t$ we ran PC$_1$ separately for each $\alpha \in \{0.1, 0.2, 0.3, 0.4\}$ yielding different conditioning sets $\widehat{\mathcal{P}}^{\alpha}(X^j_t)$. Then we fit a linear model for each $\alpha$
\begin{align}
X^j_t &= \widehat{\mathcal{P}}^{\alpha}(X^j_t) \beta\,,
\end{align}
yielding the residual sum of squares (RSS) and select $\alpha$ according to Akaike's Information criterion modulo constants
\begin{align}
\alpha^* &= {\argmin}_\alpha n \log(RSS^{\alpha}) + 2|\widehat{\mathcal{P}}^{\alpha}(X^j_t)|\,,
\end{align}
where $n$ is the sample size (typically the time series length $T$ minus a cutoff due to $\tau_{\max}$) and $|\cdot|$ denotes cardinality. For GPDC one can similarly select $\alpha$ based on the log marginal likelihood of the fitted Gaussian process, while for CMI one can use cross-validation based on nearest-neighbor predictions for different $\widehat{\mathcal{P}}^{\alpha}(X^j_t)$. But since GPDC and CMI are already quite computationally demanding, we picked $\alpha=0.2$ in all experiments based on our findings in Fig.~\ref{fig:algo_results_pcthres_SI}.
In the lower panels of Figs.~\ref{fig:algo_par_corr_allvsmit_SI},\ref{fig:algo_par_corr_allvsmit_SI_300},\ref{fig:algo_results_degree_SI},\ref{fig:algo_results_degree_300_SI} we analyzed $\alpha=0.2$ also for ParCorr for all numerical experiments and found that this option also gave good results for sparse networks and runs even faster than Lasso. For very dense networks, we, however, found inflated false positives. Unfortunately, we have no finite sample consistency results for choosing $\alpha$.

\paragraph{Choice of $p_X$}
While the parents $\widehat{\mathcal{P}}(X^j_t)$ are sufficient to assess conditional independence, the additional conditions $\widehat{\mathcal{P}}_{p_X}(X^i_{t-\tau})\subseteq \widehat{\mathcal{P}}(X^i_{t-\tau})$  are used to account for autocorrelation and make the MCI test statistic a measure of causal strength as analyzed in Sect.~\ref{sec:mit}. To limit high dimensionality, one can strongly restrict the number of conditions $\widehat{\mathcal{P}}_{p_X}(X^i_{t-\tau})$ with the free parameter $p_X$. 
To avoid having another free parameter, we kept $p_X$ unrestricted in most experiments, but in our tests we found that a small value $p_X=1$ or $p_X=3$ already suffices to reduce inflated false positives due to strong autocorrelation and estimate causal strength, resulting in a confined scaling of power with link strength. The reason is that typically the largest driver will be the autodependency and conditioning out its influence already diminishes the effect of strong autocorrelations. 

In the lower panels of Figs.~\ref{fig:algo_par_corr_allvsmit_SI},\ref{fig:algo_par_corr_allvsmit_SI_300},\ref{fig:algo_results_degree_SI},\ref{fig:algo_results_degree_300_SI} we found that PCMCI with $p_X=3$ is slightly faster and yields more power. It also leads to a smaller power difference between weakly and strongly autocorrelated links than with MCI$_{all}$. In theory, a too small $p_X$ does not guarantee a well-calibrated test (see Sect.~\ref{sec:mci_props}), but in practice it seems like a sensible tradeoff.

\paragraph{False discovery rate control}
PCMCI can also be combined with \textit{false discovery rate} controls, e.g., using the Hochberg-Benjamini approach \cite{Benjamini1995}. This approach controls the expected number of false discoveries by adjusting the $p$-values resulting from the MCI step for the whole time series graph. More precisely, we obtain the $q$-values as
\begin{align}
q &= \min( p\frac{m}{r} , 1)\,,
\end{align}
where $p$ is the original p-value, $r$ is the rank of the original p-value when p-values are sorted in ascending order, and $m$ is the number of computed p-values in total, that is, $m=N^2\tau_{\max}$ to adjust only directed links for $\tau>0$ and correspondingly if also contemporaneous links for $\tau=0$ are taken into account. 
The false discovery approach controls the expected proportion of discoveries (rejected null hypotheses) that were false. In our numerical experiments we did not control the false discovery rate since we are interested in the individual link performances.  

\subsection{Alternative methods} \label{sec:alternative_methods}

Here we define alternative causal and non-causal methods, Tab.~\ref{tab:methods} gives an overview over the methods compared in the numerical experiments.

\subsubsection{FullCI} \label{sec:fullci}
The most straightforward way to test the existence of causal links as defined in Def.~\eqref{eq:def_graph} is to directly test 
\begin{align} \label{def:fullci}
\text{FullCI:}~~~~X^i_{t-\tau} ~&\ci~ X^j_{t} ~|~ \mathbf{X}^-_{t}\setminus \{X^i_{t-\tau}\}\,,
\end{align}
where $\mathbf{X}=(X^1,\,X^2,\ldots)$ denotes the multivariate process and $\mathbf{X}^-_{t}=(\mathbf{X}_{t-1},\,\mathbf{X}_{t-2},\,\ldots)$ its past. In practice, the past is truncated at a maximum time lag $\tau_{\max}$. This test can be implemented with any of the conditional independence tests defined in Sect.~\ref{sec:ci_tests}

In its original formulation, Granger causality between time series variables $X$ and $Y$ is based on fitting a linear or nonlinear model including possible confounders to $Y$ and a causal link $X\to Y$ is assessed by quantifying whether the inclusion of the past of variable $X$ in the model significantly reduces the prediction error about $Y$ \cite{Granger1969}. Our formulation can be interpreted as a general, lag-specific version.
For linear models we implement FullCI by fitting a vector-autoregressive (VAR) model using the \texttt{statsmodels} package which also allows to obtain the corresponding p-values.


\subsubsection{Lasso} \label{sec:lasso}

In the linear case, FullCI can be phrased as testing for nonzero coefficients in a VAR model. If implemented using standard ordinary least-square regression, the problem becomes ill-defined if the number of coefficients exceeds the number of samples. One way to address this problem are \emph{regularized} high-dimensional regression techniques. In particular, Lasso\cite{Tibshirani1996}  is a common regression method that can also be used for variable selection. Lasso regression is known to be inconsistent in some scenarios which is overcome by the \emph{adaptive Lasso}\cite{Zou2006} which yields a consistent estimator by utilizing an adaptively weighted penalty. Here we implemented an adaptive Lasso that consists of computing several Lasso regressions with iterative feature reweighting, see Algorithm~\ref{algo:stabs_lasso}. After several iterations the active set of variables is determined as the non-zero coefficients. Then all zero-coefficients are assigned a p-value of one, while the p-values for the active set of variables is determined by an OLS regression including only the active variables.
To select the optimal regularization parameter we used a time-series based cross-validation scheme. Some tests based on AIC-selected hyperparameters did not yield good results.

\subsubsection{PC algorithm} \label{sec:pcalgo}
The original PC algorithm was formulated for general random variables without assuming a time order. It consists of several phases where first, in the skeleton-discovery phase, an undirected graphical model \cite{lauritzen1996graphical} is estimated whose links are then oriented using a set of rules \cite{Spirtes1991,Spirtes2000}.
We implement the skeleton-discovery phase of the PC algorithm as given in Algorithm~\ref{algo:pcs1}, but in contrast to our fast variant, the original version does not restrict the number of condition combinations $q_{\max}$ to test. Here we use a large choice of $q_{\max}=10$. In contrast to FullCI or PCMCI, it is not straightforward to assess the confidence of causal links with the PC algorithm because links are iteratively removed. Here we use a $p$-value assessment for causal links according to ref.~\cite{Tsamardinos2008a}, 
\begin{align} 
p(X^i_{t-\tau} \to X^j_{t}) &= \max_{\{\mathcal{S}\}} p \left( X^i_{t- \tau}  \ci X^j_{t}|\mathcal{S} \right)\,,
\end{align}
that is, the maximum of all $p$-values from the conditional independence tests for different condition sets $\mathcal{S}$ in Eq.~\ref{eq:pc_test} defines the aggregated $p$-value of a causal link. This causal discovery method we term PC in the numerical experiments. 
We found that the $p$-values tend to be over-conservative for the most part with outliers for strong autocorrelations. FPR levels cannot be reliably controlled just below a desired threshold.

\subsubsection{PCMCI variants} \label{sec:ity}
In the MCI step of the PCMCI method, one can also test
\begin{align} \label{eq:ity_test}
\text{MCI$_0$:}~~~~X^i_{t-\tau} ~&\ci~ X^j_{t} ~|~ \mathcal{P}(X^j_t)\setminus \{X^i_{t-\tau}\}
\end{align}
for all links, that is, the MCI test without the condition on the parents $\mathcal{P}(X^i_{t-\tau})$, or, equivalently, $p_X=0$, denoted PC$_1$+MCI$_0$. Note that for $X^i_{t-\tau}\in\mathcal{P}(X^j_t)$ this is the test in the last step of the PC algorithm, but here also links that were removed in Algorithm~\ref{algo:pcs1} are tested again. 
In our numerical experiments (Figs.~\ref{fig:algo_par_corr_allvsmit_SI},\ref{fig:algo_par_corr_allvsmit_SI_300},\ref{fig:algo_results_degree_SI},\ref{fig:algo_results_degree_300_SI}) we found that PC$_1$+MCI$_0$ has inflated false positives.

We also test a variant, called PC$_1$+MCI$_0$pw, where all time series are pre-whitenend prior to running PC$_1$+MCI$_0$, that is, we preprocessed all $N$ time series by estimating the univariate lag-1 autocorrelation coefficients $\hat{a}_i=\rho(X^i_{t-1};X^i_t)$ and regressing out the AR(1) autocorrelation part of the signals:
\begin{align}
\tilde{X}^i_t = X^i_t - \hat{a}_i X^i_{t-1} ~~\text{$\forall t$ and $i=1,\ldots,N$} \,. 
\end{align}
Then the PCMCI test is applied to these residuals $\tilde{\mathbf{X}}$. Our numerical results in Figs.~\ref{fig:algo_par_corr_allvsmit_SI},\ref{fig:algo_par_corr_allvsmit_SI_300},\ref{fig:algo_results_degree_SI},\ref{fig:algo_results_degree_300_SI} show that this approach also fails to control false positives.

For conditional mutual information as a test statistic, MCI$_0$ was called \textit{information transfer to Y} (ITY) in ref.~\cite{Runge2012b}. For a link $X\to Y$, it quantifies the unique information in $X$ entering a process $Y$ without excluding the information in the past of $X$. ITY and the information-theoretic version of MCI called \emph{momentary information transfer} (MIT) are part of a whole suite of measures to quantify causal information flow in complex systems \cite{Runge2015}. Another measure called ITX quantifies the unique information emanating from a process as an information-theoretic analog to Sims causality, and several other measures quantify mediation on information pathways as discussed in ref.~\cite{Runge2015}.

\subsubsection{BivCI} \label{sec:biv}
In the numerical experiments, we also evaluate a bivariate conditional independence test (BivCI),
\begin{align} \label{eq:biv_test}
\text{BivCI:}~~X^i_{t-\tau} ~&\ci~ X^j_{t} ~|~ X^{j-}_{t}\,.
\end{align}
where $X^{j-}_{t}= (X^{j}_{t-1},\,X^{j}_{t-2},\,\ldots,\,X^{j}_{t-\tau_{\max}})$ denotes the past of $X^j_t$. This corresponds to a pairwise lag-specific version of Granger causality or Transfer Entropy\cite{Schreiber2000b}. Thus, this test removes autocorrelation to some extent, but does not exclude common drivers or indirect dependencies. Also autodependencies induced by common drivers are not removed. As analyzed in Figs.~\ref{fig:algo_par_corr_allvsmit_SI},\ref{fig:algo_par_corr_allvsmit_SI_300},\ref{fig:algo_results_degree_SI},\ref{fig:algo_results_degree_300_SI} BivCI reduces autocorrelation effects to some extent, but false positives are not correctly controlled, as expected.


\subsubsection{Unconditional pairwise measures} \label{sec:corr}
We also investigate the unconditional pairwise measures $I(X_{t-\tau};Y_t)$. In the ParCorr implementation this is simply the Pearson correlation coefficient (Corr), in the GPDC implementation the distance correlation (dCor), and for CMI the mutual information (MI).


\section{Conditional independence tests} \label{sec:ci_tests}

Similarly to the PC algorithm, the PCMCI framework (Algorithm~\ref{algo:pcs1} and \ref{algo:mit}) that we propose can be used in conjunction with any conditional independence test -- these will typically be based on estimating different dependence measures with associated test statistics. Here we implement three tests: partial correlation (ParCorr), a nonlinear two-step conditional independence test we term GPDC, and a fully non-parametric test based on conditional mutual information\cite{Runge2018a}. Table~\ref{tab:CI_overview} gives an overview over the tests.

\subsection{ParCorr}  \label{sec:parcorr}
Partial correlation for testing $X \ci Y ~|~ \mathbf{Z}$ here  is estimated (Tab.~\ref{tab:CI_overview}) in a two-stage procedure with a multivariate regression of $X$ and $Y$ on $\mathbf{Z}$ followed by a correlation test on the residuals.
Its advantages are fast computation and that the distribution under the null hypothesis of conditional independence is known analytically, but it is applicable only to the multivariate Gaussian case which can only capture linear dependencies.

\subsection{GPDC} \label{sec:gpdc}
GPDC also belongs to the class of residual-based conditional independence tests. Instead of a linear regression, here the first step is a Gaussian process (GP) regression \cite{Rasmussen2006} and the second step a test for the (unconditional) independence of the uniformized residuals with the \emph{distance correlation coefficient} \cite{Szekely2007}. GP regression is a widely used Bayesian nonparametric regression approach.
Distance correlation\cite{Szekely2007} is a measure of dependence between two random variables and is zero if and only if the variables are independent. Thus, distance correlation measures both linear and nonlinear association. See Tab.~\ref{tab:CI_overview} for details. Note that the underlying assumption of GPDC is that of an additive functional dependency, see ref.~\cite{Zhang2017b} for a more general, but still residual-based, test.

The estimator for distance correlation is defined as 
\begin{align}
\operatorname{dCor}(X,Y) = \frac{\operatorname{dCov}(X,Y)}{\sqrt{\operatorname{dVar}(X)\,\operatorname{dVar}(Y)}},
\end{align}
where the distance covariance $\operatorname{dCov}$ and variance $\operatorname{dVar}$ are computed by 
\begin{align}
\operatorname{dCov}^2(X,Y) = \frac 1 {n^2} \sum_{j = 1}^n \sum_{k = 1}^n A_{j, k} \, B_{j, k}\,,
\end{align}
where $A_{j, k} \, B_{j, k}$ are the doubly-centered distance matrices of $X$ and $Y$ (see ref.~\cite{Szekely2007}), respectively, and $n$ is the sample size (in the present context the length of the time series minus cutoffs due to $\tau_{\max}$).
Distance correlation is implemented in \texttt{Tigramite} based on the code in the original \texttt{dcov.test} function in the energy package for the R-language.
The prior transformation of $X$ and $Y$ to uniform marginals allows to pre-compute the distribution for each sample size $n$ (implemented in \texttt{Tigramite}) and critical values under the independence null hypothesis -- thereby avoiding computationally expensive permutation tests for each test of $X~\ci~Y~|~\mathbf{Z}$. Note that since we are not aware of a way to account for the GP-step in assessing the degrees of freedom for the subsequent distance correlation, we simply used the sample size $n$ which did not seem to inflate false positives in our experiments.

Figure~\ref{fig:mci_schematic}D illustrates a quadratic relationship $X=Z^2+\eta^X$ and $Y=-Z^2+\eta^Y$ for $Z\sim \mathcal{N}(0, 1)$ where $X ~\ci Y~|~Z$ can be identified with GPDC, but not with ParCorr.

\subsection{CMI} \label{sec:cmiknn}
ParCorr and GPDC, like any two-step procedure for conditional independence testing, which includes a regression followed by an unconditional test on the regression residuals, has an underlying assumption of additive noise and a functional dependence on the conditioning variables. In the presence of dependencies which cannot be represented even in a nonlinear functional form, regression will not be able to remove these dependencies on the conditioning variables. Figure~\ref{fig:mci_schematic}D for CMI illustrates a multiplicative case with $X=Z\eta^X$ and $Y=-Z \eta^Y$ for $Z\sim \mathcal{N}(0, 1)$ where both ParCorr and GPDC fail to establish $X ~\ci Y~|~Z$. Then the two-step procedure should be replaced with fully non-parametric techniques to measuring and testing conditional dependence.

Here a fully non-parametric test\cite{Runge2018a} for continuous data based on conditional mutual information combined with a local permutation scheme is implemented. The conditional mutual information (CMI) is zero if and only if $X \ci Y | \mathbf{Z}$. 
From the nearest-neighbor entropy estimator by Kozachenko et al.\cite{kozachenko1987sample}, Kraskov et al.\cite{Kraskov2004a} developed an estimator for mutual information that was generalized to CMI\cite{FrenzelPompe2007,Vejmelka2008}:
\begin{align} \label{eq:cmi_knn_est}
     \widehat{I}_{ XY|Z} &=   \psi (k) + \frac{1}{n} \sum_{i=1}^n \left[ \psi(k^{z}_i) - \psi(k^{xz}_i) - \psi(k^{yz}_i) \right]
\end{align}
with the Digamma function as the logarithmic derivative of the Gamma function $\psi(x)=\frac{d }{d x} \ln \Gamma(x)$ and sample length $n$. The only free parameter $k$ is the number of nearest neighbors in the joint space of $\mathcal{X}\otimes \mathcal{Y}\otimes \mathcal{Z}$ which defines the local length scale (in maximum norm) $\epsilon_i$ around each sample point $i$. Then $k^{xz}_i$, $k^{yz}_i$ and $k^z_i$ are computed by counting the number of points with distance strictly smaller than $\epsilon_i$ (including the reference point $i$) in the subspace $\mathcal{X}\otimes \mathcal{Z}$ to get $k^{xz}_i$, in the subspace $\mathcal{Y}\otimes \mathcal{Z}$ to get $k^{yz}_i$, and in the subspace $\mathcal{Z}$ to get $k^z_i$.
The decisive advantage of this estimator compared to fixed global bandwidth approaches is its local \emph{data-adaptiveness} (Fig.~\ref{fig:mci_schematic}D): The hypercubes around each sample point are smaller where more samples are available. As opposed to GP regression, this feature allows to detect also highly non-smooth dependencies. Unfortunately, few theoretical results are available for the complex mutual information estimator. While the Kozachenko-Leonenko estimator is asymptotically unbiased and consistent \cite{kozachenko1987sample,Leonenko2008a}, the variance and finite sample convergence rates are unknown. Hence, the null distribution in the CMI test relies on a local permutation test that is also based on nearest neighbors and data-adaptive. Here we set $k_{\rm CMI}=50$, and the nearest neighbors in the local permutation scheme, $k_{\rm perm}=5$, as well as $B=500$ as the number of surrogates for the null distribution. These choices are based on the findings in ref.~\cite{Runge2018a}.
CMI is implemented in \texttt{Tigramite}.
Alternative conditional independence tests are, e.g., kernel conditional independence tests \cite{Fukumizu2008,Zhang2012,Strobl2017}. PCMCI-CMI can be considered a \emph{doubly adaptive} causal discovery method: PCMCI adapts the condition-selection locally to the causal network and CMI adapts the conditional independence estimation locally to the sample density.

\subsection{Discrete data}
The previous tests were developed for continuously-valued data. For discrete data, the software package \texttt{Tigramite} implements a CMI test based on discrete entropy estimation, called CMIsymb. This test directly estimates CMI based on
\begin{align} 
\widehat{I}(X;Y | Z ) &=\sum_{x,y,z}  \widehat{p}(x,y,z) \log \frac{ \widehat{p}(x,y,z) \widehat{p}(z)}{\widehat{p}(x,z)\widehat{p}(y,z)}\,,
\end{align}
where the discrete densities are estimated from symbol frequencies. For this test no analytical results for the null distribution are available and a permutation test is recommended.

\subsection{Multivariate extensions}
We note that we focus on the case of univariate time series $X$ and $Y$ as is typical in causal discovery applications (while the conditioning variable $\mathbf{Z}$ can and mostly will be multivariate). However, our method can be readily extended to multivariate time series $X$ and $Y$, as well as those taking values in structured or non-Euclidean domains for suitable conditional independence tests. CMI, for example, readily extends to multivariate $X$ and $Y$. Also a two-step procedure involving vector-valued regression and kernel conditional independence tests \cite{Fukumizu2008,Zhang2012,Sejdinovic2013} can be used. In complex dynamical systems, multivariate variables can, for example, result from delay-embeddings to better unfold the dynamics of nonlinear systems \cite{Kantz2003}. Or, one can create multivariate variables from aggregating time series of multiple subprocesses to represent a dynamical process on a higher layer \cite{Chung2018}. Causal networks reconstructed on such different layers then lend themselves to analysis with measures from the growing field of multilayer networks \cite{Boccaletti2014}.


\section{Properties of PCMCI} \label{sec:mit}
In this section, we provide more formal definitions and proofs for the properties of PCMCI stated in the main article.

\subsection{Computational complexity}
PCMCI has polynomial worst case complexity in the number of variables $N$ and $\tau_{\max}$. The computational complexity of the first step of PCMCI (Algorithm~\ref{algo:pcs1}) strongly depends on the network structure and the parameter $\alpha$. The sparser the causal dependencies, the faster the convergence. In the worst case where the network is completely connected (which is rather pathological), the computational complexity of the PC$_1$ condition-selection step for $N$ variables amounts to 
\begin{align}
N \sum_{p=0}^{N \tau_{\max}-1} N \tau_{\max} = N^3 \tau_{\max}^2
\end{align}
conditional independence tests with iteratively increasing cardinality. The MCI step (Algorithm~\ref{algo:mit}) further involves $N^2 \tau_{\max}$ tests (for $\tau>0$) of a maximal dimensionality of $2+|\widehat{\mathcal{P}}(X^j_t)|+|\widehat{\mathcal{P}}(X^i_{t-\tau})|$. Hence the worst case total computational complexity in the number of variables is polynomial and given by $N^3 \tau_{\max}^2+N^2 \tau_{\max}$ if $\alpha$-optimization is not taken into account. The computational time will then depend on how the conditional independence test scales with this dimensionality and the time series length $T$. In the numerical experiments, we analyze runtimes of the different techniques for different network sizes $N$ and time series lengths $T$. 

In \texttt{Tigramite} we implement a \texttt{recycle\_residuals}-option that can be used for ParCorr and GPDC to save already computed residuals in memory to be reused in later tests of PC$_1$ or MCI.

\subsection{Consistency}
We here give a proof of the consistency of PCMCI for the population version of PCMCI, that is, PCMCI estimates the true graph in the limit of infinite sample size where there are no errors in the conditional independence tests. The proof relies on the three standard assumptions in causal discovery\cite{Spirtes2000}: (1)~\emph{Causal Sufficiency} implying that there exist no other unobserved variables that directly or indirectly influence any other pair of our set of observed variables, (2)~the \emph{Causal Markov Condition} implying that $Y_t$ is independent of $\mathbf{X}^-_t\setminus \mathcal{P}_{Y_t}$ given its parents $\mathcal{P}_{Y_t}$, and (3)~\emph{Faithfulness} which guarantees that the graph entails \emph{all} conditional independence relations that are implied by the Markov condition. Faithfulness implies that if two variables are independent conditionally on a set $\mathcal{S}$, then they are also \emph{separated} by $\mathcal{S}$ in the graph. See ref.~\cite{Runge2015} for definitions of separation in time series graphs. As part of the Causal Markov Condition in the time series graph context, we also assume \emph{no-contemporaneous causal effects} which excludes causal effects at $\tau=0$. Only then the lagged parents are sufficient for the Causal Markov condition. Last, we also assume stationarity here for all the time series considered so that dependencies (or lack of them) remain unchanged across time. With these assumptions we can prove the consistency of PCMCI. 
\begin{figure*}[tbhp]
\centering
\includegraphics[width=.3\linewidth]{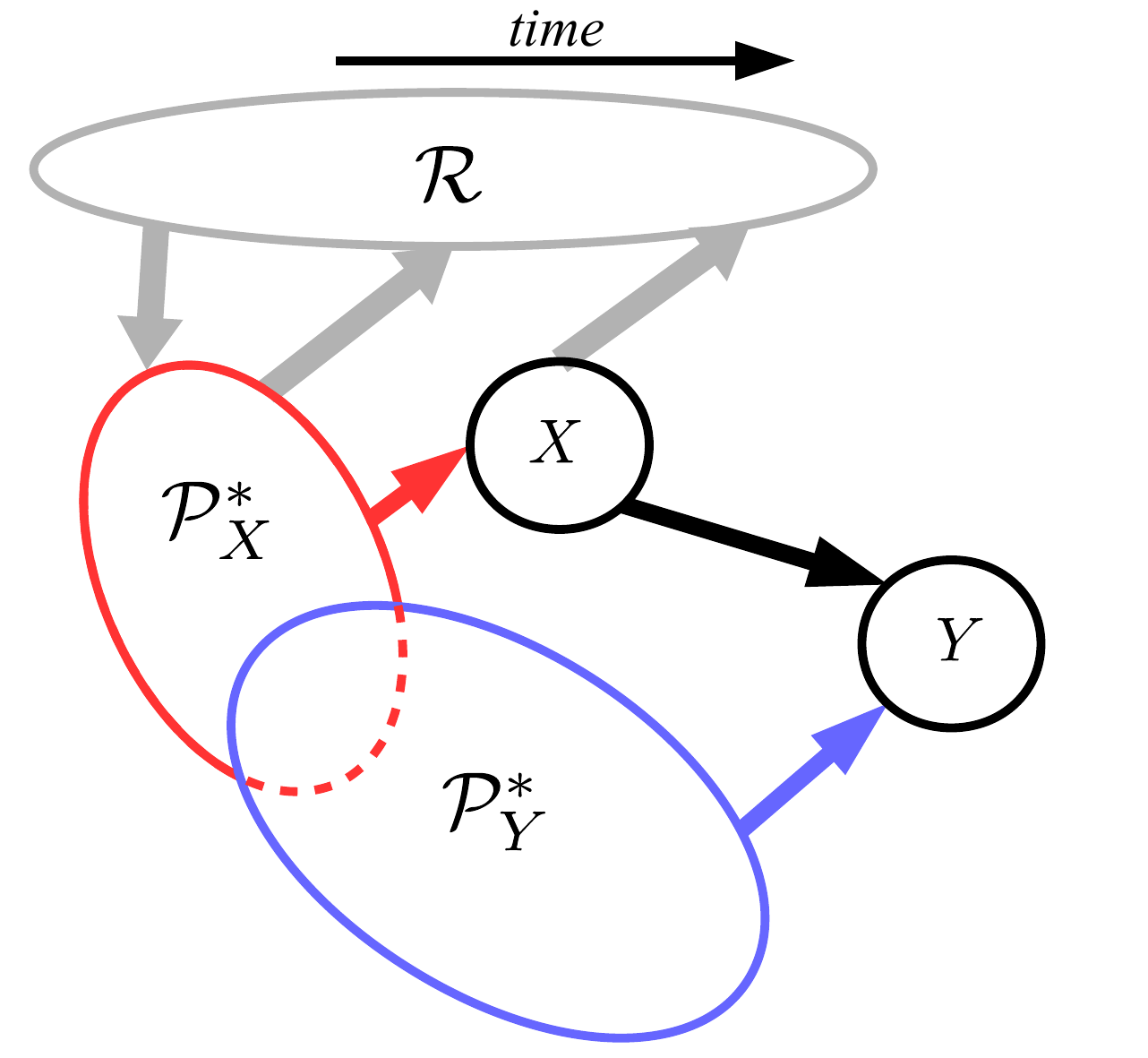}
\caption{Illustration of notation for proving Theorem~\ref{thm:pcmci_consistency} and Theorem~\ref{thm:fullci}.
}
\label{fig:proof}
\end{figure*}

\begin{mythm}{(Consistency)} \label{thm:pcmci_consistency}
Let $\mathbf{X}$ be a stochastic process with true time series graph $\mathcal{G}$ as defined in Def.~\ref{eq:def_graph} and let $\widehat{\mathcal{G}}$ be the estimated graph with PCMCI (Algorithms~\ref{algo:pcs1},\ref{algo:mit}) implemented with a consistent conditional independence test. Assuming Causal Sufficiency, Faithfulness and the Causal Markov Condition we have that 
\begin{align}
\widehat{\mathcal{G}} = \mathcal{G}\,.
\end{align}
\end{mythm}
\noindent
To prove consistency, we need the following two lemmas.
\begin{mylemma} \label{lemma:superset}
Let $\widehat{\mathcal{P}}(X^j_t)$ denote the estimated condition set of Algorithm~\ref{algo:pcs1} for $X^j\in \mathbf{X}$ and let $\mathcal{P}(X^j_t)$ denote the true parents. Assuming Faithfulness and Causal Sufficiency with a consistent conditional independence test in the limit of infinite sample size we have that
\begin{align}
\mathcal{P}(X^j_t) \subseteq \widehat{\mathcal{P}}(X^j_t) ~~~\forall j 
\end{align}
that is, the estimated parents are a superset of the true parents.
\end{mylemma}
\noindent
\begin{proof}
Suppose $X^i_{t-\tau}\notin \widehat{\mathcal{P}}(X^j_t)$. Causal Sufficiency implies that $X^i_{t-\tau}$ was observed and independence can be tested. For $p_{\max}=\infty$ (default parameter value in PC$_1$) and with a consistent conditional independence test in the limit of infinite sample size PC$_1$ removes $X^i_{t-\tau}$ from $\widehat{\mathcal{P}}(X^j_t)$ if and only if $X^i_{t-\tau}\ci X^j_t | \widehat{\mathcal{P}}(X^j_t){\setminus}\{X^i_{t-\tau}\}$ in the last step of PC$_1$. Now Faithfulness implies that then $X^i_{t-\tau}~\cancel{\to}~X^j_t$ and hence $X^i_{t-\tau}\notin\mathcal{P}(X^j_t)$.
\end{proof}
Lemma~\ref{lemma:superset} holds assuming only Causal Sufficiency and Faithfulness. If we additionally assume the Causal Markov Condition, PC$_1$ estimates exactly the true parents.
\begin{mylemma} \label{lemma:exactparents}
Let $\widehat{\mathcal{P}}(X^j_t)$ denote the estimated condition set of Algorithm~\ref{algo:pcs1} for $X^j\in \mathbf{X}$. Assuming Faithfulness, the Causal Markov Condition and Causal Sufficiency in the limit of infinite sample size we have that
\begin{align}
\widehat{\mathcal{P}}(X^j_t) = \mathcal{P}(X^j_t) ~~~\forall j 
\end{align}
that is, the estimated parents are the true parents.
\end{mylemma}
\noindent
\begin{proof}
From Lemma~\ref{lemma:superset} we know that  $\widehat{\mathcal{P}}(X^j_t)$ is a superset of $\mathcal{P}(X^j_t)$, so we only need to check whether \emph{all} parents in $\widehat{\mathcal{P}}(X^j_t)$ are also in $\mathcal{P}(X^j_t)$.
Assume the contrary that $X^i_{t-\tau}\in \widehat{\mathcal{P}}(X^j_t)$, but $X^i_{t-\tau}\notin \mathcal{P}(X^j_t)$. The contraposition of Faithfulness $X^i_{t-\tau}\in \widehat{\mathcal{P}}(X^j_t)$ implies that  $X^i_{t-\tau} ~\cancel{\ci}~ X^j_t | \widehat{\mathcal{P}}(X^j_t){\setminus}\{X^i_{t-\tau}\}$. Define $W=\widehat{\mathcal{P}}(X^j_t)\setminus\{\mathcal{P}(X^j_t),X^i_{t-\tau}\}$. The Causal Markov Condition reads $W \cup X^i_{t-\tau} ~\ci~ X^j_t | \mathcal{P}(X^j_t)$. From the weak union property of conditional independence it follows that $X^i_{t-\tau} ~\ci~ X^j_t | \mathcal{P}(X^j_t) \cup W$ which is equivalent to $X^i_{t-\tau} ~\ci~ X^j_t | \widehat{\mathcal{P}}(X^j_t){\setminus}\{X^i_{t-\tau}\}$, contrary to the assumption. Hence $\widehat{\mathcal{P}}(X^j_t)=\mathcal{P}(X^j_t)$.
\end{proof}
With these two Lemmas we can proof Theorem~\ref{thm:pcmci_consistency}.
\begin{proof}{(Theorem~\ref{thm:pcmci_consistency})}
From Lemma~\ref{lemma:exactparents} under the assumptions of Causal Sufficiency, Faithfulness, Causal Markov Condition, and with a consistent conditional independence test in the limit of infinite sample size, the first step of PCMCI estimates the true set of parents, that is $\widehat{\mathcal{P}}(X^j_t)=\mathcal{P}(X^j_t)$. The MCI test (Def.~\ref{eq:mit_test_SI}, Algorithm~\ref{algo:mit}) establishes the absence of a link, that is, $X^i_{t-\tau}\to X^j_{t} \notin \widehat{\mathcal{G}}$ if and only if
\begin{align} 
X^i_{t-\tau} ~&\ci~ X^j_{t} ~|~ \widehat{\mathcal{P}}(X^j_t)\setminus \{X^i_{t-\tau}\},\,\widehat{\mathcal{P}_{p_X}}(X^i_{t-\tau})\\
\stackrel{Lemma~\ref{lemma:exactparents}}{\Longleftrightarrow}
X^i_{t-\tau} ~&\ci~ X^j_{t} ~|~ \mathcal{P}(X^j_t)\setminus \{X^i_{t-\tau}\},\,\mathcal{P}_{p_X}(X^i_{t-\tau}).
\end{align}
We need to proof
\begin{align}
\text{1)}~~~X^i_{t-\tau}&\to X^j_{t} \notin \mathcal{G} ~~~\implies~~~X^i_{t-\tau}\to X^j_{t} \notin \widehat{\mathcal{G}} \\
\text{2)}~~~X^i_{t-\tau}&\to X^j_{t} \in \mathcal{G} ~~~\implies~~~X^i_{t-\tau}\to X^j_{t} \in \widehat{\mathcal{G}}\,.
\end{align}
Let $X=X^i_{t-\tau}$, $Y=X^j_t$, $\mathcal{P}^*_X=\mathcal{P}^*_{p_X}(X^i_{t-\tau})=\mathcal{P}_{p_X}(X^i_{t-\tau})\setminus\mathcal{P}(X^j_{t})$, $\mathcal{P}_Y=\mathcal{P}(X^j_t)$, $\mathcal{P}^*_Y=\mathcal{P}(X^j_{t})\setminus \{X^i_{t-\tau}\}$, and $\mathcal{R}=\mathbf{X}^-_t\setminus \{X_{t-\tau}, \mathcal{P}(X^j_{t}), \mathcal{P}_{p_X}(X^i_{t-\tau})\}$ for notational simplicity (see Fig.~\ref{fig:proof}). In addition to the standard assumptions of causal discovery, we will make use of the basic properties of conditional independence: Decomposition, weak union, and contraction, as well as their contrapositions\cite{Cover2006}.

Ad 1)
\begin{align}
&X\to Y \notin \mathcal{G} ~~\text{and}~~ (\mathcal{P}^*_X, \mathcal{R}) \cap \mathcal{P}_Y = \emptyset
~~\stackrel{\text{Markov}}{\Longrightarrow}~~ X, \mathcal{P}^*_X, \mathcal{R} ~\ci~ Y ~|~ \mathcal{P}_Y 
~~\stackrel{\text{Decomposition}}{\Longrightarrow}~~ X \mathcal{P}^*_X ~\ci~ Y ~|~ \mathcal{P}_Y  \\
&\stackrel{\text{Weak union}}{\Longrightarrow}~~ X ~\ci~ Y ~|~ \mathcal{P}_Y, \mathcal{P}^*_X
\end{align}
From Lemma 2 it now follows that 
\begin{align}
& X ~\ci~ Y ~|~ \mathcal{P}_Y, \mathcal{P}^*_X ~~\stackrel{\text{Lemma 2}}{\Longrightarrow}~~ X ~\ci~ Y ~|~ \widehat{\mathcal{P}}_Y, \widehat{\mathcal{P}}^*_X ~~\stackrel{\text{Def.~\ref{eq:mit_test_SI}}}{\Longrightarrow}~~X^i_{t-\tau}\to X^j_{t} \notin \widehat{\mathcal{G}}\,,
\end{align}
which proves the first part. 

Ad 2)
\begin{align}
&X\to Y \in \mathcal{G} 
~~\stackrel{\text{Def.~\ref{eq:def_graph}}}{\Longrightarrow}~~ X ~\cancel{\ci}~ Y ~|~ \mathcal{P}^*_Y, \mathcal{P}^*_X, \mathcal{R} 
~~\stackrel{\text{Contraposition of weak union}}{\Longrightarrow}~~ X \mathcal{R} ~\cancel{\ci}~ Y ~|~ \mathcal{P}^*_Y, \mathcal{P}^*_X \,.
\end{align}
Now the contraposition of contraction implies that 
\begin{align}
&  X, \mathcal{R} ~\cancel{\ci}~ Y ~|~ \mathcal{P}^*_Y, \mathcal{P}^*_X ~~\Longrightarrow~~ \text{either}~~~X ~\cancel{\ci}~ Y ~|~ \mathcal{P}^*_Y, \mathcal{P}^*_X ~~~\text{or}~~~ \mathcal{R} ~\cancel{\ci}~ Y ~|~ \mathcal{P}^*_Y, \mathcal{P}^*_X, X\,.
\end{align}
But the latter of these independence relations cannot hold since we assume the Causal Markov Condition which implies
\begin{align}
&X, \mathcal{P}^*_X, \mathcal{R} ~\ci~ Y ~|~ \mathcal{P}_Y 
~~\stackrel{\text{weak union}}{\Longrightarrow}~~ \mathcal{R}  ~\ci~ Y ~|~ \mathcal{P}^*_Y, \mathcal{P}^*_X, X\,.
\end{align}
Hence
\begin{align}
& X ~\cancel{\ci}~ Y ~|~ \mathcal{P}^*_Y, \mathcal{P}^*_X ~~\stackrel{\text{Lemma 2}}{\Longrightarrow}~~ X ~\cancel{\ci}~ Y ~|~ \widehat{\mathcal{P}}^*_Y, \widehat{\mathcal{P}}^*_X ~~\stackrel{\text{Def.~\ref{eq:mit_test_SI}}}{\Longrightarrow}~~X^i_{t-\tau}\to X^j_{t} \in \widehat{\mathcal{G}}\,,
\end{align}
which proves the second part.
\end{proof}

Note that the consistency of the population-version of PCMCI is a weaker statement than, for example, \emph{uniform consistency} which bounds the error probability as a function of the sample size $n$ giving a \emph{rate of convergence}. Robins \textit{et al.}\cite{Robins2003} showed that no uniformly consistent causal discovery technique from the class of independence-based methods \cite{Spirtes2000} exists since the convergence can always be made arbitrarily slow by a distribution that is \emph{almost unfaithful} with some dependencies made arbitrarily small. Uniform consistency can only be achieved under further assumptions that exclude these almost unfaithful dependencies\cite{Kalisch2008}. 
The causal assumptions are discussed further in ref.~\cite{Runge2018b}. 


\subsection{False positive control, effect size, and causal strength} \label{sec:mci_props}
The consistency proof does not require that the MCI test conditions on the parents $\mathcal{P}(X^i_{t-\tau})$, conditioning on $\mathcal{P}(X^j_t)$ suffices. We condition on the parents of the lagged variable for two reasons: (1)~For finite sample sizes, this approach helps to account for autocorrelation leading to correctly controlled false positive rates and (2)~the MCI test statistic value can be interpreted as a notion of \emph{causal strength} which allows to rank causal links in large-scale studies in a meaningful way. In the following we provide a mathematical intuition behind these two properties.

\paragraph{Autocorrelation and false positive control}
Conditional independence testing requires access to the \emph{null distribution}  of the test statistic under the null hypothesis of conditional independence. As described in Tab.~\ref{tab:CI_overview}, for the conditional independence tests considered in this paper the null distribution is either analytically given (ParCorr), pre-computed in advance (GPDC), or generated via a local permutation test (CMI). All three methods assume that the data for a particular test $X \ci Y ~|~ \mathbf{Z}$ is independent and identically distributed (\emph{iid}). 
Consider the simple two-variable model
\begin{align}
X_t &= a X_{t-1} + \eta^X_t \nonumber\\
Y_t &= b Y_{t-1} + c X_{t-1} + \eta^Y_t 
\end{align}
where $\eta^{X,Y}$ are \emph{iid}. For $c=0$ we have (unconditional) independence between $X$ and $Y$. But the Pearson correlation test statistic $\widehat{\rho}(X,Y)$ for this case is not distributed according to a $t$-distribution with $n-2$ degrees of freedom (Tab.~\ref{tab:CI_overview}). In fact, due to the autocorrelation between samples for $a,b>0$, the unknown true distribution has fewer degrees of freedom and will be typically wider than the assumed null distribution, leading to more false positives. The same holds for the pre-computed distribution for the distance correlation or the permutation-based distribution for mutual information.

An alternative approach is to consider a causality measure called transfer entropy (TE)\cite{Schreiber2000b}, which excludes information of the past of $Y$, defined as
\begin{align}
TE_{X\to Y}=I(X_{t-\tau};Y_t|Y_{t-1})\,,
\end{align}
if we truncate TE at lag one. For the above model for $c=0$ the TE can be simplified to $I(X_{t-\tau};Y_t|Y_{t-1})=I(X_{t-\tau};\eta^Y_t|Y_{t-1})$\cite{Runge2015}. $\eta^Y$ is \emph{iid}, but $X_{t-\tau}$ is not for $a>0$ and a permutation-based approach would still lead to false positives as analyzed further in ref.~\cite{Runge2018b}. The conditioning of the standalone PC algorithm also is based only on the parents of $Y$ and, hence, does not control false positives correctly for large autocorrelation in $X$ (see Fig.~\ref{fig:highdim_parcorr}C).

Typical remedies to account for autocorrelation are to adjust the degrees of freedom in some way, using pre-whitening, or by block-shuffling. While these approaches help to some extent for the simple bivariate case, they fail in the multivariate case that is relevant for causal discovery\cite{Runge2018b}.

Now consider the MCI test for this example which, in the ParCorr implementation, can be simplified as
\begin{align}
\rho_{X\to Y}^{\rm MCI}\left(\tau\right) &= \rho\left(X_{t-\tau};Y_{t}|\mathcal{P}\left(X_{t-\tau}\right),\mathcal{P}\left(Y_{t}\right)\backslash\left\{ X_{t-\tau}\right\} \right)\\
 &=  \rho\left(a X_{t-\tau-1} + \eta^X_{t-\tau};b Y_{t-1} + \eta^Y_t~|~Y_{t-1},X_{t-\tau-1}\right)\\
 &=  \rho\left( \eta^X_{t-\tau};\eta^Y_t~|~Y_{t-1},X_{t-\tau-1}\right)\\
 &=  \rho\left( \eta^X_{t-\tau};\eta^Y_t\right)\,.
\end{align}
Thus, the final Pearson correlation test on the residuals after regressing out the conditions only depends on the noise terms which are \emph{iid}. Therefore, the analytical null distribution for $n-2-2$ degrees of freedom is appropriate here and yields expected false positive rates. A similar reasoning holds for GPDC where also nonlinear auto-dependencies are regressed out. 

This case can be generalized to nonlinear additive models as discussed in Refs.~\cite{Runge2012b,Runge2015}, here we briefly summarize this result. 

\begin{mythm}{(MCI iid-ness)}
Assume a model with no link between $X_{t-\tau}$ and $Y_t$,
\begin{align} \label{eq:model}
X_{t-\tau} & = g_{X}\left(\mathcal{P}\left(X_{t-\tau}\right)\right)+\eta_{t-\tau}^{X} \nonumber\\
Y_{t} & = g_{Y}\left(\mathcal{P}\left(Y_{t}\right)\right)+\eta_{t}^{Y}, 
\end{align}
where $g_{X}$ and $g_{Y}$ are arbitrarily linear or nonlinear deterministic functions and the noise terms $\eta_{t-\tau}^{X},\eta^Y_t$ are \emph{iid} and we assume
\begin{align} \label{eq:nosidepaths}
\eta_{t}^{Y},\eta_{t-\tau}^{X} &~\ci~ \mathcal{P}\left(X_{t-\tau}\right),\mathcal{P}\left(Y_{t}\right)\,.
\end{align} 
Then
\begin{align}
I_{X\to Y}^{\rm MCI}\left(\tau\right) & = I\left(\eta_{t-\tau}^{X};\eta_{t}^{Y}\right)\\
 & =  0\,.
\end{align}
\end{mythm}
\begin{proof}
\begin{align}
I_{X\to Y}^{\rm MCI}\left(\tau\right) & =  I\left(X_{t-\tau};Y_{t}|\mathcal{P}\left(X_{t-\tau}\right),\mathcal{P}\left(Y_{t}\right) \right)\\
 & = I\left(g_{X}\left(\mathcal{P}\left(X_{t-\tau}\right)\right)+\eta_{t-\tau}^{X};g_{Y}\left(\mathcal{P}\left(Y_{t}\right) \right)+\eta_{t}^{Y}|\mathcal{P}\left(X_{t-\tau}\right),\mathcal{P}\left(Y_{t}\right) \right)\\
 & = I\left(\eta_{t-\tau}^{X};\eta_{t}^{Y}|\mathcal{P}\left(X_{t-\tau}\right),\mathcal{P}\left(Y_{t}\right) \right) \label{eq:mci_trans}\\
 & = I\left(\eta_{t-\tau}^{X};\eta_{t}^{Y}\right)=0 \label{eq:a}
\end{align}
where Eq.~\ref{eq:mci_trans} follows from translational invariance of CMI\cite{Cover2006} and Eq.~\eqref{eq:a} from the independence of the noise terms Eq.~\eqref{eq:nosidepaths}.
\end{proof}
Importantly, the innovation terms $\eta_{t-\tau}^{X},\eta_{t}^{Y}$ are \emph{iid}. Then the dependence of MCI only on these innovation terms implies that statistical tests on $I_{X\to Y}^{\rm MCI}\left(\tau\right)=0$ can be conducted under the \emph{iid-assumption} and the null distribution assumptions discussed above are appropriate yielding well-calibrated tests. Assumption~\eqref{eq:nosidepaths} is further discussed in ref.~\cite{Runge2015}.

Note, however, that this result is derived here for the population version of MCI and its application to empirical estimators should be considered with some caution and would rely on consistency and unbiasedness of these estimators, e.g., linear regression in ParCorr and GP in GPDC. The consistency properties of GP regression for specific classes of functions have been studied in ref.~\cite{Choi2007}. A full analysis of GPDC would require considering those learning theoretic guarantees on regression functions and how they impact the properties of the subsequent distance correlation independence test, which is beyond the scope of this work. For CMI no finite sample consistency results are available\cite{Runge2018a}.

Our numerical experiments show that the MCI test largely has the expected rate of false positives even for strongly autocorrelated and nonlinear dependencies. This  approach to avoiding time-dependence in the sample we found to outperform other remedies such as pre-whitening or block-shuffling\cite{Runge2018b}.

Also the FullCI test is essentially performed on independent samples since the condition on $\mathbf{X}^-_t\backslash\left\{ X_{t-\tau}\right\}$ removes any dependence with the past. However, in the GPDC implementation, we found inflated FPRs (Fig.~\ref{fig:highdim_nonlin}A), which is likely due to high dimensionality where the autocorrelations are not properly regressed out.

\paragraph{Causal effect size and FullCI} 
Now we turn to the dependent case where there \emph{is} a causal link between $X$ and $Y$. 
Next to the lower dimensionality of the MCI test compared to FullCI, one can prove that the MCI test statistic generally has a larger or equal effect size compared to FullCI. Let $I$ denote conditional mutual information as a general measure of dependence.
\begin{mythm}{(MCI is larger or equal than FullCI)} \label{thm:fullci}
With FullCI defined in Eq.~\ref{def:fullci} it holds that 
\begin{align} \label{eq:gc_mit}
I^{\rm FullCI}_{X\to Y}(\tau) & \leq I^{\rm MCI}_{X\to Y}(\tau)\,.
\end{align}
\end{mythm}
\begin{proof}
To simplify notation (see Fig.~\ref{fig:proof}), denote $X=X_{t-\tau}$, $Y=Y_{t}$, $\mathcal{P}_X=\mathcal{P}_{p_X}(X^i_{t-\tau})$, $\mathcal{P}^*_Y=\mathcal{P}(X^j_{t})\setminus \{X^i_{t-\tau}\}$, and $\mathcal{R}=\mathbf{X}_t^-\setminus \mathcal{P}^*_Y,\mathcal{P}_X$. 
Thus, $\mathcal{R}$ denotes the additional conditions of FullCI compared to MCI. Note that these are independent of $Y$ given $(\mathcal{P}^*_Y,\mathcal{P}_X,X)$, because $\{\mathcal{P}^*_Y,\mathcal{P}_X\} \cup X=\left\{\mathcal{P}\left(Y_{t}\right) \cup \mathcal{P}\left(X_{t-\tau}\right)\right\}\backslash\left\{ X_{t-\tau}\right\} \cup X_{t-\tau}=\mathcal{P}\left(Y_{t}\right) \cup \mathcal{P}\left(X_{t-\tau}\right)$ contains all of $Y$'s parents and by the Markov assumption $I\left(\mathcal{R};Y|\mathcal{P}^*_Y,\mathcal{P}_X,X\right)=0$. Now consider the following two different possibilities for decomposing a multivariate mutual information using the chain rule:
\begin{align}
I\left((X,\mathcal{R});Y|\mathcal{P}^*_Y,\mathcal{P}_X\right) &=  \underbrace{I\left(X;Y|\mathcal{P}^*_Y,\mathcal{P}_X\right)}_{\rm MCI} + \underbrace{I\left(\mathcal{R};Y|\mathcal{P}^*_Y,\mathcal{P}_X,X\right)}_{=0} \\
&=  \underbrace{I\left(\mathcal{R};Y|\mathcal{P}^*_Y,\mathcal{P}_X\right)}_{\geq 0} +  \underbrace{I\left(X;Y|W,\mathcal{R}\right)}_{\rm FullCI}\\
\implies& I^{\rm MCI}_{X\to Y}(\tau)=I\left(X;Y|\mathcal{P}^*_Y,\mathcal{P}_X\right) \geq I\left(X;Y|\mathcal{P}^*_Y,\mathcal{P}_X,\mathcal{R}\right)=I^{\rm FullCI}_{X\to Y}(\tau)
\end{align}
\end{proof}
FullCI and MCI are equal if the additional conditioning variables $Z$ are independent of $Y$ given $W$. Both the lower dimensionality and higher effect size are responsible for the empirically found higher power of the MCI test compared to FullCI.

\paragraph{Causal strength}
MCI's effect size is not only always larger or equal to FullCI, but also can be interpreted as a measure of causal strength.
Consider model~\eqref{eq:model} with an added dependency term of $Y$ on $X$:
\begin{align} \label{model:theorem_dep}
X_{t-\tau} & = g_{X}\left(\mathcal{P}\left(X_{t-\tau}\right)\right)+\eta_{t-\tau}^{X} \nonumber\\
Y_{t} & = g_{Y}\left(\mathcal{P}\left(Y_{t}\right)\backslash\left\{ X_{t-\tau}\right\} \right)+f\left(X_{t-\tau}\right)+\eta_{t}^{Y}\,.
\end{align}
We now investigate an information-theoretic definition of causal strength based on conditional mutual information:
\begin{align}
I\left(\eta^X_{t-\tau}; f\left(X_{t-\tau}\right)+\eta_{t}^{Y}~|~\mathcal{P}\left(X_{t-\tau}\right)\right)
\end{align}
If we had experimental access for intervening in $\eta_{t-\tau}^{X}$ at a particular time $t-\tau$, then causal strength information-theoretically quantifies how much of this momentary perturbation can be detected in $Y_t$, excluding information contained in the past. This measure directly corresponds to ``momentary'' dependence in $Y_{t}$ on $X_{t-\tau}$ that does not come through the parents of $X_{t-\tau}$.  There are several proposals for measures of causal strength, see, for example, ref.~\cite{Janzing2013}. Our definition of causal strength is based on the fundamental concept of \emph{source entropy} as further discussed in ref.~\cite{Runge2015}.

MCI for this model is an estimator of causal strength since, similar to the above proof,
\begin{align}
I_{X\to Y}^{\rm MCI}\left(\tau\right) & =  I\left(X_{t-\tau};Y_{t}|\mathcal{P}\left(X_{t-\tau}\right),\mathcal{P}\left(Y_{t}\right)\backslash\left\{ X_{t-\tau}\right\} \right)\\
 & =  I\left(g_{X}\left(\mathcal{P}\left(X_{t-\tau}\right)\right)+\eta_{t-\tau}^{X};g_{Y}\left(\mathcal{P}\left(Y_{t}\right)\backslash\left\{ X_{t-\tau}\right\} \right)+f\left(X_{t-\tau}\right)+\eta_{t}^{Y}|\mathcal{P}\left(X_{t-\tau}\right),\mathcal{P}\left(Y_{t}\right)\backslash\left\{ X_{t-\tau}\right\} \right)\\
 & =  I\left(\eta_{t-\tau}^{X};f\left(X_{t-\tau}\right)+\eta_{t}^{Y}|\mathcal{P}\left(X_{t-\tau}\right),\mathcal{P}\left(Y_{t}\right)\backslash\left\{ X_{t-\tau}\right\} \right) \\
 & = I\left(\eta_{t-\tau}^{X};f\left(X_{t-\tau}\right)+\eta_{t}^{Y}|\mathcal{P}\left(X_{t-\tau}\right)\right)\,.
\end{align}
For a linear dependence $f\left(X_{t-\tau}\right)=c X_{t-\tau}$, MCI can be further simplified:
\begin{align}
I_{X\to Y}^{\rm MCI}\left(\tau\right) &= I\left(\eta_{t-\tau}^{X};c\eta_{t-\tau}^{X} + \eta_{t}^{Y}\right) \label{eq:c}\,
\end{align}
which for partial correlation in the Gaussian case becomes
\begin{align}
\rho^{\rm MCI}_{X\to Y} &= \frac{c \sigma_X}{\sqrt{\sigma^2_Y + c^2 \sigma^2_X}}\,,
\end{align}
where $\sigma^2_{\cdot}$ now denotes the variances of the noise terms $\eta$.
Thus, for a linear additive dependency, where causal strength can be attributed to a single coefficient $c$, MCI depends only on this coefficient and on the noise terms, but not on $g_{Y},g_{X}$. MCI is then independent of dependencies due to the parents $\mathcal{P}\left(X_{t-\tau}\right)$ and $\mathcal{P}\left(Y_{t}\right)$, which could include autodependencies. A causal signal can, thus, be better detected against noise coming from confounding drivers or autocorrelation. This theoretical result is confirmed in the numerical experiments in Fig.~\ref{fig:algo_results_powerscaling}.

Pure correlation, for example, has a power that scales with the correlation coefficient as an effect size. But correlation can be very different from the causal effect, that is, from the link coefficient in a linear model. Take the following example:
\begin{align}
Z_t &= \eta^Z_t \nonumber\\
X_t &= a Z_{t-1} + \eta^X_t \nonumber\\
Y_t &= b Z_{t-2} + c X_{t-1} + \eta^X_t 
\end{align}
Here the correlation for the link $X_{t-1}\to Y_t$ is
\begin{align}
\rho(X_{t-1},Y_{t}) &= \frac{c \Gamma_X + ab \Gamma_Z}{\sqrt{\Gamma_X}\sqrt{\Gamma_Y}}\,,
\end{align}
where $\Gamma_{\cdot}$ denotes the variances. The correlation, thus, depends not only on $c$ and may even become zero depending on $a$ and $b$. The MCI partial correlation, on the other hand, estimates the causal strength given by $\frac{c \sigma_X}{\sqrt{\sigma^2_Y + c^2 \sigma^2_X}}$ as derived above. Hence, MCI depends only on the coefficient $c$ and the noise variances. This explains that PCMCI closely follows the actual causal strength seen in Fig.~\ref{fig:algo_results_powerscaling}.
On the other hand, for the nonlinear cases, there can still be various dependencies because the function $f$ mixes $\eta_{t-\tau}^{X}$ with $\mathcal{P}\left(X_{t-\tau}\right)$. As for the consistency proof given above, the results here are only derived for the population version of MCI.



\section{Numerical experiments} \label{sec:algo_model_description}

\subsection{Model setup}
To evaluate and compare different causal discovery methods, we use a model that mimics the properties of real data, but where the true underlying relationships are known. Here we model four of the major challenges of time series from complex systems such as the Earth: High-dimensionality, nonlinearity, strong autocorrelation, and time lagged causal dependencies.
Consider the following model from which we generate 20 ensemble members per number of variables $N$, number of links $L$, and coupling strength $c$:
For $i,j \in \{1,... , N\}$ we randomly choose $L$ links $i\to j$ with $i\neq j$ and generate time series according to
\begin{align} \label{eq:numericalmodel}
X_t^j &= a_j X^j_{t-1} + c\sum_i f_{i}(X^i_{t-\tau_i}) + \eta^j_t
\end{align}
for $j=1,... , N$ and where
\begin{itemize}
\item $a_j$ are uniformly randomly drawn from $\{0, 0.2, 0.4,0.6, 0.8, 0.9 \}$ for one half of the ensemble and from $\{0.6, 0.8, 0.9, 0.95\}$ for another, more autocorrelated, half of the 20 network ensemble members, except for the high-density experiments where $a_j$ are only drawn from $\{0, 0.2, 0.4,0.6, 0.8, 0.9 \}$.
\item \emph{iid} Gaussian noise $\eta^j \sim \mathcal{N}(0,1)$
\item for ParCorr experiments: $f_i(x)=f^{(1)}(x)=x$; for nonlinear model experiments of the $L$ links in each network 50\% are linear functions $f^{(1)}(x)=x$, 25\% are nonlinear $f^{(2)}(x)=(1-4 e^{-x^2/2})x$ and 25\% are nonlinear $f^{(3)}(x)=(1-4 x^3 e^{-x^2/2})x$
\item $\tau_i$ uniformly randomly drawn from $\{1,\,2\}$
\item $c$ is constant for all links in a model and its absolute value differs among the experiments (see descriptions in Tab.~\ref{tab:experiments}); the sign of $c$ is positive or negative with equal probability
\end{itemize}
To guarantee stationarity, the functions $f_i(x)$ are all linear in the limit of large $x$ and we dismiss all models for which the corresponding vector autoregressive model with nonlinear functions $f_i$ replaced by linear ones is nonstationary according to a unit root test\cite{Li2009a}. Fig.~\ref{fig:highdim_parcorr}B gives an example realization. For each number of variables $N$ and coefficient $c$ we generated 20 (if not noted otherwise) randomly drawn network topologies with $L=N$ links (except for some experiments with $L=2N$ and the bivariate case, $N=2$, with $L=1$). With $L=N$ links in each model, we have an average cross-in-degree of $1$ for all network sizes (plus an autodependency). The cross-link density, on the other hand, decays with $N$ as $\frac{N}{N(N-1)\tau_{\max}}=\frac{1}{(N-1)\tau_{\max}}$.

\subsection{Performance evaluation}
To assess false positives (FPR) and true positives (TPR) for the individual links in each model, $100$ time series realizations were generated for each model. Note that the error in the estimate of a FPR of $0.05$ (or a TPR of $0.95$) is roughly $\sqrt{0.05(1-0.05)/100}\approx 0.02$.

The bottom rows in most figures show boxplots of the distribution of FPRs and the upper row(s) of the TPR for linear (and nonlinear) dependencies. Only cross-links were considered here. As illustrated in Fig.~\ref{fig:highdim_parcorr}A, the left and right boxplots in the figures depict the distributions for all weakly autocorrelated pairs with mean autocorrelation $(\rho(X_{t-1},X_t)+\rho(Y_{t-1},Y_t))/2<0.7$ among the two variables $X$ and $Y$ of a link, and for strongly autocorrelated pairs ($(\rho(X_{t-1},X_t)+\rho(Y_{t-1},Y_t))/2\geq 0.7$), respectively.  
The boxes show the 25-75\% and whiskers the 1-99\% percentile range, the median is marked by a bar and the mean with `x'. Note the logarithmic y-axis in the bottom panel for FPR $> 0.1$.

The tick labels on the top of the figures note the average runtime and its standard deviation across the different model setups. 
The runtime estimates were evaluated on Intel Xeon E5-2667 v3 8C processors with 3.2GHz. These runtimes will depend on implementation.


In Tab.~\ref{tab:experiments} we list the model setups for the numerical experiments. Table~\ref{tab:methods} gives details on the compared methods. The experiments were evaluated on a high-performance cluster.
\section{Tigramite software package} \label{sec:tigramite}

PCMCI is implemented in the \verb|Tigramite| software package (current version 3.0).  \verb|Tigramite| is a time series analysis python module for linear and nonlinear causal inference available from \verb|https://github.com/jakobrunge/tigramite|. 
\verb|Tigramite| contains classes for PCMCI and the different conditional independence tests, as well as a module that contains several plotting functions to generate high-quality plots of time series, lag functions, and causal graphs as depicted in Fig.~\ref{fig:highdim_parcorr}A. 
Documentation can be found on the repository site.

\onecolumn

\clearpage
\section{Algorithms}

\begin{algorithm}[H]
\caption{
Pseudo-code for condition-selection algorithm to estimate parents of $X^j_t$; we use this algorithm as a pre-selection step in PCMCI with $p_{\max}=N \tau_{\max}$ (i.e., no restriction on the maximum number of parents) and $q_{\max}=1$; for the standalone PC-stable algorithm, we set $q_{\max}$ to a large value of $10$.}
\begin{algorithmic}[1]
\Require Time series dataset $\mathbf{X}=(X^1,\,X^2,\ldots,X^N)$, selected variable $X^j$, maximum time lag $\tau_{\max}$, significance threshold $\alpha$, maximum condition dimension $p_{\max}$ (default $p_{\max}=N \tau_{\max}$), maximum number of combinations $q_{\max}$, conditional independence test function
\Function{CI}{$X,\,Y,\,\mathbf{Z}$} 
    \State Test $X ~\ci~ Y ~|~ \mathbf{Z}$ using test statistic measure $I$
    \State \Return $p$-value, test statistic value $I$
\EndFunction
\State Initialize preliminary set of parents $\widehat{\mathcal{P}}(X^j_t)=\{ X^i_{t-\tau}: i =1,\ldots,N,~~\tau=1,\ldots,\tau_{\max} \}$
\State Initialize dictionary of test statistic values  $I^{\min}(X^i_{t-\tau}\to X^j_t)=\infty ~~\forall~ X^i_{t-\tau} \in \widehat{\mathcal{P}}(X^j_t)$
\For {$p=0,\ldots,p_{\max}$ }
    \If{$|\widehat{\mathcal{P}}(X^j_t)|-1<p$}
        \State Break for-loop
    \EndIf
    \ForAll{$X^i_{t-\tau}$ in $\widehat{\mathcal{P}}(X^j_t)$} 
        \State $q=-1$
        \ForAll{lexicographically chosen $\mathcal{S}\subseteq\widehat{\mathcal{P}}(X^j_t)\setminus \{X^i_{t-\tau}\}$ with $|\mathcal{S}|=p$}
            \State $q=q+1$
            \If{$q\geq q_{\max}$}
                \State Break from inner for-loop
            \EndIf
            \State Run CI test to obtain $(\text{$p$-value},\,I) \gets$ \Call{CI}{$X^i_{t-\tau},\,X^j_{t},\,\mathcal{S}$}

            \If{$|I| < I^{\min}(X^i_{t-\tau}\to X^j_t)$}   \Comment Store minimum $I$ of parent among all tests until now
                \State $I^{\min}(X^i_{t-\tau}\to X^j_t) = |I|$
            \EndIf
            \If{$p$-value $> \alpha$} \Comment Removed only after all $X^i_{t-\tau}$ have been tested
                \State Mark $X^i_{t-\tau}$ for removal from $\widehat{\mathcal{P}}(X^j_t)$ 
                \State Break from inner for-loop
            \EndIf
            
        \EndFor
    \EndFor
    \State Remove non-significant parents from $\widehat{\mathcal{P}}(X^j_t)$
    \State Sort parents in $\widehat{\mathcal{P}}(X^j_t)$ by $I^{\min}(X^i_{t-\tau}\to X^j_t)$ from largest to smallest
\EndFor
\State \Return $\widehat{\mathcal{P}}(X^j_t)$
\end{algorithmic}
 \label{algo:pcs1}
\end{algorithm}

\begin{algorithm}[H]
\caption{
Pseudo-code for MCI causal discovery step. Here we state the algorithm for $\tau\geq0$, then causal links for $\tau=0$ correspond to contemporaneous links, which are left undirected here.}
\begin{algorithmic}[1]
\Require Time series dataset $\mathbf{X}=(X^1,\,X^2,\ldots,X^N)$, significance level $\alpha$, sorted parents $\widehat{\mathcal{P}}(X^j_t)$ for all variables $X^j$ estimated with Algorithm~\ref{algo:pcs1}, maximum time lag $\tau_{\max}$, maximum number $p_{X}$ of parents of variable $X^i$
\ForAll{$(X^i_{t-\tau}, X^j_t)$ with $i =1,\ldots,N$ and  $\tau=0,\ldots,\tau_{\max}$, excluding $(X^j_{t}, X^j_t)$ }
        \State Remove $X^i_{t-\tau}$ from $\widehat{\mathcal{P}}(X^j_t)$  if necessary
        \State Define $\widehat{\mathcal{P}}_{p_X}(X^i_{t-\tau})$ as the first $p_X$ parents from $\widehat{\mathcal{P}}(X^i_t)$, shifted by $\tau$
        \State Run MCI test to obtain $(\text{$p$-value},\,I) \gets$ \Call{CI}{$X^i_{t-\tau},\,X^j_{t},\,\mathbf{Z}=\{\widehat{\mathcal{P}}(X^j_t),\, \widehat{\mathcal{P}}_{p_X}(X^i_{t-\tau})\}$} 
\EndFor
\State Optionally adjust $p$-values of all links by False Discovery Rate-approach with significance level $\alpha$ 
\State \Return $p$-values or $q$-values (for FDR-adjusted tests) and MCI test statistic values
\end{algorithmic}
 \label{algo:mit}
\end{algorithm}

\clearpage
\begin{algorithm}[b]
\caption{Adaptive lasso regression for inference of non-zero coefficients and p-values for model $Y=\mathbf{X}\beta$. Lasso  implemented in \texttt{sklearn} package \texttt{LassoCV} with default parameters (except \texttt{fit\_intercept=False}) and  with $\lambda_n$ chosen by cross-validation using \texttt{TimeSeriesSplit(n\_splits=5)}. In the numerical experiments we used $k_{\max}=5$. The first part is adapted from \texttt{gist.github.com/agramfort/1610922}.}
\begin{algorithmic}[1]
\Require Data $\mathbf{x}\in \mathbb{R}^{n\times d}$, $y\in \mathbb{R}^{n}$, maximum number $k_{\max}$ of iterations
\State Standardize $\mathbf{x}$ and $y$
\State Initialize weights $w_j = 1$ for $j=1,\ldots,d$
\For {$k=1,\ldots,k_{\max}$ }
    \State Scale features $\mathbf{x}_j^* = \mathbf{x}_j/w_j$ for $j=1,\ldots,d$
    \State Solve Lasso problem with $\lambda_n$ chosen by time series based cross-validation
    \begin{align}
        \mathbf{\beta}^* &= \argmin_{\mathbf{\beta}} \left\Vert y-\sum_{j=1}^d \mathbf{x}_j^* \beta_j  \right\Vert^2 + \lambda_n  \sum_{j=1}^d |\beta_j|
    \end{align}
    \State Re-weight coefficients $\beta^{**}_j = \beta^*_j/w_j$ for $j=1,\ldots,d$
    \State Compute new weights $w_j = 1/(2|\beta^{**}_j|^{\frac{1}{2}}+\epsilon)$ for $j=1,\ldots,d$, where $\epsilon$ is the machine limit for floats
\EndFor
\State Define active set as  $\mathcal{A}=\{j: \beta^{**}_j\neq 0\}$
\State Solve OLS regression on active set
\begin{align}
    \tilde{\mathbf{\beta}} &= \argmin_{\mathbf{\beta}} \left\Vert y-\sum_{j\in \mathcal{A}} \mathbf{x}_j \beta_j  \right\Vert^2
\end{align}
and record corresponding p-values $\tilde{p}_j$ for $j\in \mathcal{A}$
\State Define p-values
\begin{align}
    p_j &=  \begin{cases}
               \tilde{p}_j        & \text{if } j\in \mathcal{A} \\
               1        & \text{otherwise}
  \end{cases}
\end{align}
\Return $p_j$ for $j=1,\ldots,d$
\end{algorithmic}
 \label{algo:stabs_lasso}
\end{algorithm}

\clearpage

\section{~Supplementary Tables}
\begin{table}[tbh]
\centering
\caption{Overview over conditional independence tests for $X \ci Y ~|~ \mathbf{Z}$ considered in this paper. The tests are discussed in Sect.~\ref{sec:ci_tests}. All tests assume continuously-valued data. In the implementations, all data are standardized. The Gaussian process (GP) was fitted with \texttt{sklearn}`s \texttt{GaussianProcessRegressor} with \texttt{kernel=RBF()+WhiteKernel()} and \texttt{alpha=0}. The bandwidth of the Kernel in \texttt{sklearn} is estimated by maximizing marginal likelihood (ML-II). $D_Z$ is the cardinality of $\mathbf{Z}$.}
\begin{tabular}{c|ccc}
              &  ParCorr      &  GPDC         & CMI \\
\midrule
\midrule
Assumed model & \makecell{$\begin{aligned}[t] 
X &= \mathbf{Z} \beta_X + \epsilon^X \\
Y &= \mathbf{Z}  \beta_Y + \epsilon^Y \\
&\epsilon^{\cdot} \sim \mathcal{N}(0,\sigma^2_{\cdot})
\end{aligned}$} &\makecell{  
$\begin{aligned}[t] 
X &= h_X(\mathbf{Z}) + \epsilon^X \\
Y &= h_Y(\mathbf{Z}) + \epsilon^Y\\
&\epsilon^{\cdot} \sim \mathcal{N}(0,\sigma^2_{\cdot})
\end{aligned}$} &  \makecell{No parametric assumptions,\\ direct estimation of CMI \\$I(X;Y  | \mathbf{Z})$} \\
\midrule
Estimation & \makecell{Get residuals from OLS fit\\ 
$\begin{aligned}[t] 
\widehat{r_X} &= X - \mathbf{Z} \widehat{\beta_X}\\
\widehat{r_Y} &= Y - \mathbf{Z} \widehat{\beta_Y}\\
\end{aligned}$ \\ 
Estimate correlation $\widehat{\rho}(\widehat{r_X},\widehat{r_Y})$
} &
\makecell{Fit $\widehat{h_X}$, $\widehat{h_Y}$ with GP, get residuals\\ 
$\begin{aligned}[t] 
\widehat{r_X} &= X - \widehat{h_X}(\mathbf{Z}) \\
\widehat{r_Y} &= Y - \widehat{h_Y}(\mathbf{Z}) \\
\end{aligned}$ \\
Transform $\widehat{r_X},\widehat{r_Y}$ to\\ uniform marginals (copula), estimate \\distance correlation $\widehat{\operatorname{dCor}}(\widehat{r_X},\widehat{r_Y})$
  }
 & \makecell{Nearest-neighbor\\ estimator\cite{FrenzelPompe2007}} \\
 \midrule
 Parameter(s)  & --  &   \makecell{MLE estimation with \\ Radial Basis Function+White kernel,\\no parameters for dCor} &  Nearest neighbors $k_{\rm CMI}=50$ \\
 \midrule
Null distribution & \makecell{Analytically known, \\$t=\widehat{\rho}(\widehat{r_X},\widehat{r_Y})\sqrt{\frac{n-2-D_Z}{1-\widehat{\rho}(\widehat{r_X},\widehat{r_Y})^2}}$ \\follows $t$-distribution with \\$n-2-D_Z$ degrees of freedom} &
\makecell{Pre-computed for each sample size $n$\\from $\widehat{\operatorname{dCor}}(u,v)$ with $u,v\sim U(0,1)$\\(valid due to copula transform)} 
& \makecell{Local permutation test\cite{Runge2018a}\\ with $k_{\rm perm}=5$} \\
\bottomrule
\bottomrule
\end{tabular}
\label{tab:CI_overview}
\end{table}

\clearpage

\begin{table}[tbhp]
\centering
\caption{Model configurations for different experiments. The model is given in Eq.~\eqref{eq:numericalmodel}. For each configuration, $100$ time series realizations were generated to evaluate false and true positives. The coupling functions are $f^{(1)}(x)=x$, $f^{(2)}(x)=(1-4 e^{-x^2/2})x$, $f^{(3)}(x)=(1-4 x^3 e^{-x^2/2})x$.}
\begin{tabular}{c|ccccc}
Experiment &  \makecell{Variables $N$\\and links $L$}  & Functions $f_i$ & Sample size $T$   &  Coefficient $c$   &  \makecell{Number of\\random networks} \\
\midrule
\midrule
\makecell{High-dimensionality\\ ParCorr \\ Figs.~\ref{fig:highdim_parcorr}C,\ref{fig:algo_par_corr_allvsmit_SI}}  &  \makecell{$N=2,5,10,20,40,$\\ $60,80,100$\\ $L=N$} & $f^{(1)}$  & $150$  & $0.287$  &  $20$ per $N$ \\ 
\midrule
\makecell{High-dimensionality\\ ParCorr \\ Fig.~\ref{fig:algo_par_corr_allvsmit_SI_300}}  &  \makecell{$N=2,5,10,20,40,$\\ $60,80,100$\\$L=N$} & $f^{(1)}$  & $300$  & $0.2$  &  $20$ per $N$ \\ 
\midrule
\makecell{High-density\\ ParCorr\\Fig.~\ref{fig:algo_results_degree_SI}}  & \makecell{$N=10,20,40,60$\\ $L=2N$} & $f^{(1)}$  & $150$   &  $0.287$   &  $10$ \\ 
\midrule
\makecell{High-density\\ ParCorr\\Fig.~\ref{fig:algo_results_degree_300_SI}}  & \makecell{$N=10,20,40,60$\\ $L=2N$} & $f^{(1)}$  & $300$   &  $0.2$   &  $10$ \\ 
\midrule
\makecell{Sample size\\ ParCorr\\Fig.~\ref{fig:algo_results_samplesize_SI}}  & $N=20$, $L=N$ & $f^{(1)}$  & $150,\,300,\,600$   &  $0.200$ &  $20$ \\ 
\midrule
\makecell{Causal strength\\ ParCorr\\Fig.~\ref{fig:algo_results_powerscaling}}  & $N=20$, $L=N$& $f^{(1)}$  & $150$   & \makecell{ $0.200,\,0.247,$\\ $0.287,\,0.324,$\\ $0.414$} &  $20$ per $c$ \\ 
\midrule
\makecell{Observational noise\\ ParCorr\\Fig.~\ref{fig:algo_results_noise_SI}}  & \makecell{$N=20$, $L=N$\\ Noise $\mathcal{N}\left(0, \sigma^2\right)$ with \\ $\sigma=0,\,0.1,\,0.25,\,0.5,\,1,\,2$}  & $f^{(1)}$  & $150$   &  $0.287$   &  $20$ \\ 
\midrule
\makecell{Parameter PC$_1$ $\alpha$\\ ParCorr\\ Fig.~\ref{fig:algo_results_pcthres_SI}A}  & $N=20$, $L=N$ & $f^{(1)}$  & $150$  &   $0.287$   &  $20$ \\ 
\midrule
\makecell{Parameter PC$_1$ $\alpha$\\ GPDC\\ Fig.~\ref{fig:algo_results_pcthres_SI}B}  &  $N=10$, $L=N$ & \makecell{50\% $f^{(1)}$\\ 25\% $f^{(2)}$ \\ 25\% $f^{(3)}$}   & $250$   &   $0.287$ &  $20$ \\ 
\midrule
\makecell{Parameter PC$_1$ $\alpha$\\ CMI\\ Fig.~\ref{fig:algo_results_pcthres_SI}C}  &  $N=5$, $L=N$ & \makecell{50\% $f^{(1)}$\\ 25\% $f^{(2)}$ \\ 25\% $f^{(3)}$}   &  $500$   &   $0.324$   &  $20$ \\ 
\midrule
\makecell{High-dimensionality\\ GPDC \\Figs.~\ref{fig:highdim_nonlin}A,\ref{fig:algo_gpdc_allvsmit_SI}}  & \makecell{$N=2,5,10,20,40$\\$L=N$} & \makecell{50\% $f^{(1)}$\\25\% $f^{(2)}$\\ 25\% $f^{(3)}$ }  & $250$   &  $0.287$   &  $20$ per $N$\\ 
\midrule
\makecell{Sample size\\ GPDC\\Fig.~\ref{fig:algo_results_samplesize_SI_gpdc}}  & $N=10$, $L=N$ & \makecell{50\% $f^{(1)}$\\25\% $f^{(2)}$\\ 25\% $f^{(3)}$ }  & $250,\,500$   &  $0.200$ &  $20$ \\ 
\midrule
\makecell{High-dimensionality\\ CMI\\Fig.~\ref{fig:highdim_nonlin}B,\ref{fig:algo_cmiknn_allvsmit_SI}}  &  \makecell{$N=2,5,10$\\ $L=N$} & \makecell{50\% $f^{(1)}$\\ 25\% $f^{(2)}$\\ 25\% $f^{(3)}$ }  & $500$ &  $0.324$ &  $20$ per $N$\\ 
\midrule
\makecell{Sample size\\ CMI\\Fig.~\ref{fig:algo_results_samplesize_SI_cmi}}  & $N=5$, $L=N$ & \makecell{50\% $f^{(1)}$\\25\% $f^{(2)}$\\ 25\% $f^{(3)}$ }  & $500,\,1000$   &  $0.324$ &  $20$ \\ 
\bottomrule
\end{tabular}
\label{tab:experiments}
\end{table}

\begin{table}[tbhp]
\centering
\caption{Overview over methods compared in numerical experiments. All methods are run with $\tau_{\max}=5$.}
\begin{tabular}{c|ccccc}
Acronym & Method & Details  \\
\midrule
\midrule
\makecell{Corr / dCor / MI} & Pairwise unconditional independence tests  & see Sect.~\ref{sec:corr} \\\midrule
\makecell{BivCI\\ ParCorr} & \makecell{Bivariate conditional independence tests\\equivalent to bivariate transfer entropy}  & \makecell{condition only on past of response variable,\\ see Sect.~\ref{sec:biv}} \\
\midrule
\makecell{FullCI\\ ParCorr} & \makecell{Vector-autoregressive model\\(in Fig.~\ref{fig:algo_results_powerscaling} with partial correlation)}    & \makecell{fit with OLS in \texttt{statsmodels},\\ see Sect.~\ref{sec:fullci}} \\
\midrule
\makecell{FullCI \\ GPDC / CMI} & \makecell{Independence test conditioning\\ on full past}   & see Sect.~\ref{sec:fullci}   \\
\midrule
Lasso & Adaptive Lasso regression   &  see Sect.~\ref{sec:lasso} and Algorithm~\ref{algo:stabs_lasso}  \\  
\midrule
\makecell{PC\\ ParCorr / GPDC / CMI} & Standalone PC algorithm    & \makecell{see Sect.~\ref{sec:pcalgo} and Algorithm~\ref{algo:pcs1} \\ $\alpha=0.2$, $q_{\max}=10$}\\
\midrule
\makecell{PC$_1$\\ ParCorr} & \makecell{Condition-selection step\\ as standalone}    &\makecell{ see Sect.~\ref{sec:pcmci} and Algorithm~\ref{algo:pcs1} \\ $\alpha=\{0.1,\, 0.2,\, 0.3,\, 0.4\}$ via AIC, $q_{\max}=1$} \\  
\midrule
\makecell{PC$_1$+MCI$_{all}$ \\ParCorr} & PCMCI with $\alpha$-optimization   & \makecell{see Sect.~\ref{sec:pcmci} and Algorithms~\ref{algo:pcs1},\ref{algo:mit}\\ $\alpha=\{0.1,\, 0.2,\, 0.3,\, 0.4\}$  via AIC, \\$q_{\max}=1$, $p_X$ unrestricted} \\
\midrule
\makecell{PC$^\alpha_1$+MCI$_{all}$ \\ ParCorr / GPDC / CMI} & PCMCI without $\alpha$-optimization   & \makecell{see Sect.~\ref{sec:pcmci} and Algorithms~\ref{algo:pcs1},\ref{algo:mit}\\$\alpha=0.2$ (or given), $q_{\max}=1$, $p_X$ unrestricted }\\ 
\midrule
\makecell{PC$_1$+MCI$_3$\\ ParCorr} & \makecell{PCMCI with $\alpha$-optimization\\ and truncated $\mathcal{P}_X$}    & \makecell{see Sect.~\ref{sec:pcmci} and Algorithms~\ref{algo:pcs1},\ref{algo:mit}\\$\alpha=\{0.1,\, 0.2,\, 0.3,\, 0.4\}$ via AIC,\\ $q_{\max}=1$, $p_X=3$} \\ 
\midrule
\makecell{PC$^\alpha_1$+MCI$_3$\\ GPDC / CMI} & \makecell{PCMCI without $\alpha$-optimization\\ and truncated $\mathcal{P}_X$ }   & \makecell{see Sect.~\ref{sec:pcmci} and Algorithms~\ref{algo:pcs1},\ref{algo:mit}\\$\alpha=0.2$, $q_{\max}=1$, $p_X=3$} \\ 
\midrule
\makecell{PC$_1$+MCI$_{0}$\\ ParCorr }& \makecell{PCMCI with $\alpha$-optimization\\ and no condition on $\mathcal{P}_X$}  & \makecell{see Sect.~\ref{sec:ity} and  Algorithms~\ref{algo:pcs1},\ref{algo:mit}\\$\alpha=\{0.1,\, 0.2,\, 0.3,\, 0.4\}$ via AIC,\\ $q_{\max}=1$, $p_X=0$} \\
\midrule
\makecell{PC$_1$+MCI$_{0}$pw\\ ParCorr} & \makecell{PCMCI with $\alpha$-optimization\\ and no condition on $\mathcal{P}_X$\\ and pre-whitening}    & \makecell{see Sect.~\ref{sec:ity} and  Algorithms~\ref{algo:pcs1},\ref{algo:mit}\\$\alpha=\{0.1,\, 0.2,\, 0.3,\, 0.4\}$ via AIC, \\ $q_{\max}=1$, $p_X=0$} \\
\bottomrule
\end{tabular}
\label{tab:methods}
\end{table}


\clearpage
\section{~Supplementary Figures}
\begin{figure*}[tbhp]
\centering
\includegraphics[width=.48\linewidth]{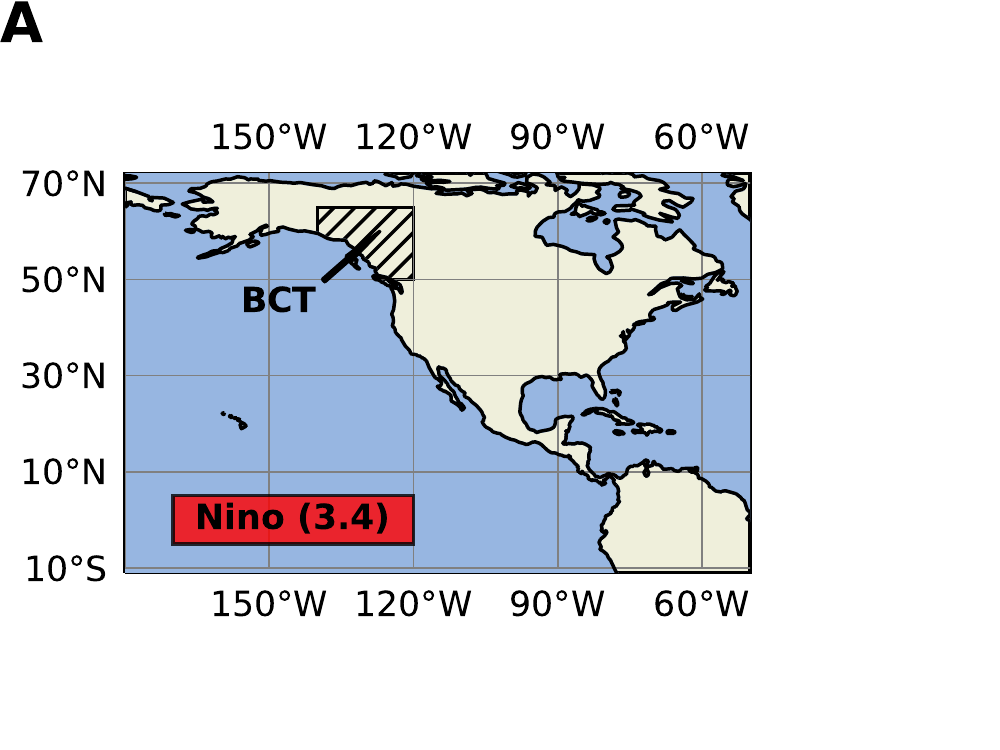}%
\includegraphics[width=.48\linewidth]{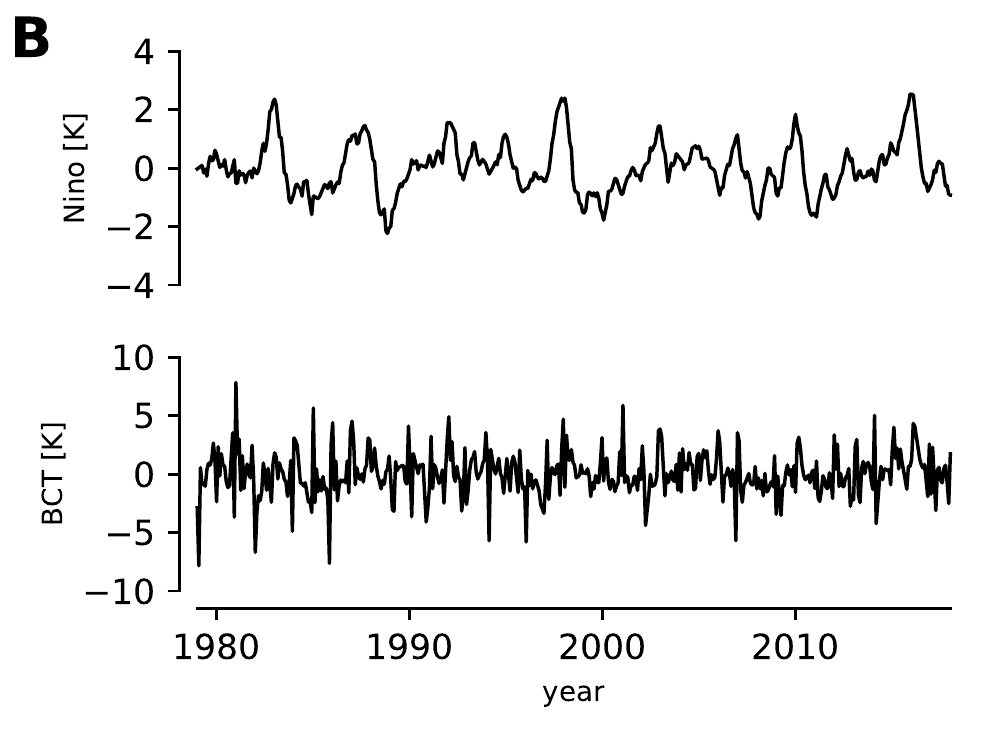}
\includegraphics[width=.48\linewidth]{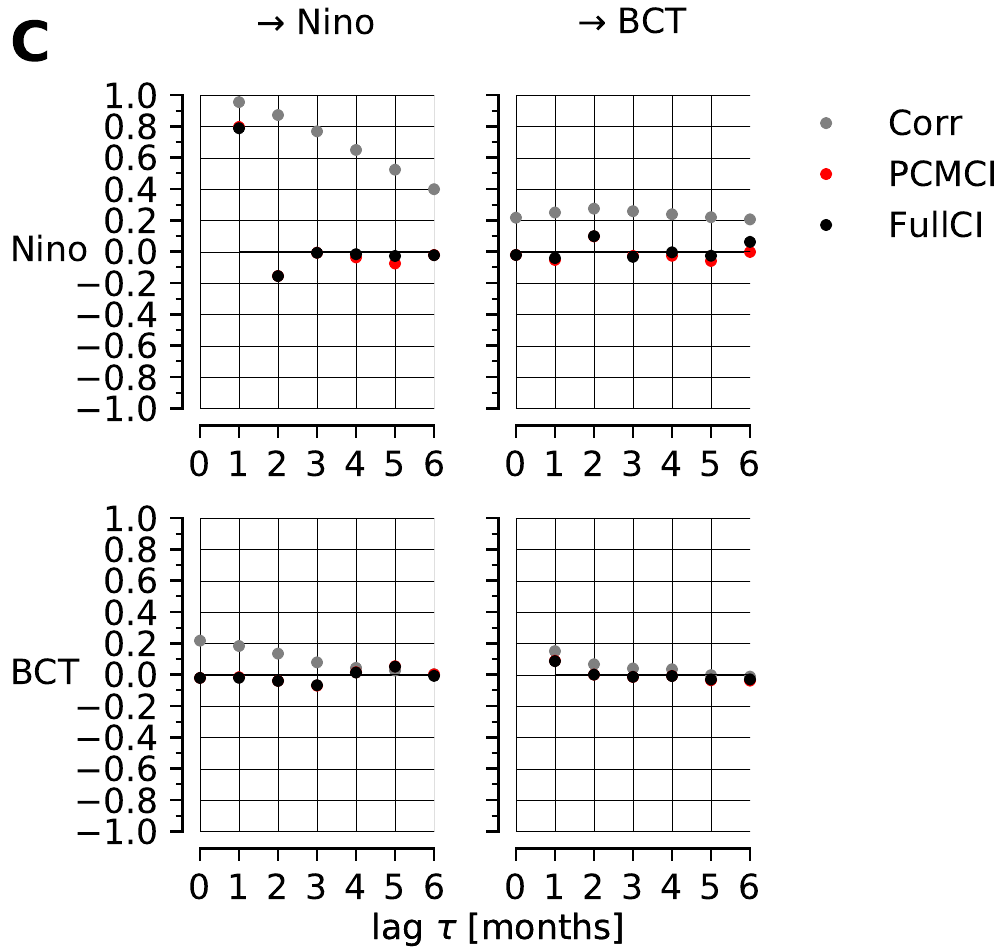}%
\includegraphics[width=.48\linewidth]{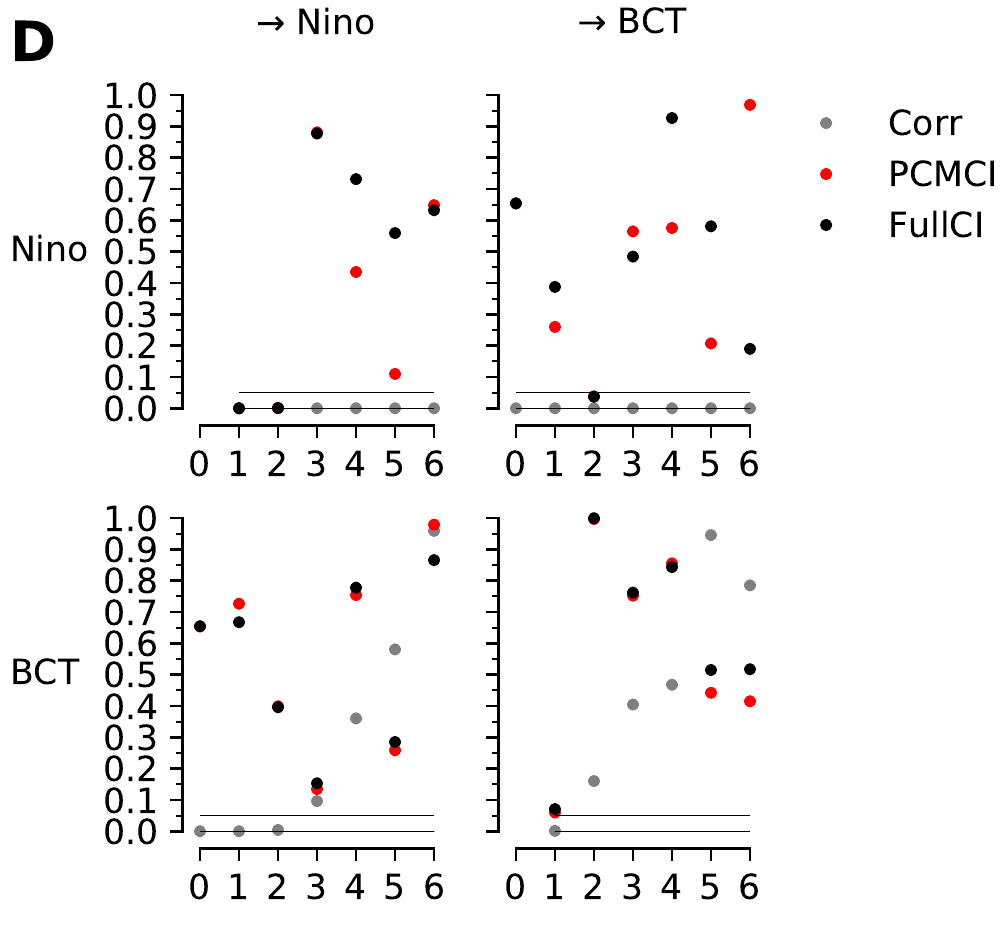}
\caption{
Motivational climate example.
(\textbf{A}) We investigate the relationship between the monthly climate index Nino and land temperature anomalies over Northwestern Canada, mostly British Columbia (BCT, hatched region). Nino is defined as the average sea-surface temperature anomaly (HadISST dataset\cite{Rayner2003}) over the red Nino3.4 region (5$^\circ$North-5$^\circ$South and 170-120$^\circ$West). BCT is defined as the area-weighted land surface temperature (CRUTEM4 dataset\cite{Jones2012}) over British Columbia and parts of Yukon and the Northwestern Territories, Canada (50-65$^\circ$North and 120-140$^\circ$West). The grid location 62.5$^\circ$North, 132.5$^\circ$West was excluded since more than 1\% of the samples where missing. Anomalized time series have the seasonal cycle removed. We constrain our analysis to the period with reliable satellite data (1979--2017) with a length of $T=468$ months. To remove any long-term temperature trend, a Gaussian kernel smoothing mean with a bandwidth of $\sigma=120$ months was removed from the raw time series.
(\textbf{B}) Time series of Nino and BCT. 
(\textbf{C}) Matrix of lag functions between Nino and BCT for Correlation (Corr), PCMCI, and FullCI (conditional on the whole past of both time series up to $\tau_{\max}=6$). Note, that for autocorrelations (on the diagonal) the zero-lag is not drawn. 
(\textbf{D}) Matrix of p-values. The black line denotes the 5\% significance level. Note that the sample size here is $n=T-2 \tau_{\max} = 456$.
}
\label{fig:example_map}
\end{figure*}

\clearpage
\begin{figure*}[tbhp]
\centering
\includegraphics[width=1.\linewidth]{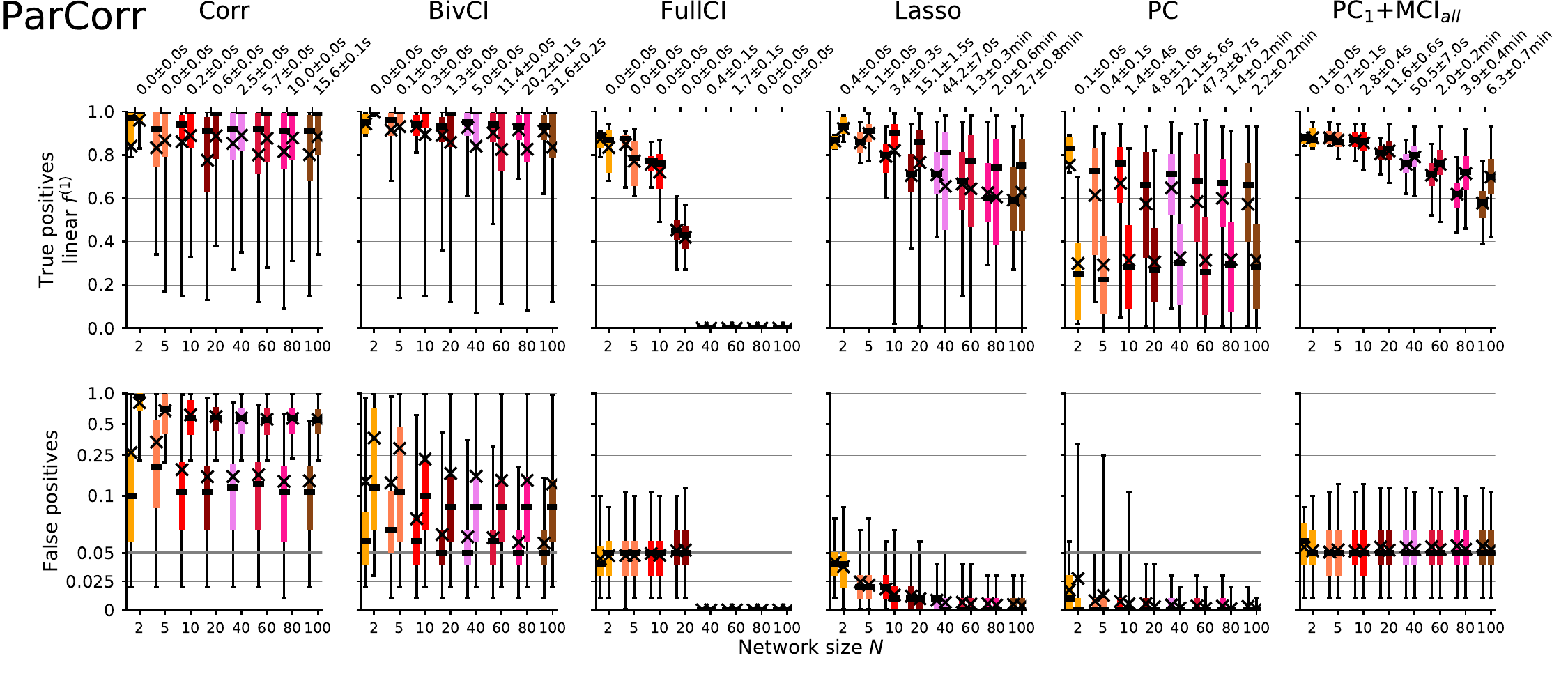}
\includegraphics[width=1.\linewidth]{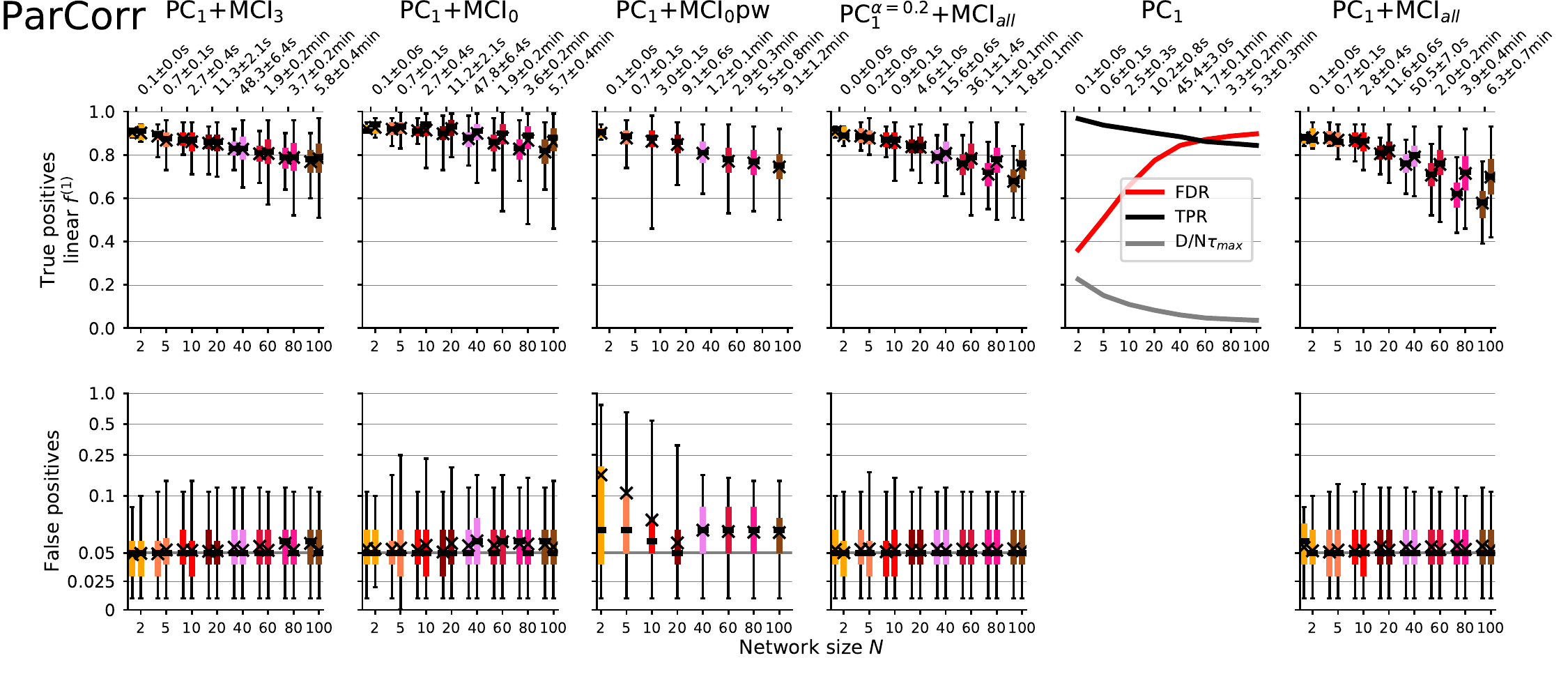}
\caption{Numerical experiments for linear models with different numbers of variables $N$, number of links $L=N$, and time series length $T=150$. The detailed setup is listed in Tab.~\ref{tab:experiments} and Tab.~\ref{tab:methods} provides details on the evaluated methods.
%
}
\label{fig:algo_par_corr_allvsmit_SI}
\end{figure*}

\clearpage
\begin{figure*}[tbhp]
\centering
\includegraphics[width=1.\linewidth]{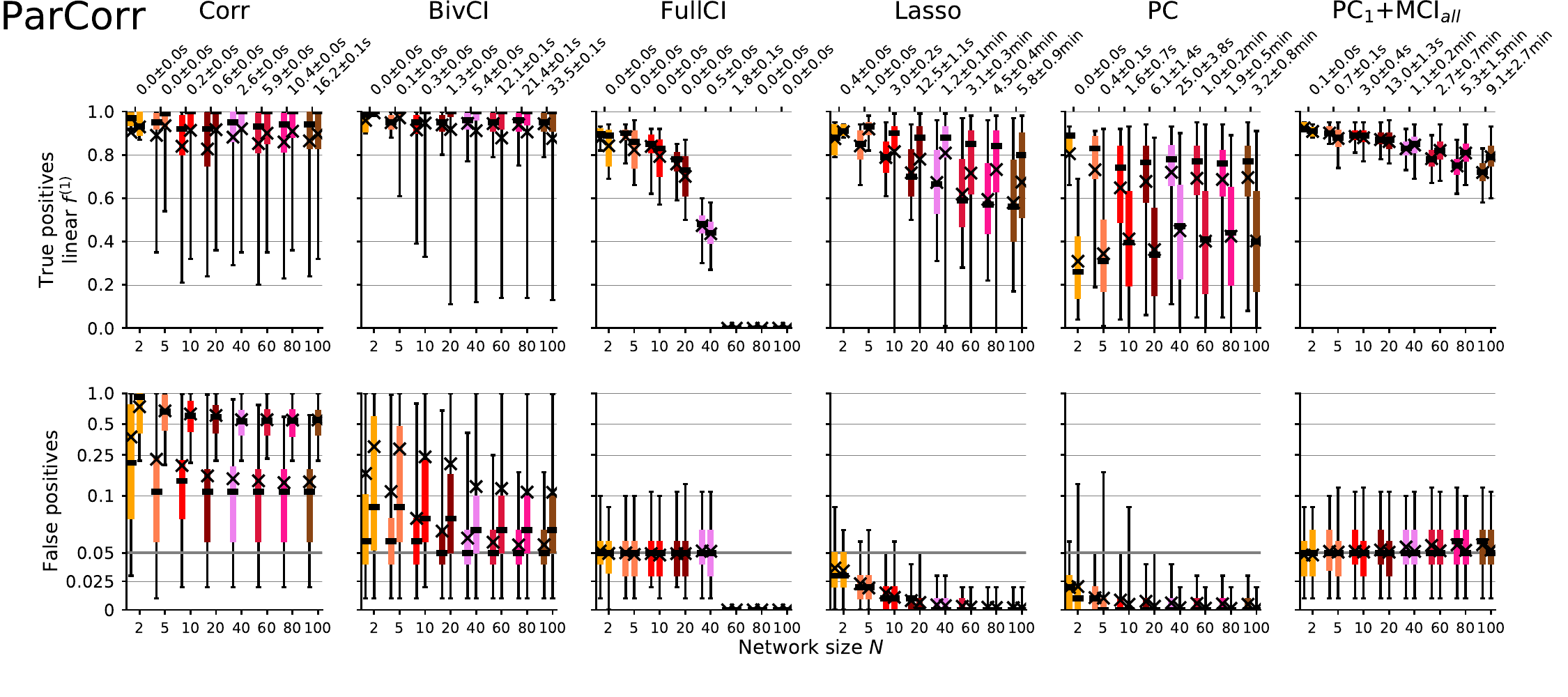}
\includegraphics[width=1.\linewidth]{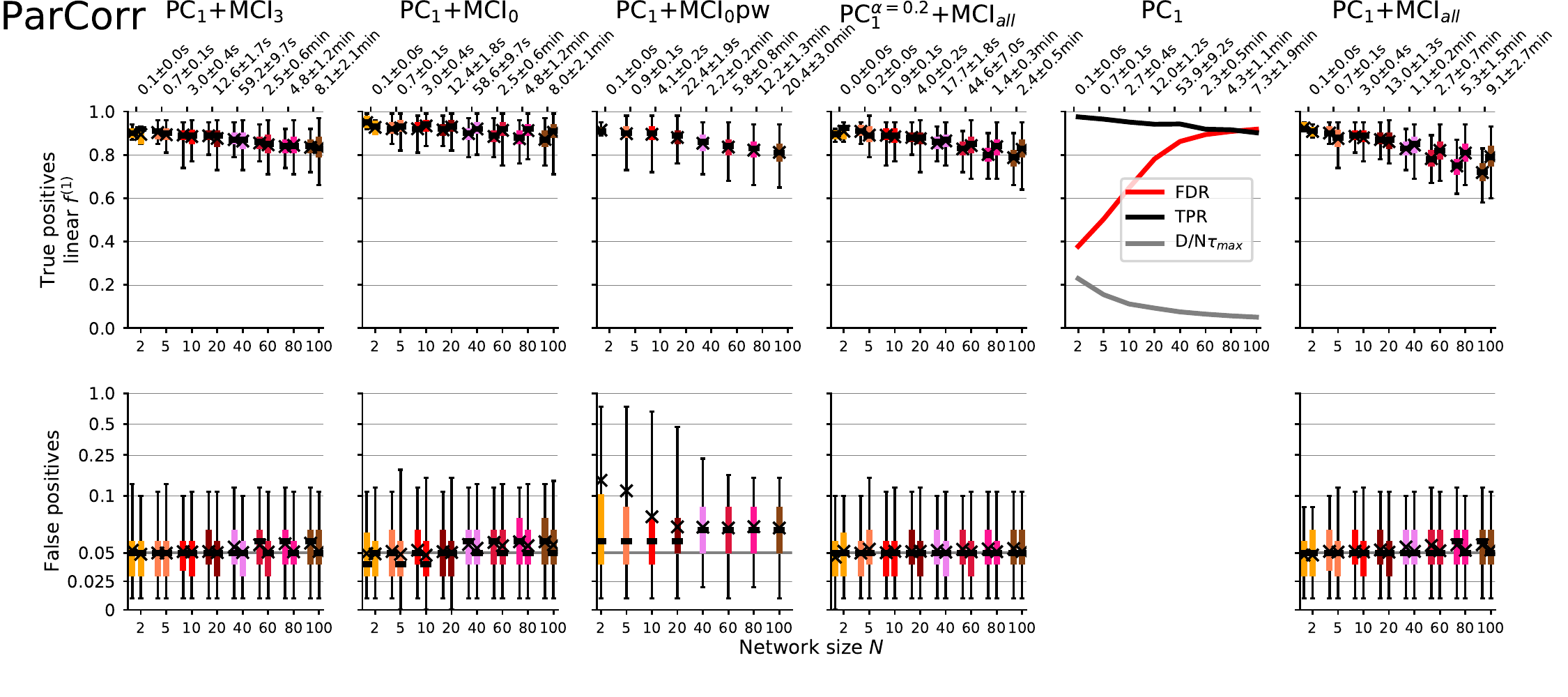}
\caption{
Numerical experiments for linear models with different numbers of variables $N$, number of links $L=N$, and time series length $T=300$. The detailed setup is listed in Tab.~\ref{tab:experiments} and Tab.~\ref{tab:methods} provides details on the evaluated methods.
}
\label{fig:algo_par_corr_allvsmit_SI_300}
\end{figure*}

\clearpage
\begin{figure*}[tbhp]
\centering
\includegraphics[width=1.\linewidth]{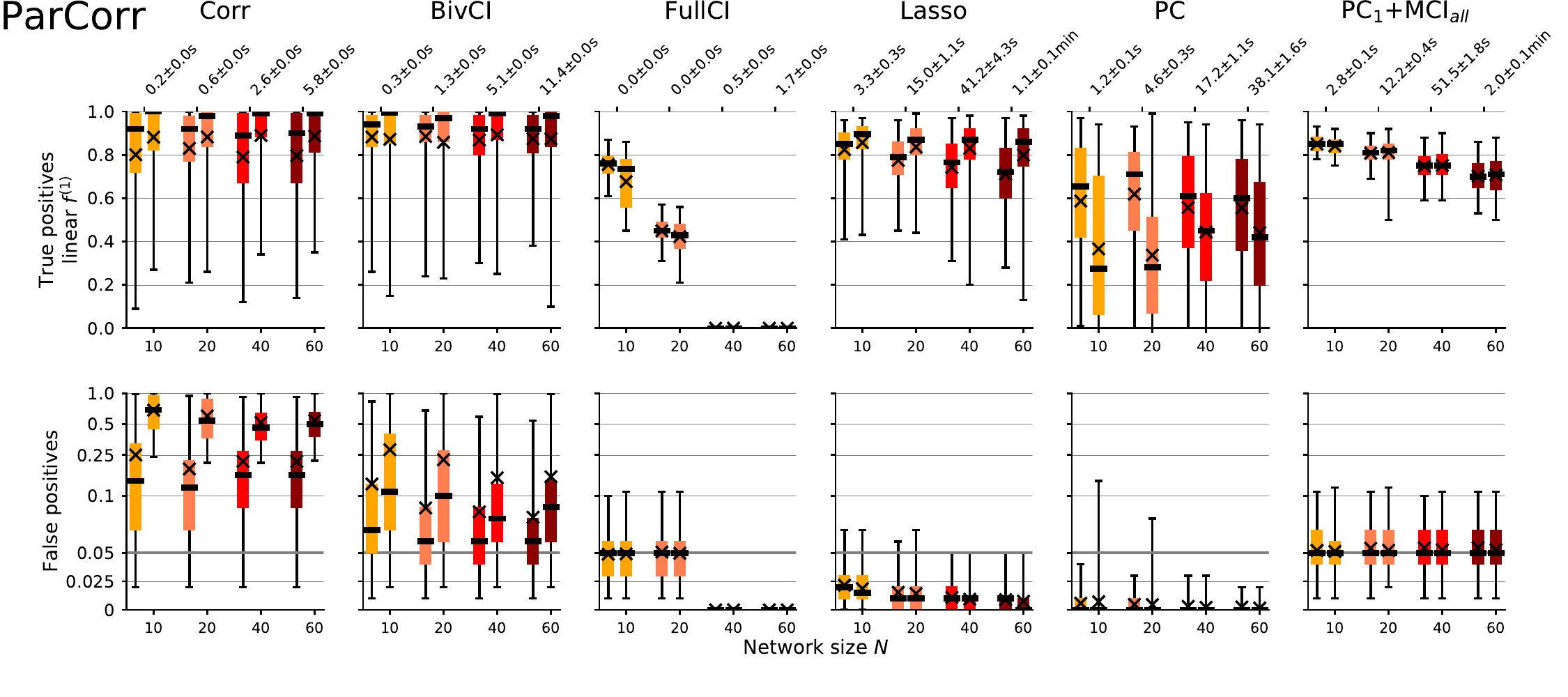}
\includegraphics[width=1.\linewidth]{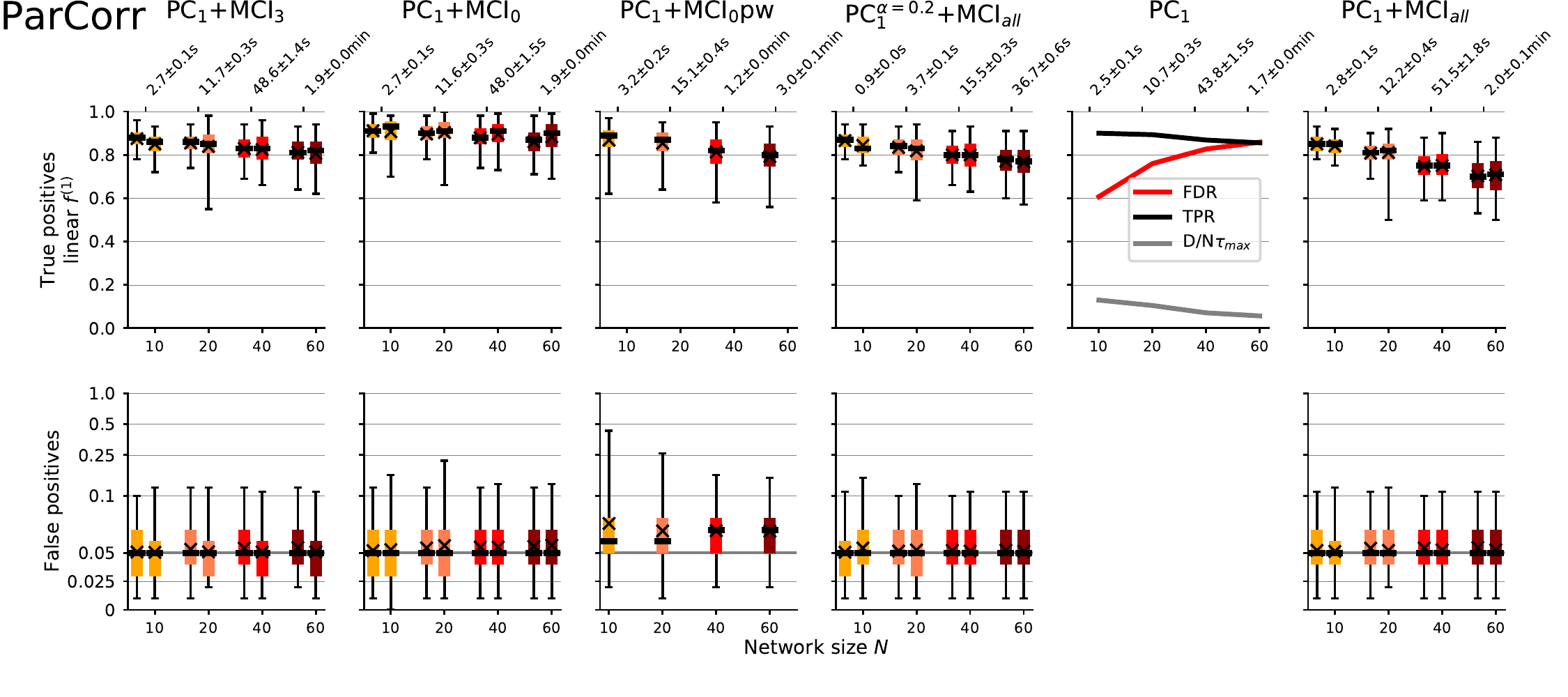}
\caption{
Numerical experiments for linear models with different numbers of variables $N$, number of links $L=2N$, and time series length $T=150$. The detailed setup is listed in Tab.~\ref{tab:experiments} and Tab.~\ref{tab:methods} provides details on the evaluated methods.
}
\label{fig:algo_results_degree_SI}
\end{figure*}

\clearpage
\begin{figure*}[tbhp]
\centering
\includegraphics[width=1.\linewidth]{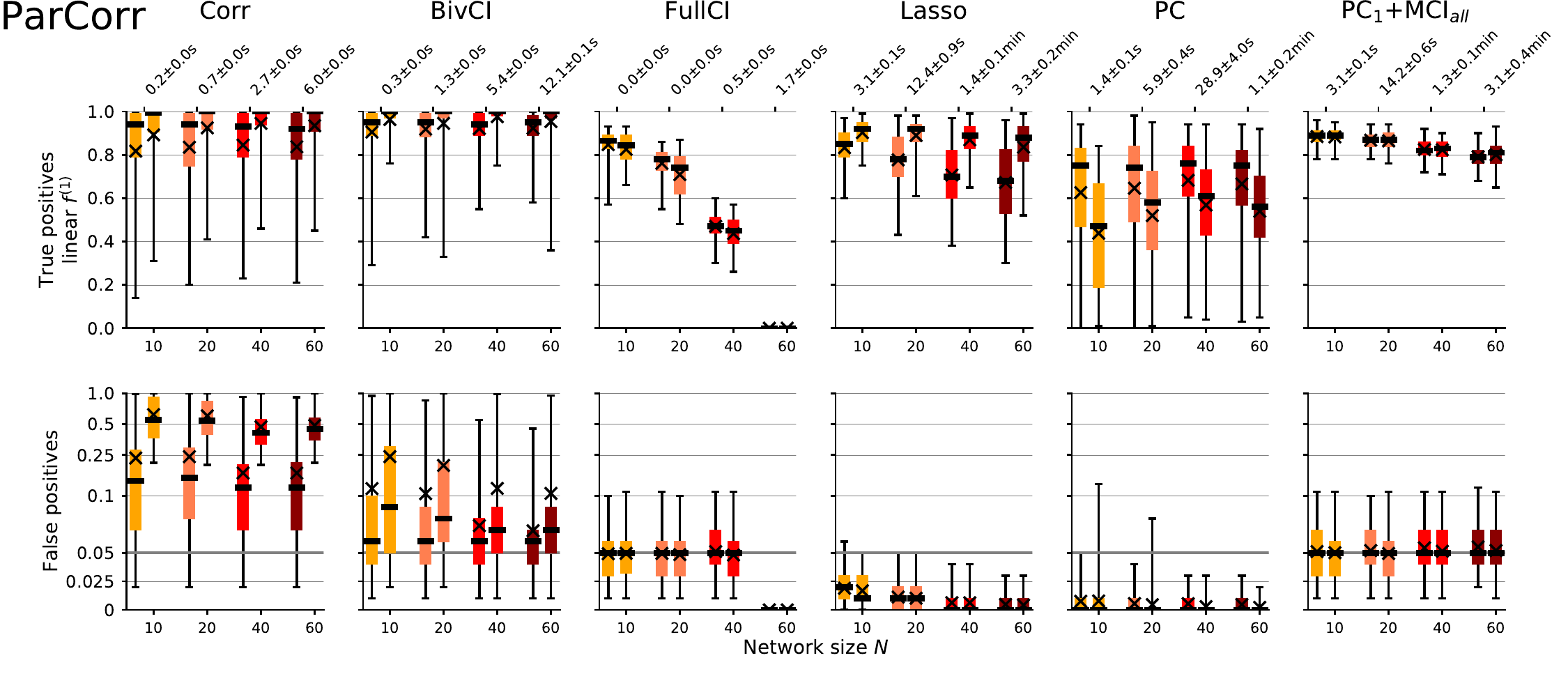}
\includegraphics[width=1.\linewidth]{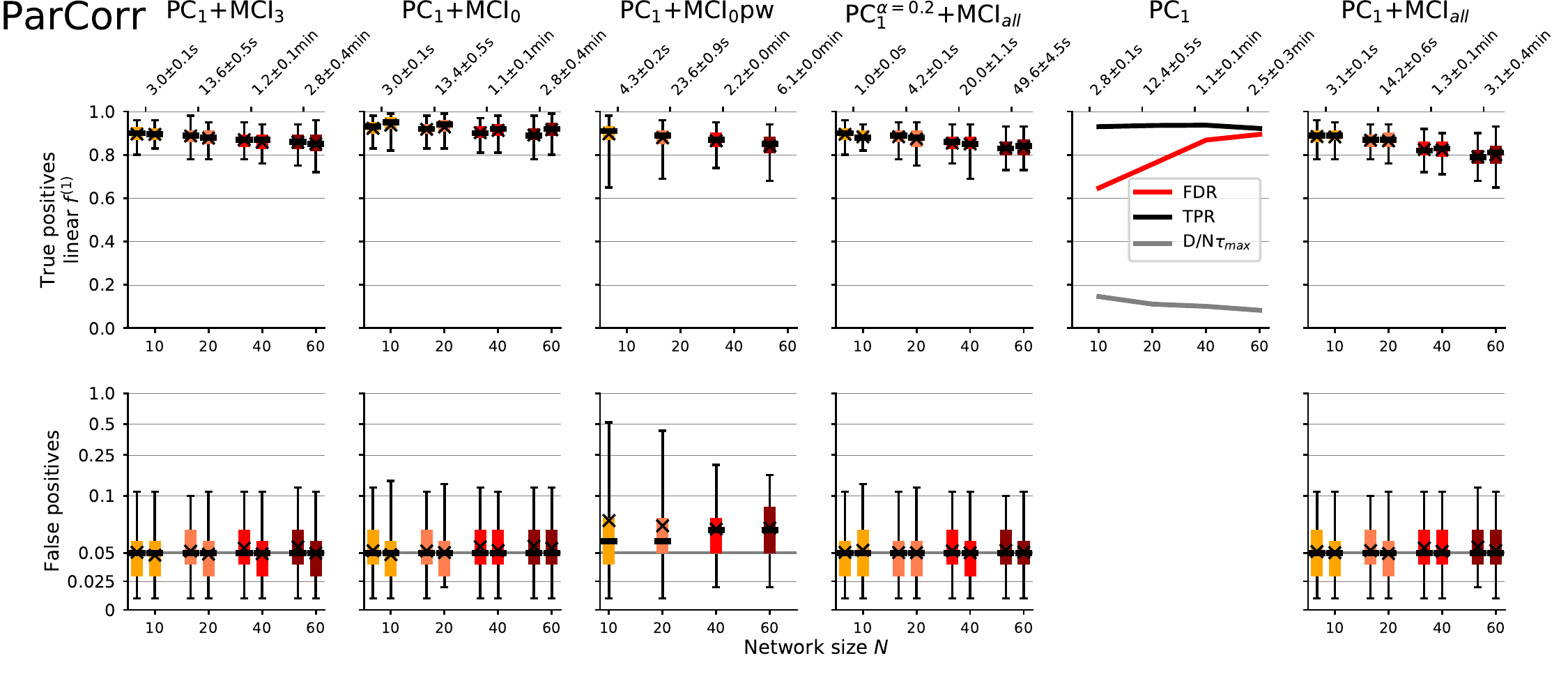}
\caption{
Numerical experiments for linear models with different numbers of variables $N$, number of links $L=2N$, and time series length $T=300$. The detailed setup is listed in Tab.~\ref{tab:experiments} and Tab.~\ref{tab:methods} provides details on the evaluated methods.
}
\label{fig:algo_results_degree_300_SI}
\end{figure*}

\clearpage
\begin{figure*}[tbhp]
\centering
\includegraphics[width=.32\linewidth]{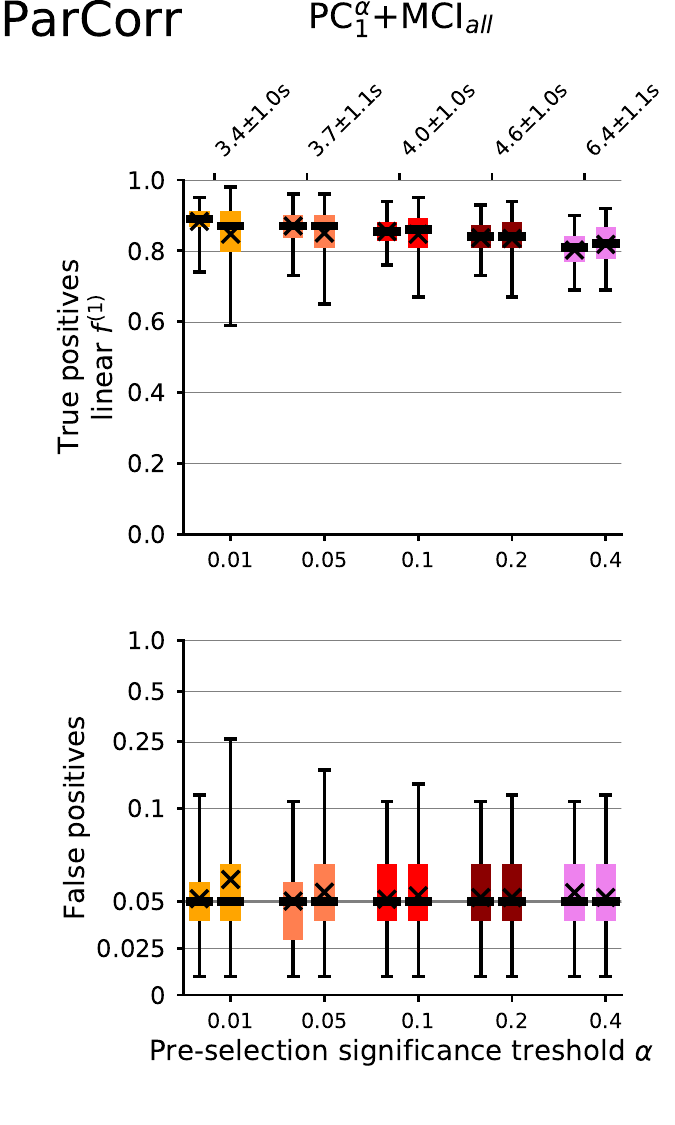}
\includegraphics[width=.32\linewidth]{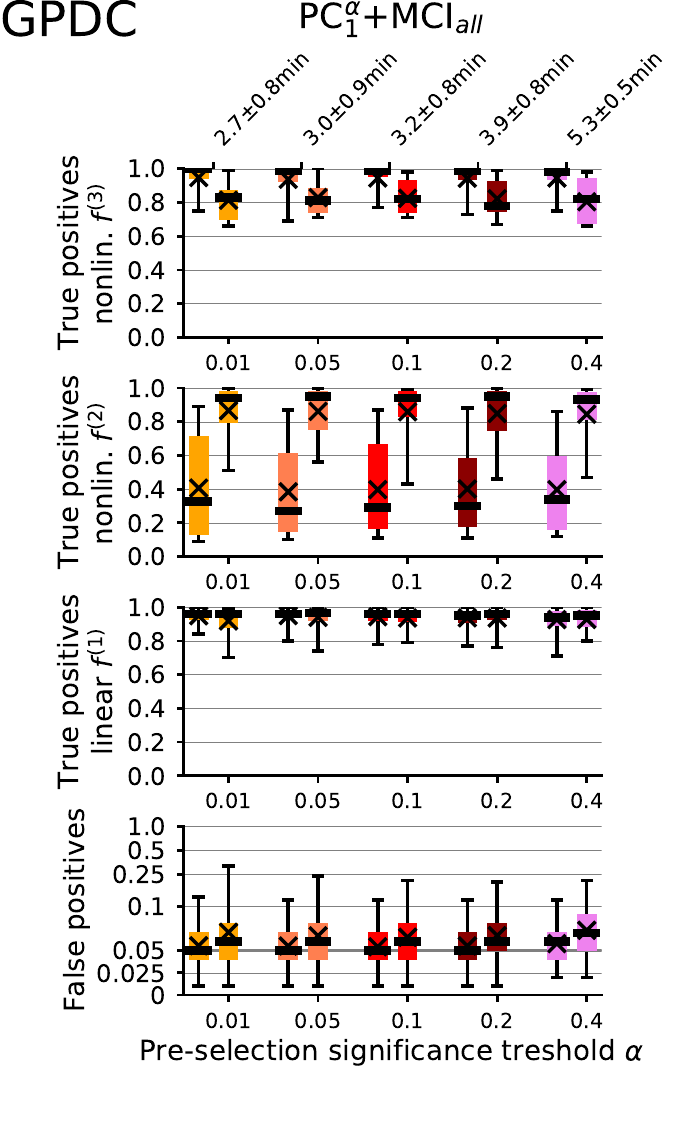}
\includegraphics[width=.32\linewidth]{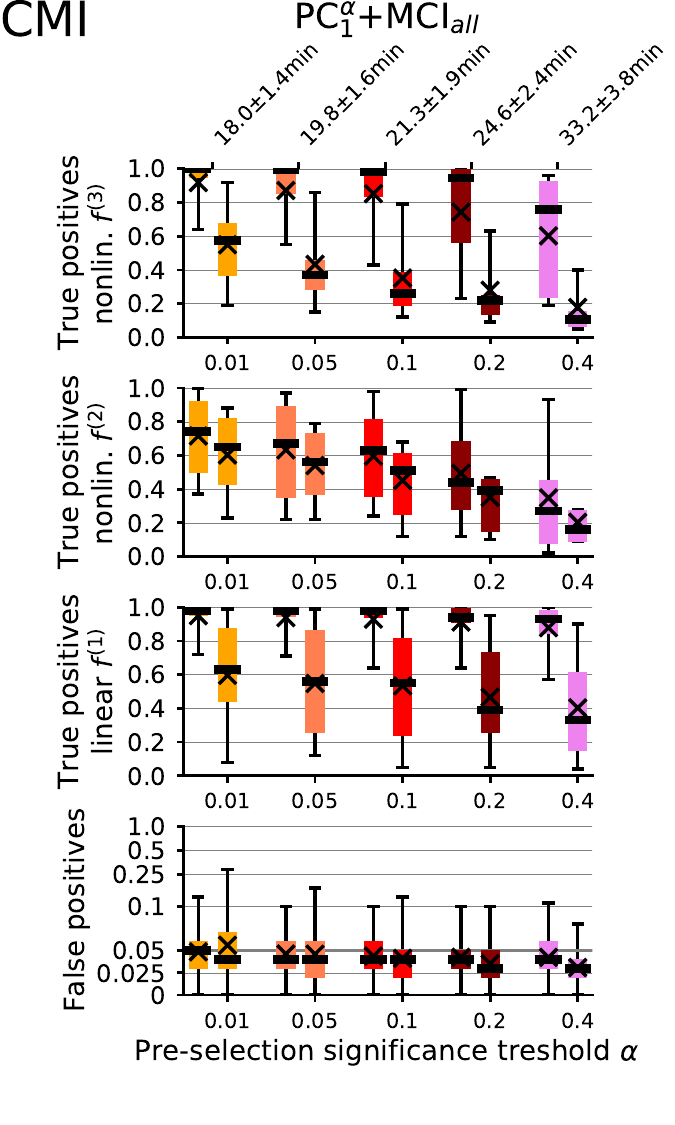}
\caption{ 
Influence of different PC thresholds $\alpha$ in the condition-selection Algorithm~\ref{algo:pcs1} for ParCorr (left), GPDC (center), and CMI (right). The full model setup is described in Sect.~\ref{sec:algo_model_description} and Tab.~\ref{tab:experiments}.
Note that for most numerical experiments for ParCorr we use an AIC-based optimization scheme to choose $\alpha$. See discussion in Sect.~\ref{sec:causal_discovery_SI}.
}
\label{fig:algo_results_pcthres_SI}
\end{figure*}

\clearpage
\begin{figure*}[tbhp]
\centering
\includegraphics[width=1.\linewidth]{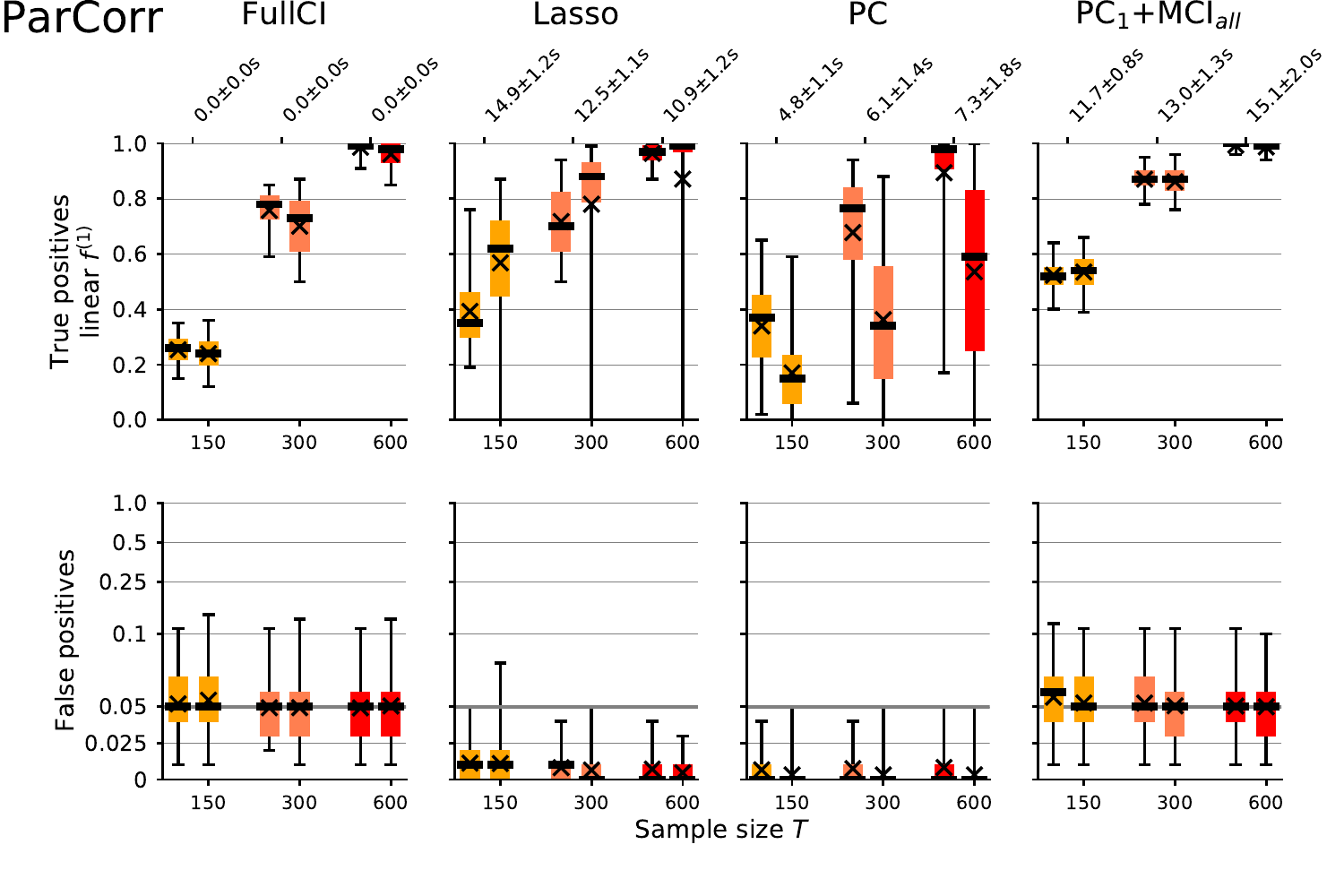}
\caption{
Numerical experiments for linear models with different time series length $T$, fixed numbers of variables $N=20$ and number of links $L=N$. The detailed setup is listed in Tab.~\ref{tab:experiments} and Tab.~\ref{tab:methods} provides details on the evaluated methods.
For ParCorr, power levels go up with larger samples as expected for all methods. Still, PCMCI always outperforms FullCI, Lasso, and PC. Even for the largest sample size Lasso still cannot detect some links.
The runtime does not linearly increase with sample size like the individual partial correlation tests because larger samples also lead to a faster convergence for PCMCI, and also for Lasso.
}
\label{fig:algo_results_samplesize_SI}
\end{figure*}

\clearpage
\begin{figure*}[tbhp]
\centering
\includegraphics[width=.9\linewidth]{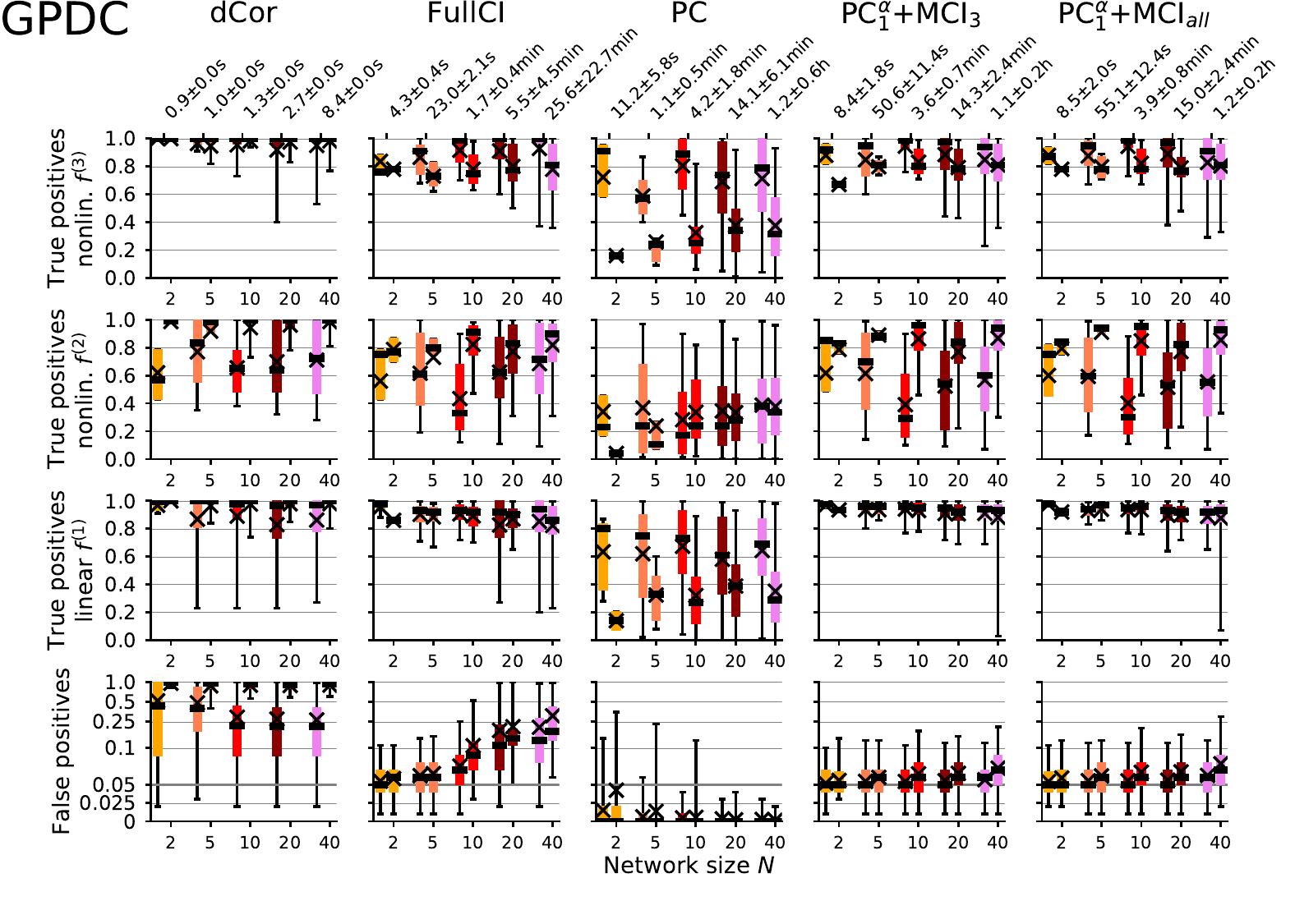}
\caption{Numerical experiments for nonlinear models with different numbers of variables $N$, number of links $L=N$, and time series length $T=250$. The detailed setup is listed in Tab.~\ref{tab:experiments} and Tab.~\ref{tab:methods} provides details on the evaluated methods.
}
\label{fig:algo_gpdc_allvsmit_SI}
\end{figure*}

\clearpage
\begin{figure*}[tbhp]
\centering
\includegraphics[width=.7\linewidth]{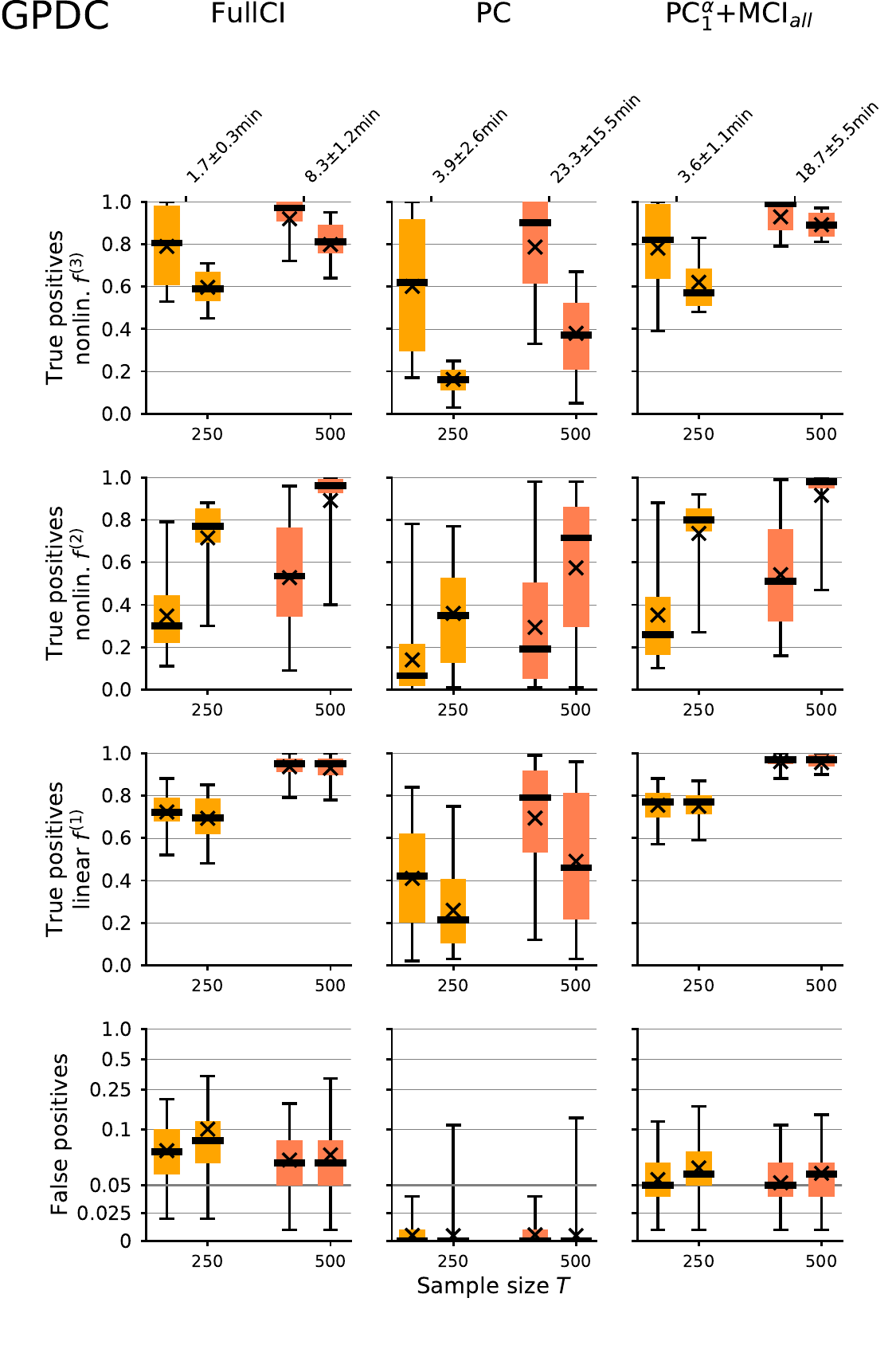}
\caption{Numerical experiments for nonlinear models with different time series length $T$, fixed numbers of variables $N=10$ and number of links $L=N$. The detailed setup is listed in Tab.~\ref{tab:experiments} and Tab.~\ref{tab:methods} provides details on the evaluated methods.
Similarly, for GPDC power goes up as expected. Additionally, FullCI better controls false positives for larger sample size. Gaussian process regression's runtime scales as $\sim T^3$ with sample size making GPDC not very suitable for large sample sizes. However, there are efficient approximation methods of GP that can help to speed up estimation\cite{Rasmussen2006}.
}
\label{fig:algo_results_samplesize_SI_gpdc}
\end{figure*}

\clearpage
\begin{figure*}[tbhp]
\centering
\includegraphics[width=.9\linewidth]{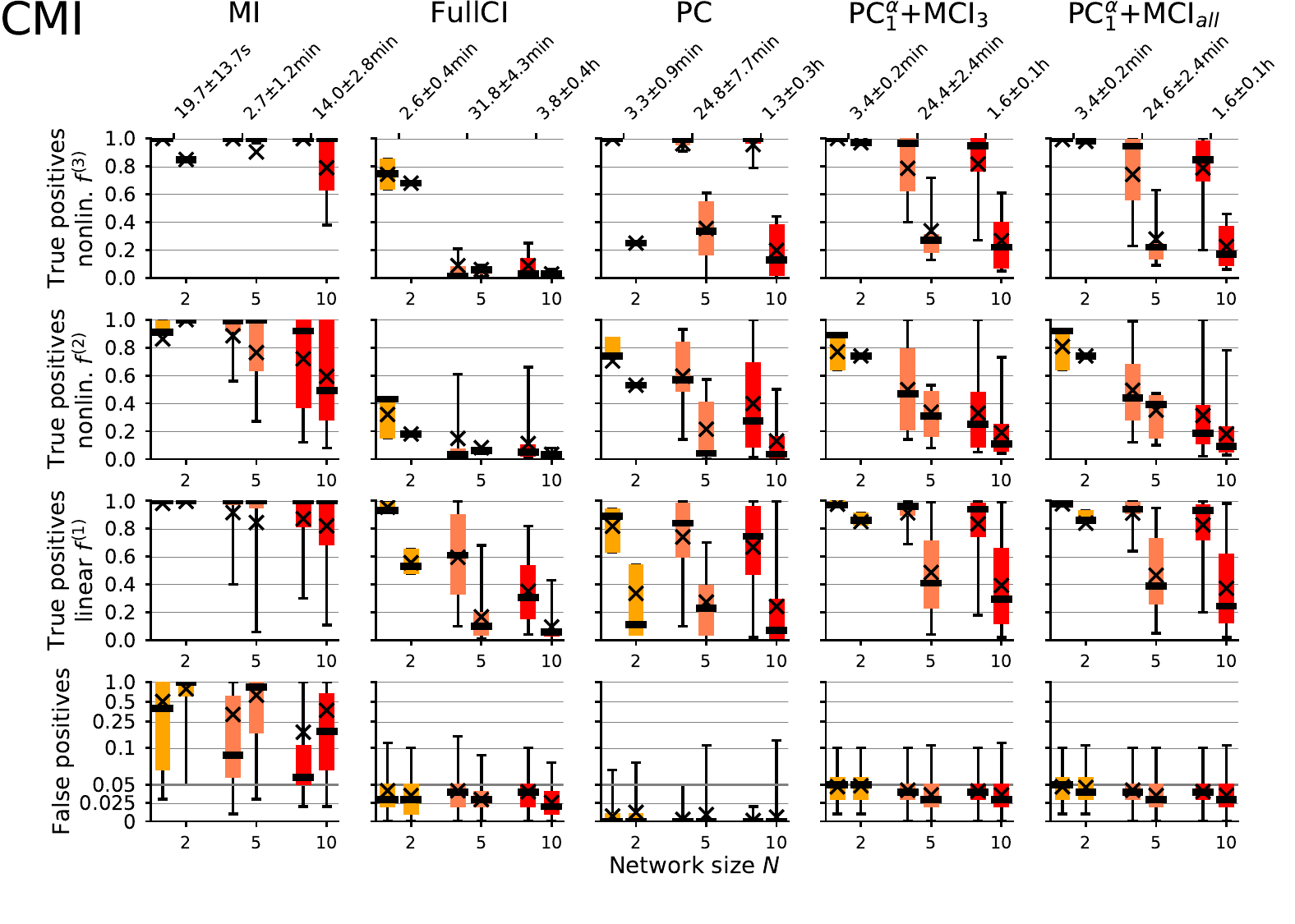}
\caption{Numerical experiments for nonlinear models with different numbers of variables $N$, number of links $L=N$, and time series length $T=500$. The detailed setup is listed in Tab.~\ref{tab:experiments} and Tab.~\ref{tab:methods} provides details on the evaluated methods
}
\label{fig:algo_cmiknn_allvsmit_SI}
\end{figure*}

\clearpage
\begin{figure*}[tbhp]
\centering
\includegraphics[width=.7\linewidth]{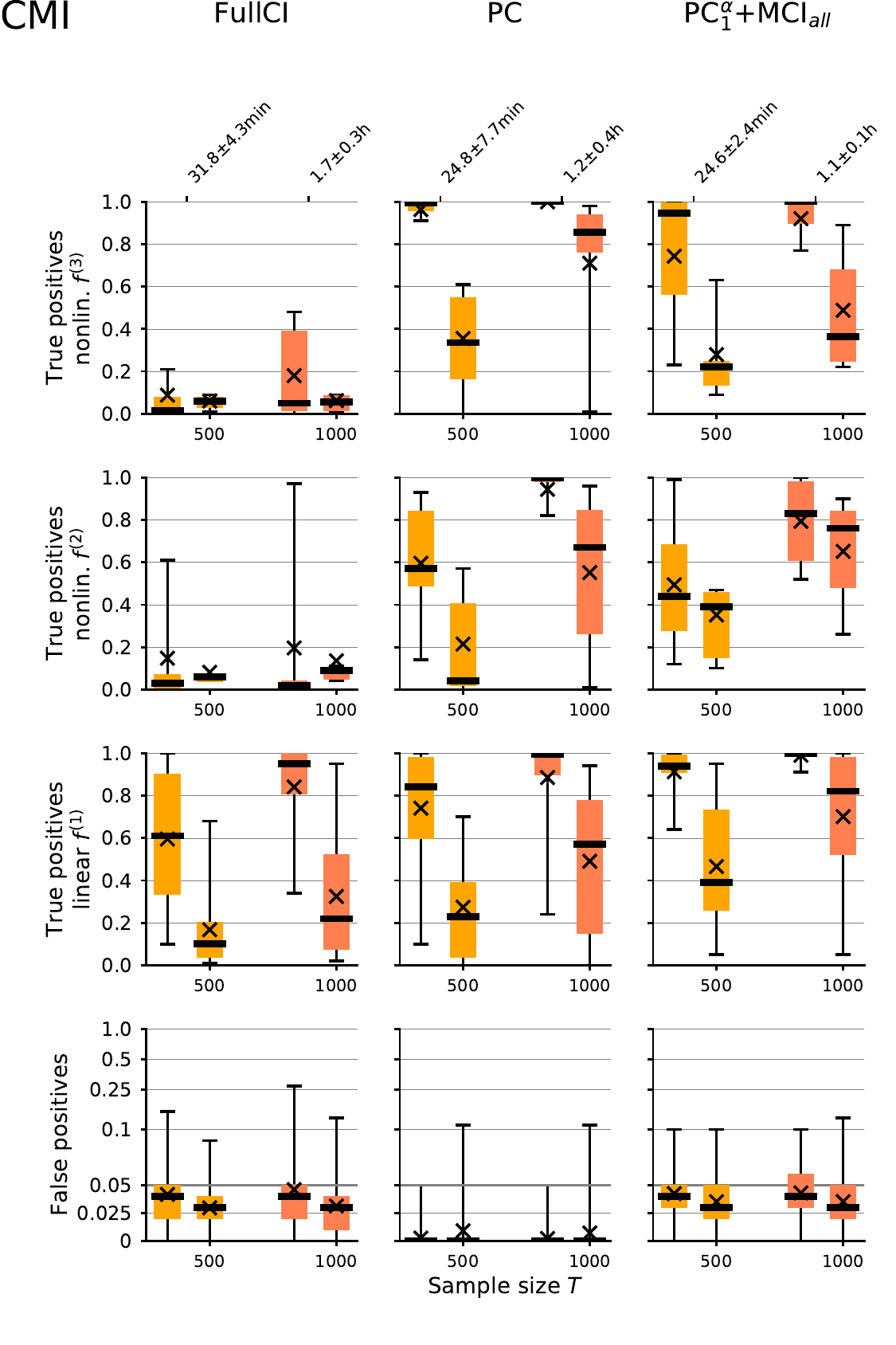}
\caption{Numerical experiments for nonlinear models with different time series lengths $T$, fixed numbers of variables $N=5$ and number of links $L=N$. The detailed setup is listed in Tab.~\ref{tab:experiments} and Tab.~\ref{tab:methods} provides details on the evaluated methods.
CMI has increasing power as expected. Here runtime does not increase quadratically like the individual CMI tests.
}
\label{fig:algo_results_samplesize_SI_cmi}
\end{figure*}


\clearpage
\begin{figure*}[tbhp]
\centering
\includegraphics[width=1.\linewidth]{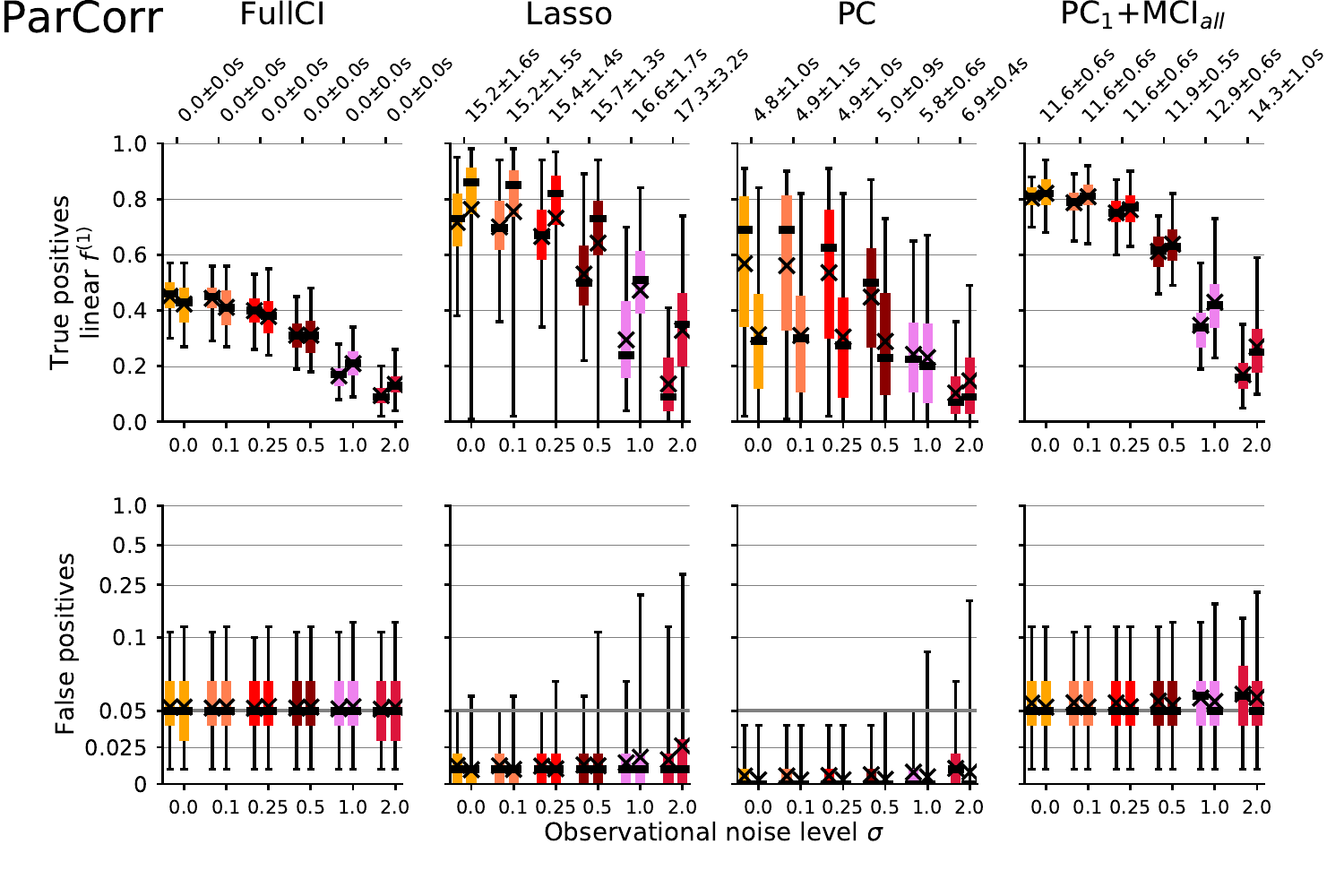}
\caption{
Numerical experiments for linear models with added observational noise, fixed time series length $T=150$, numbers of variables $N=20$ and number of links $L=N$. The detailed setup is listed in Tab.~\ref{tab:experiments} and Tab.~\ref{tab:methods} provides details on the evaluated methods.
Gaussian noise $\mathcal{N}\left(0, \sigma^2\right)$ with different standard deviations $\sigma$ was added to the data. Note that the original time series were generated with dynamical noise with standard deviation of one.
All methods display a similar sensitivity to observational noise with levels up to 25\% of the dynamical noise standard deviation having only minor effects. For levels of the same order as the dynamical noise we observe a stronger degradation with also the false positives not being well-controlled anymore since common drivers are essentially not well detected any longer. See ref.~\cite{Runge2018b} for a discussion on observational error.
}
\label{fig:algo_results_noise_SI}
\end{figure*}

\clearpage

{\small
\bibliographystyle{naturemag_noURL}

}


\end{document}